\documentclass[sigplan,screen]{acmart} %\settopmatter{authorsperrow=4}
\usepackage{booktabs}                           
\usepackage{subcaption} 
\usepackage{scalerel}
\usepackage{stackengine,wasysym}
\usepackage{microtype}

\newcommand{\BPp}[1]{\Big{(}#1\Big{)}}
\newcommand{\numgate}{\# \textrm{gates}}
\newcommand{\numline}{\# \textrm{lines}}
\newcommand{\numlayer}{\# \textrm{layers}}
\newcommand{\numqb}{\# \textrm{qb's}}

\newcommand{\entrideri}{\frac{d}{d\theta}}

\usepackage{subcaption}
\usepackage{qcircuit}
\usepackage{graphicx}
\usepackage[utf8]{inputenc}
\usepackage{CJKutf8}

\newcommand{\VE}{\qvec{v}}

\newcommand{\Ancil}{ \{A_{j,\qvec{v}}\}}
\newcommand{\Anc}{\{A\}}
\newcommand{\VA}{\qvec{v}\cup \Anc}

\usepackage{mathtools,mathrsfs}
\usepackage{proof}
\usepackage{ebproof,amssymb, scalerel}
\usepackage{enumitem}
\usepackage[linesnumbered]{algorithm2e}

\newcommand{\nc}{\newcommand}
 
\nc{\ketbra}[2]{|#1\rangle\!\langle#2|}
\nc{\braket}[2]{\langle#1|#2\rangle}
\nc{\proj}[1]{| #1\rangle\!\langle #1 |}

\newcommand{\ObsSemMacro}[3]{\sem{(#1,#3)\rightarrow #2}}
\newcommand{\NumNonAbort}[1]{|\# #1|}
\newcommand{\ResCount}[2]{\mathrm{OC}_{#1}(#2)}
\newtheorem{theorem}{Theorem} 
\newtheorem{lemma}[theorem]{Lemma}

\newtheorem{prop}[theorem]{Proposition}

\newtheorem{example}{Example}

\newtheorem{definition}[theorem]{Definition}
\newtheorem{defn}[theorem]{Definition}

\newcommand{\weirdS}{\mathcal{D(\mathcal{H}}_{\overline{v}})}

\newcommand{\udd}[1]{\underline{#1}}

\newcommand{\weirdSwhat}[1]{\mathcal{D}(\mathcal{H}_{#1})}

\newcommand{\QNN}{\texttt{QNN}}
\newcommand{\QAOA}{\texttt{QAOA}}
\newcommand{\VQE}{\texttt{VQE}}

\def\msquare{
\boldgreek{+}}
\newcommand{\Macroadditive}{additive }

\newcommand{\MacroAdditive}{Additive }

\numberwithin{equation}{section}
\numberwithin{theorem}{section}
\numberwithin{example}{section}
\newcommand{\eq}[1]{\hyperref[eqn::#1]{(\ref*{eqn::#1})}}
\renewcommand{\sec}[1]{\hyperref[sec:#1]{Section~\ref*{sec:#1}}}
\newcommand{\cor}[1]{\hyperref[cor:#1]{Corollary~\ref*{cor:#1}}}
\newcommand{\itm}[1]{\hyperref[itm:#1]{\ref*{itm:#1}}}
\newcommand{\fig}[1]{\hyperref[fig:#1]{Figure~\ref*{fig:#1}}}

\newcommand{\calH}{\mathcal{H}}

\newcommand{\nm}[1]{\lVert #1\rVert}
\newcommand{\sem}[1]{[\![ #1 ]\!]}
\newcommand{\bra}[1]{\langle #1 \vert}
\newcommand{\ket}[1]{\vert #1 \rangle}
\newcommand{\tr}[0]{\mathrm{tr}}

\newcommand{\cskip}[0]{{\mathbf{skip}}}
\newcommand{\cabort}[0]{{\mathbf{abort}}}

\newcommand{\ccase}[0]{{\mathbf{case}}}
\newcommand{\rwhile}{\mathbf{while}}
\newcommand{\rwhileT}{\mathbf{while^{(T)}}}
\newcommand{\rwhileTminus}{\mathbf{while^{(T-1)}}}

\newcommand{\qif}[1]{\mathbf{case}~#1~\mathbf{end}}

\newcommand{\qifStandardPara}{\qif{M[\qvec{q}]=\overline{m\to P_{m}(\vth)}}}

\newcommand{\uqif}[1]{\udd{\mathbf{case}}~#1~\mathbf{end}}

\newcommand{\uqifParaStandard}{\uqif{M[\qvec{q}]=\overline{m\to \udd{P_{m}(\vth)}}}}
\newcommand{\uqifParaStandardS}{\uqif{M[\qvec{q}]=\overline{m\to \udd{S_{m}(\vth)}}}}

\newcommand{\qwhileT}[2]{\mathbf{while}^{(T)}~#1~\mathbf{do}~#2~\mathbf{done}}
\newcommand{\qwhileO}[2]{\mathbf{while}^{(1)}~#1~\mathbf{do}~#2~\mathbf{done}}

\newcommand{\qwhileTgT}[2]{\mathbf{while}^{(T\geq 2)}~#1~\mathbf{do}~#2~\mathbf{done}}

\newcommand{\qwhileTParaStandard}{\mathbf{while}^{(T)}~M[\overline{q}]=1~\mathbf{do}~P_1(\vth)~\mathbf{done}}

\newcommand{\uqwhileT}[2]{\udd{\mathbf{while}^{(T)}}~#1~\mathbf{do}~#2~\mathbf{done}}

\newcommand{\uqwhileTParaStandard}{\udd{\mathbf{while}^{(T)}}~M[\overline{q}]=1~\mathbf{do}~\udd{P_1(\vth)}~\mathbf{done}}
\newcommand{\uqwhileTParaStandardS}{\udd{\mathbf{while}^{(T)}}~M[\overline{q}]=1~\mathbf{do}~\udd{S_1(\vth)}~\mathbf{done}}

\newcommand{\vth}{\boldgreek{\theta}}

\newcommand{\weirdPV}{\textbf{q-while}^{(T)}_{\overline{v}}}

\newcommand{\calU}{\mathcal{U}}
\newcommand{\calI}{\mathcal{I}}

\newcommand{\weMV}{\mathcal{O}_{\overline{v}}}

\newcommand{\augweirdPV}{\textbf{q-while}^{(T)}_{\overline{v}}(\vth)}

\newcommand{\detaugweirdCPV}{\textbf{add-q-while}^{(T)}_{\overline{v}}(\vth)}
\newcommand{\detaugweirdCPVA}{\textbf{add-q-while}^{(T)}_{\overline{v}\cup \Ancil} (\vth)}
\newcommand{\detaugweirdCPVa}{\textbf{add-q-while}^{(T)}_{\overline{v}\cup \Anc} (\vth)}

\newcommand{\augweirdCPV}{\textbf{q-while}^{(T)}_{\overline{v}}(\vth)}
\newcommand{\augweirdCPVA}{\textbf{q-while}^{(T)}_{\overline{v}\cup \Ancil} (\vth)}
\newcommand{\augweirdCPVa}{\textbf{q-while}^{(T)}_{\overline{v}\cup \Anc} (\vth)}

\newcommand{\peeeep}{\partial}

\newcommand{\parthetaj}[1]{\frac{\peeeep }{\peeeep \theta_j}#1}

\newcommand{\parthetajgenNotRelation}[1]{\frac{\partial}{\partial\theta_j}(#1)}
\newcommand{\parthetagenNotRelation}[1]{\frac{\partial}{\partial\theta}(#1)}
\newcommand{\parthetajgenNotRelationNopar}[0]{\frac{\partial}{\partial\theta_j}}
\newcommand{\parthetagenNotRelationNopar}[0]{\frac{\partial}{\partial\theta}}

\newcommand{\qvec}[1]{\overline{#1}}

\newcommand{\fparthetaj}[1]{\frac{d}{d \theta_j}#1}

\newcommand{\detqchs}{\Macroadditive parameterized quantum bounded  while-programs}

\newcommand{\detQCHs}{\MacroAdditive Parameterized Quantum Bounded While-Programs}

\newcommand{\compl}{\mathtt{Compile}}

\newcommand{\QCHs}{Parameterized  Quantum Bounded While-Programs}

\usepackage[T1]{fontenc}
\usepackage{pdfrender}

\newcommand*{\boldgreek}[1]{%
  \textpdfrender{%
   TextRenderingMode=FillStroke,%
    LineWidth=.6pt,%
  }{#1}%
}

\newcommand{\dagg}{^\dagger}
\newcommand{\modify}[1]{{\color{blue}#1}}

\newcommand{\colorful}[2]{{\color{#1}#2}}

\newcommand{\RR}{\mathbb{R}}

\newcommand{\ZZ}{\mathbb{Z}}

\newcommand{\Facebook}[1]{\mathtt{FB}(#1)}
\newcommand{\FacebookBPp}[1]{\mathtt{FB}\BPp{#1}}

\newcommand{\lag}{\langle}
\newcommand{\rag}{\rangle}

\newcommand{\A}{ \N \bf{eg.} \rm}

\newcommand{\D}[0]{\mathcal{D}}
\newcommand{\E}[0]{\mathcal{E}}
\renewcommand{\H}[0]{\mathcal{H}}
\newcommand{\F}[0]{\mathcal{F}}

\newcommand{\StepsTo}[3]{\langle #1,~\rho \rangle \rightarrow \langle #2,~#3 \rangle}

\newif\ifsubmit\submitfalse %%%%% produces appendix.
\newif\iftechrep\techreptrue  %%%%%%%%%%% says "Appendix"
\iftechrep
\newcommand{\appdName}{\textrm{Appendix}}         
\else
\newcommand{\appdName}{\textrm{the full version~\cite{zhu20dqplextended}}} 
\fi

\usepackage{xcolor}
\ifsubmit
\newcommand{\khh}[1]{}
\newcommand{\szhu}[1]{}
\newcommand{\shh}[1]{}
\newcommand{\xw}[1]{}
\newcommand{\mwh}[1]{}
{}
\newcommand{\mszhu}[1]{}

\else
\newcommand{\khh}[1]{\textcolor{blue}{\bf Kesha: #1}}
\newcommand{\szhu}[1]{\textcolor{cyan}{\bf Shaopeng: #1}}
\newcommand{\mszhu}[1]{\small{\szhu{#1}}}
\newcommand{\shh}[1]{\textcolor{magenta}{\bf Shih-Han: #1}}
\newcommand{\xw}[1]{\textcolor{purple}{\bf XW: #1}}
\newcommand{\mwh}[1]{\textcolor{green}{\bf Mike: #1}}

\fi

\def \quwhile {quantum \textbf{while}}
\def \A {\mathcal{A}}

\begin{document}

%%% The following is specific to PLDI '20 and the paper
%%% 'On the Principles of Differentiable Quantum Programming Languages'
%%% by Shaopeng Zhu, Shih-Han Hung, Shouvanik Chakrabarti, and Xiaodi Wu.
%%%
\setcopyright{acmlicensed}
\acmPrice{}
\acmDOI{10.1145/3385412.3386011}
\acmYear{2020}
\copyrightyear{2020}
\acmSubmissionID{pldi20main-p469-p}
\acmISBN{978-1-4503-7613-6/20/06}
\acmConference[PLDI '20]{Proceedings of the 41st ACM SIGPLAN International Conference on Programming Language Design and Implementation}{June 15--20, 2020}{London, UK}
\acmBooktitle{Proceedings of the 41st ACM SIGPLAN International Conference on Programming Language Design and Implementation (PLDI '20), June 15--20, 2020, London, UK}

%% Title information
%%%\title[Short Title]{Full Title} 
\title{On the Principles of Differentiable Quantum Programming Languages} 
\titlenote{This work was partially funded by the U.S. Department of Energy, Office of Science, Office of Advanced Scientific Computing Research, Quantum Testbed Pathfinder Program under Award Number DE-SC0019040, Quantum Algorithms Team program, the U.S. Air Force Office of Scientific Research MURI grant FA9550-16-1-0082, and the U.S. National Science Foundation grant CCF-1755800, CCF-1816695, and CCF-1942837(CAREER). }

\author{Shaopeng Zhu, \, Shih-Han Hung, \, Shouvanik Chakrabarti, \,Xiaodi Wu}
\affiliation{
 \institution{University of Maryland, College Park, USA}
 }

 \def\titlerunning{On the Principles of Differentiable Quantum Programming Languages}
 \def\authorrunning{S. Zhu, S. Hung, S. Chakrabarti \& X. Wu}
\renewcommand{\shortauthors}{S. Zhu, S. Hung, S. Chakrabarti \& X. Wu}

\begin{abstract}
Variational Quantum Circuits (VQCs), or the so-called quantum neural-networks, are predicted to be one of the most important near-term quantum applications, not only because of their similar promises as classical neural-networks, but also because of their feasibility on near-term noisy intermediate-size quantum (NISQ) machines.
The need for gradient information in the training procedure of VQC applications has stimulated the development of auto-differentiation techniques for quantum circuits.
We propose the first formalization of this technique, not only in the context of quantum circuits but also for imperative quantum programs (e.g., with controls), inspired by the success of differentiable programming languages in classical machine learning. 
In particular, we overcome a few unique difficulties caused by exotic quantum features (such as quantum no-cloning) and provide a rigorous formulation of differentiation applied to bounded-loop imperative quantum programs, its code-transformation rules, as well as a sound logic 
to reason about their correctness. 
Moreover, we have implemented our code transformation in OCaml and demonstrated the resource-efficiency of our scheme both analytically and empirically. 
We also conduct a case study of training a VQC instance with controls, which shows the advantage of our scheme over existing auto-differentiation for quantum circuits without controls.
\end{abstract}




\maketitle

\section{Introduction} \label{sec::intro}

\vspace{1mm} \noindent \textbf{Background.} Recent years have witnessed the rapid development of quantum computing, 
with practical advances coming from both
research and industry. Quantum programming is one topic that has been actively investigated. Early work on language design~\citep{Om03,
  SZ00, Sabry-Haskel, Se04, AG05} has been followed up recently by several implementations of these languages,
including Quipper \citep{Green2013}, Scaffold \citep{Sca12}, 
LIQUi$\vert\rangle$ \citep{WS2014}, Q\#~\citep{Svore:2018}, and QWIRE \citep{PRZ2017}.
Extensions of program logics have also
been proposed for verification of quantum programs \citep{BJ04, CHADHA200619,
  Baltag2011, Feng:2007, Kaku09, Yin11, YYW17, hung19}. See also 
surveys \citep{Selinger04, Gay:2006, Ying16}.

With the availability of prototypes of quantum machines, especially the recent establishment of quantum supremacy~\cite{google-supremacy}, the research of quantum computing has entered a new stage where near-term Noisy Intermediate-Scale
Quantum (NISQ) computers~\cite{Preskill2018NISQ}, e.g., the 53-qubit quantum machines from Google~\cite{google-supremacy} and IBM~\cite{IBM-53}, become the important platform for demonstrating quantum applications. 
Variational quantum circuits (VQCs)~\cite{NC-VQE, EH18, QAOA}, or the so-called \emph{quantum neural networks}, are predicted to be one of the most important applications on NISQ machines. 
It is not only because VQCs bear a lot of similar promises like classical neural networks as well as potential quantum speed-ups from the perspective of machine learning (e.g., see the survey~\cite{biamonte2017quantum}), but also because VQC is, if not the only, one of the few candidates that can be implemented on NISQ machines. Because of this, a lot of study has already been devoted to the design, analysis, and small-scale implementation of VQCs (e.g., see the survey~\cite{2019arXiv190607682B}). 

Typical VQC applications replace classical neural networks, which are just parameterized classical circuits, by quantum circuits with \emph{classically parameterized unitary gates}. 
Namely, one will have a "quantum" mapping from input to output replacing classical mapping in machine learning applications. 
An important component of these applications is a training procedure which optimizes a \emph{loss function} 
that now depends on the \emph{read-outs} and the \emph{parameters} of VQCs. 

Gradient-based approaches are widely used in the training procedure. However, computing these gradients of loss functions from quantum circuits has a similar complexity of simulating quantum circuits, which is infeasible for classical computation.
Thus, the ability of evaluating these "quantum" gradients efficiently by  quantum computation is critical for the scalability of VQC applications.  

Fortunately, analytical formulas of gradients used in VQCs have been studied by \cite{GM17,PhysRevA.98.062324,EH18, MAKN18, SBVCK18}. In particular, \citet{SBVCK18} proposed the so-called \emph{phase-shift rule} that uses two quantum circuits to compute the partial derivative respective to one parameter for quantum circuits.  
One of the very successful tools for quantum machine learning, called PennyLane~\cite{bergholm2018pennylane}, implemented the phase-shift rule to achieve \emph{auto differentiation} (AD) for the read-outs of quantum circuits. It also integrated automatic differentiation from established machine learning libraries such as TensorFlow or PyTorch for any additional classical calculation in the training procedure.  
However, none of these studies was conducted from the perspective of programming languages and no rigorous foundation or principles have been formalized.  

\vspace{1mm} \noindent \textbf{Motivations.} 
An important motivation of this paper is to provide a \emph{rigorous formalization} of the auto-differentiation technique applied to quantum circuits. 
In particular, we will provide a formal formulation of quantum programs, their semantics, and  the meaning of differentiation of them. We will also study the code-transformation rules for auto-differentiation and prove their correctness. 

As we will highlight below, research on the formalization will encounter many new challenges that have not been considered or addressed by existing results~\cite{GM17,PhysRevA.98.062324,EH18, MAKN18, SBVCK18}. Consider one of the basic requirements, e.g., \emph{compositionality}. 
As we will show, differentiating the composition of quantum programs will necessarily involve running multiple quantum programs on copies of initial quantum states. 
How to represent the collection of quantum programs succinctly and also bound the number of required copies is a totally new question. Among our techniques to address this question, we also need to change the previously proposed construct, e.g., the phase-shift rule~\cite{SBVCK18}, to something different. 

Moreover, we want to go beyond the restriction of quantum circuits. 
Our inspiration comes from classical machine learning examples that demonstrate the advantage of neural-networks with program features (e.g. controls) over the plain ones (e.g., classical circuits), e.g.~\cite{deep-mind, Grefenstette:2015:LTU:2969442.2969444}, 
which is also the major motivation of promoting the 
the paradigm shift from deep learning toward \emph{differentiable programming}. 

Augmenting VQCs with controls, at least for simple ones, is not only feasible on NISQ machines, but also a logical step for the study of their applications in machine learning. 
Therefore, we are inspired to investigate the \textbf{\emph{principles of differentiable quantum imperative languages}} beyond circuits. Indeed, we conduct one such case study in Section~\ref{sec::case}.

\vspace{1mm} \noindent \textbf{Research challenges \& Solutions. } We will rely on a few notations that should be self-explanatory. Please refer to a detailed preliminary on quantum information in Section~\ref{sec:prelim}. 
Let us start with a simple classical program 
\begin{equation}
    \mathrm{MUL} \equiv v_3 =v_1 \times v_2, 
\end{equation}
where $v_3$ is the product of $v_1$ and $v_2$. Consider the differentiation with respect to $\theta$, we have
\begin{align}
\parthetagenNotRelation{\mathrm{MUL}}  \equiv & \quad v_3 = v_1 \times v_2;  \\
  & \quad \dot{v}_3 = \dot{v}_1 \times v_2 + v_1 \times \dot{v}_2, \label{eqn:product}
\end{align}
where $\mathrm{MUL}$ keeps track of variables $v_1, v_2, v_3$ and their derivatives $\dot{v}_1, \dot{v}_2, \dot{v}_3$ at the same time. One simple yet important observation is that 
classical variables $v_1, v_2, v_3$ are real-valued % numbers
and can be naturally differentiated. %while quantum variables are quantum states represented by matrices.

Given that quantum states are represented by matrices, what are the natural quantities to differentiate in the quantum setting? One natural choice from the principles of quantum mechanics is the (classical) read-outs of quantum systems through measurements, which we formulate as the \emph{observable semantics} of quantum programs. 
This natural choice also serves the purpose of gradient computation of loss functions in quantum machine learning, which are typically defined in terms of these read-outs. We directly model the parameterization of quantum programs after VQCs, i.e., each unitary gate becomes \emph{classically parameterized}. 
\footnote{The above modeling of quantum programs is very different from classical ones. It is unclear whether any reasonable analogue of classical chain-rule and forward/backward mode can exist within quantum programs.}

To model the meaning of one quantum program computing the derivative of another, we define the \emph{differential semantics} of programs. 
There is a subtle quantum-unique design choice. The observable semantics of any quantum program will depend on the observable and its input state. 
Thus, any program computing its derivative could potentially depend on these two extra factors. 
We find out this potential dependence is undesirable and propose the strongest possible definition: i.e., one derivative computing program should work for \emph{any pair} of observables and input states. 
We demonstrate that this strong requirement is not only achievable but also critical for the composition of auto differentiation. 

We are ready to describe the technical challenges for the compositionality. Consider the following quantum program: 
\begin{equation}
    \mathrm{QMUL} \equiv U_1(\theta); U_2 (\theta), 
\end{equation}
which performs $U_1(\theta)$ and $U_2(\theta)$ gates sequentially. Note that gate application is matrix multiplication in 
the quantum setting. 
Roughly speaking, if the product rule of differentiation (as exhibited in~\eqref{eqn:product}) remains in the quantum setting, at least symbolically, then one should expect $\parthetagenNotRelation{\mathrm{QMUL}}$ contains
\begin{equation} \label{eqn:qmul}
    \parthetagenNotRelation{U_1(\theta)}; U_2 (\theta) \text{ and } U_1(\theta); \parthetagenNotRelation{U_2 (\theta)}
\end{equation}
two different parts as sub-programs similarly in~\eqref{eqn:product}. 

However, we cannot run $\parthetagenNotRelation{U_1(\theta)}; U_2$ and $U_1(\theta); \parthetagenNotRelation{U_2 (\theta)}$ together due to the quantum \emph{no-cloning} theorem~\cite{no-cloning}.
This is simply because they share the same initial state and we cannot clone two copies of it. % and maintain the correct correlation among different parts. 
Note that this is not an issue classically as we can store all $v_i, \dot{v}_i$ at the same time as in~\eqref{eqn:product}.
As a result, quantum differentiation needs to run \emph{multiple} (sub-)programs on \emph{multiple} copies of the initial state. 

This change poses a unique challenge for differentiation of quantum composition: (1) we hope to have a simple scheme of code transformation, ideally close-to-classical, for intuition and easy implementation of the compiler, whereas it needs to express correctly the collection of quantum programs during code transformation; 
(2) for the purpose of efficiency, we also want to reasonably bound the number of required copies of the initial states, which roughly refers to the number of different quantum programs in the collection. 

We develop a few techniques to achieve both goals at the same time. 
First, we propose the so-called \emph{additive} quantum programs  
as a \emph{succinct} intermediate representation for the collection of programs during the code transformation. Now the entire differentiation procedure will be divided into two steps: 
(1) all code transformations happen on additive programs and are very similar to classical ones (see \fig{CT}) ; 
(2) the collection of programs can be recovered by a compilation procedure from any additive program. 
Additive quantum programs are equipped with a new \emph{sum} operation that models the multiple choices as exhibited in \eqref{eqn:qmul}, which resembles a similar idea in the differential lambda-calculus~\cite{ehrhard2003differential}. 

Second, we also design a new rule for  $\parthetagenNotRelation{U(\vth)}$ which is slightly different from \cite{GM17,PhysRevA.98.062324,EH18, MAKN18, SBVCK18}. 
The existing phase-shift rule makes use of two quantum circuits for one differentiation, which causes a lot of inconvenience in the formulation and potential trouble for efficiency.  
Instead, we use only one extra ancilla as the control qubit to create a superposition of two quantum circuits and effectively achieve the same differentiation with only one quantum circuit. 
We also conduct a careful resource analysis of our differentiation procedure and show the number of required copies of initial states is reasonable comparing to the classical setting. 
The correctness of the code transformation of composition critically relies on our design choice as well as the strong definition related to the differential semantics. 

With the previous setup, we can naturally build the differentiation for quantum controls (i.e., the condition statement). Note that a general solution for classical controls is unknown~\cite{BF94} due to the non-smoothness of the guard. 
Similar to the classical setting~\cite{GP18}, we only provide a solution to deal with bounded loops and leave it open for general ones.  % and leave it open unbounded loops. 
% not have a solution to deal with general unbounded loops. We instead provide a solution to bounded loops which can be deemed as a macro consisting of simple quantum condition statements. 

\vspace{1mm} \noindent \textbf{Contributions.} We formulate the parameterized quantum bounded while-programs with classically parameterized unitary gates modelled after VQCs~\cite{IBM-QE, NC-VQE, QAOA} and their realistic examples on ion-trap machines, e.g.~\cite{zhu2018training}, in Section~\ref{Section::qch}.
 
In Section~\ref{sec::detQCHs}, we illustrate our design of \emph{additive} quantum programs. Specifically, we add the syntax
$P_1 \msquare P_2$ to represent the either-or choice between $P_1$ and $P_2$ in~\eqref{eqn:qmul}.
We formulate its semantics and compilation rules that map additive programs into collections of normal ones for our purpose. 

In Section~\ref{sec::Obssemm}, we formulate the observable and the differential semantics of quantum programs and formally define the meaning of program $S'(\vth)$ computing the differential semantics of $S(\vth)$ in the strongest possible sense. 
%In particular, we require that $S'(\vth)$ should work for any pair of observable and input state, which is critical for the compositionality. 

In Section~\ref{sec::CTRules}, we show that such a strong requirement is indeed achievable by demonstrating the code-transformation rules for the differentiation procedure. Thanks to the use of additive quantum programs, the code transformation is much simplified and as intuitive as classical ones.  
We develop a logic with the judgement $S'(\vth)|S(\vth)$ stating that $S'(\vth)$ computes the differential semantics of $S(\vth)$. We prove it sound
 and use it to show the correctness of the code transformation.

In Section~\ref{sec::resource}, we conduct a resource analysis to further justify our design. We show that the \emph{occurrence count} of parameters capture the extra resource required in both the classical auto-differentiation and our scheme. Hence, our resource cost is reasonable compared with the classical setting. 

Finally, in Section~\ref{sec::case}, we demonstrate the implementation of our code transformation in OCaml and apply it to the training of one VQC instance with controls via classical simulation. 
Specifically, this instance shows an advantage of controls in machine learning tasks, which implies the advantage of our scheme over previous ones that cannot handle controls. We have also empirically verified the resource-efficiency of our scheme on representative VQC instances. 

\vspace{1mm} \noindent \textbf{Related classical work.} There is an extensive study of automatic differentiation (AD) or differentiable programming in the classical setting (e.g., see books \cite{Griewank:2000:EDP:335134,Corliss:2000:ADA:571034}). 
The most relevant to us are those studies from the programming language perspective. 
AD has traditionally been applied to imperative programs in both the forward mode, e.g. ~\cite{Wengert:1964:SAD:355586.364791,Kedem:1980:ADC:355887.355890}, and
the reverse mode, e.g.,~\cite{Speelpenning:1980:CFP:909337}. The famous 
\emph{backpropagation} algorithm~\cite{backprop} is also a special case of reserve-mode AD used to compute the gradient of a multi-layer perceptron. 
AD has also been recently applied to functional programs~\cite{Pearlmutter:2008,  Elliott:2009,Elliott:2018}. 
Motivated by the success of deep learning, there is significant recent interest to develop both the theory and the implementation of AD techniques. Please refer to the survey~\cite{Baydin:2017} and the keynote talk at POPL'18~\cite{GP18} and ~\cite{AP20} for more details.

\section{Quantum Preliminaries} \label{sec:prelim}
We present basic quantum preliminaries here (a summary of notation in Table~\ref{tab:notation}). Details 
are deferred to \appdName\iftechrep~\ref{sec::qprel}\else\fi. 

\begin{table}
  \caption{A brief summary of notation used in this paper}
  \label{tab:notation}
  \begin{tabular}{@{} lll @{}}\\
    \textbf{Spaces} & $\H$, $\A$ & ${L}(\H)$ (Linear operators)\\
    \textbf{States} & (pure states) & $\ket{\psi}, \ket{\phi}$; \\
    && $\ket{0}, \ket{1}, \ket{+},\ket{-}$ \\
  & (density) & $\rho, \sigma$;
  $\ket{\psi}\bra{\psi}$ 
  \\
    \textbf{Operations} & (unitaries) & $U, V,\boldgreek{\sigma}$;\\
    &&$H, X,
    Z,X\otimes X$ 
    \\
  & (superoperators) & $\E, \F$ (general);\\
  &&$\Phi$ (quantum channels)\\
    \textbf{Measurements} & $M$ & $\{M_m\}_m$;\\
    &&
    $\{\ket{0}\bra{0}, \ket{1}\bra{1}\}$ (example)
    \\
    \textbf{Observables} & $O$ & $\sum_m\lambda_m\ket{\phi_m}\bra{\phi_m}$;\\&&
    $\ket{0}\bra{0}-\ket{1}\bra{1}$ (example)\\
    \textbf{Programs}& 
    (no parameters)
    & $P,Q$ \\
    &(parameterized)& $P(\vth),S(\vth)$; \\&&
    $R{\boldgreek{'}}_{\boldgreek{\sigma}}(\theta)$(notable examples)\\
&(additive) & $\udd{P(\vth)},\udd{S(\vth)}
$ ;\\&&$\parthetagenNotRelation{S(\vth)}$ (example)\\
\textbf{Semantics}& (operational)& $\lag P,\rho\rag \to \lag Q,\rho'\rag$ (steps)\\
&(denotational)& $\sem{P}\rho = \rho'$\\
&(observable)& $\ObsSemMacro{O}{P(\vth)}{\rho}$ \\
%&&(example)\\
\textbf{Code 
Process}&(Transform)& $\parthetagenNotRelation{\udd{S(\vth)}}$ (example)\\
&(Compile)& $\compl{(\udd{S'(\vth)})}$ (example)\\
\textbf{Logic}&(Judgement)&$\udd{S'(\vth)}\vert \udd{S(\vth)}$ (example) \\
\textbf{Count}&(Non-Abort)& $\NumNonAbort{\parthetajgenNotRelation{P(\vth)}}$ (example)\\
&(Occurrence)& $ \ResCount{j}{P(\vth)}$ (example)
  \end{tabular}
\end{table}

\subsection{Math Preliminaries}
Let $n$ be a natural number. We refer to the complex vector space $\mathbb{C}^n$ 
as an \emph{$n$-dimensional Hilbert space $\H$}. 
We use 
$\ket{\psi}$ to denote a complex vector in $\mathbb{C}^n$. The Hermitian conjugate of $\ket{\psi}\in \mathbb{C}^n$ is denoted by $\bra{\psi}$. The \emph{inner product} of 
$\ket{\psi}$ and $\ket{\phi}$, defined as the product of $\bra{\psi}$ and $\ket{\phi}$, is denoted by $\langle\psi|\phi\rangle$.
The norm of a vector $\ket{\psi}$ is denoted by $\nm{\ket{\psi}}=\sqrt{\langle\psi|\psi\rangle}$.

We define \emph{operators} as linear maps between Hilbert spaces, which can be represented by matrices for finite dimensions. 
Let $A$ be an operator and its Hermitian conjugate $A^\dag$. 
$A$ is \emph{Hermitian} if $A=A^\dag$. The \emph{trace} of $A$ 
is the sum of the entries on the main diagonal, i.e., $\tr(A)=\sum_i A_{ii}$. 
$\bra{\psi}A\ket{\psi}$ denotes the inner product of 
$\ket{\psi}$ and $A\ket{\psi}$. Hermitian operator $A$ is \emph{positive semidefinite} 
if for all vectors $\ket{\psi}\in\H$,
$\bra{\psi}A\ket{\psi}\geq 0$.

\subsection{Quantum States and Operations}
The state space of a \emph{qubit} 
is a 2-dimensional Hilbert space. 
Two important orthonormal bases of a qubit system are: the \emph{computational} basis with $\ket{0}=(1,0)^\dag$ and $\ket{1}=(0,1)^\dag$; 
the $\pm$ basis, consisting of $\ket{+}=\frac{1}{\sqrt{2}}(\ket{0}+\ket{1})$ and $\ket{-}=\frac{1}{\sqrt{2}}(\ket{0}-\ket{1})$.

A \emph{pure} quantum state is 
a unit vector $\ket{\psi}$.
A \emph{mixed} state, which refers to an ensemble of pure states $\{(p_i,\ket{\psi_i})\}_i$ (each with probability $p_i$),
can be represented by a \emph{density operator} that is a trace-one positive semidefinite operator $\rho=\sum_i p_i\ket{\psi_i}\bra{\psi_i}$;
%i.e., $\ket{\psi_i}$ with probability $p_i$.
%\emph{Density operators}, representing the ensemble $\{(p_i,\ket{\psi_i})\}_i$ is a positive semidefinite operator $\rho=\sum_i p_i\ket{\psi_i}\bra{\psi_i}$
%.
%A positive semidefinite operator $\rho$ on $\H$ is said to be a 
$\rho$ is a \emph{partial} density operator if $\tr(\rho)\leq 1$.
The set of partial density operators on $\H$ is denoted by $\D(\H)$.

Operations on quantum systems can be characterized by unitary operators. Denoting the set of linear operators on $\H$ as $L(\H)$, an operator $U\in L(\H)$ is \emph{unitary} if 
$U^\dag U=UU^\dag=I_\H$.  A unitary \emph{evolves} a pure state $\ket{\psi}$ to $U\ket{\psi}$
, or a density operator $\rho$ to $U\rho U^\dag$. 
Common unitary operators include: 
the \emph{Hadamard} operator $H$, which transforms between the computational and the $\pm$ basis 
via $H\ket{0}=\ket{+}$ and $H\ket{1}=\ket{-}$;
the \emph{Pauli $X$} operator 
which performs a bit flip, i.e., $X\ket{0}=\ket{1}$ and  $X\ket{1}=\ket{0}$; 
\emph{Pauli $Z$} 
which performs a phase flip, i.e., $Z\ket{0}=\ket{0}$ and $Z\ket{1}=- \ket{1}$; 
%\emph{Pauli $Y$} mapping $\ket{0}$ to $i\ket{1}$ and $\ket{1}$ to $-i\ket{0}$; 
\emph{CNOT} gate 
mapping $\ket{00}\mapsto \ket{00},\ket{01}\mapsto \ket{01},\ket{10}\mapsto \ket{11},\ket{11}\mapsto \ket{10}$. 
More generally, evolution of a quantum system can be characterized by an \emph{admissible superoperator} $\E$, 
namely a \emph{completely-positive} and \emph{trace-non-increasing} %(See \appdName\iftechrep ~\ref{sec::qprel}\else\fi) 
linear map from $\D(\H)$ to $\D(\H')$. % for Hilbert spaces $\H, \H'$. 

For every superoperator $\E 
$, there exists a set of \emph{Kraus operators} 
$\{E_k\}_k$ such that $\E(\rho)=\sum_k E_k\rho E_k^\dag$ for any input $\rho\in\D(\H)$.
The \emph{Kraus form} of $\E$ is 
therefore $\E=\sum_k E_k\circ E_k^\dag$.
%An identity operation refers to the superoperator $\mathcal{I}_{\H} = I_{\H} \circ I_{\H}$.
The Schr\"odinger-Heisenberg \emph{dual} of a superoperator $\E=\sum_k E_k\circ E_k^\dag$, denoted by $\E^*$, is defined as follows: for every state $\rho\in\D(\H)$ and any operator $A$, $\tr(A\E(\rho))=\tr(\E^*(A)\rho)$. The Kraus form of $\E^*$ is $\sum_k E_k^\dag\circ E_k$. 

\vspace{-1mm}
\subsection{Quantum Measurements} 
Quantum \emph{measurements} extracts classical information out of  quantum systems.
A quantum measurement on a system over Hilbert space $\H$ can be
described by a set of linear operators 
$\{M_m\}_m$ with $\sum_m M_m^\dag M_m=I_\H$ (identity matrix on $\H$). 
If we perform a measurement $\{M_m\}_m$ on a state $\rho$, the outcome $m$ is observed with probability $p_m=\tr(M_m\rho M_m^\dag)$ for each $m$, and the post-measurement state collapses to $M_m\rho M_m^\dag/p_m$.

\vspace{-1mm}
\section{\QCHs}\label{Section::qch} 
We adopt the \emph{bounded-loop} variant of the \quwhile-language developed by \citet{Ying16}, 
and augment it by parameterizing the unitaries, as this provides sufficient expressibility for parameterized quantum operations: indeed, ${\cabort}$, ${\cskip}$ and initialization behave 
independently of parameters, while ``parameterized measurements'' can be implemented with a regular measurement followed by a parameterized unitary.

From here onward, 
$\qvec{v}$ is a finite set of variables, and $\vth$ a length-$k$  vector of real-valued parameters. 

\subsection{Syntax}
 Define $\mathit{Var}$ as the set of quantum variables. 
We use the symbol $q$ as a metavariable ranging over quantum variables and define a \emph{quantum register} $\overline{q}$ to be a finite set of distinct variables.
For each $q\in\mathit{Var}$, its state space is denoted by $\mathcal{H}_q$.
The quantum register $\overline{q}$ is associated with the Hilbert space $\H_{\overline{q}}=\bigotimes_{q\in\overline{q}}\H_q$.
\footnote{If $\mathrm{type}(q) =$ \textbf{Bool} then $\H_q= \mathrm{span}\{\ket{0}, \ket{1}\}$.
If $\mathrm{type}(q) =$ \textbf{Bounded Int} then $\H_q$ is with basis $\{\ket{n} : n \in [-N,N]\}$ $(N\in\ZZ^+)$ for some finite $N$. We require the Hilbert space to be finite dimensional for implementation.} 
A \emph{$T$-bounded, $k$-parameterized quantum \textbf{while}-program} is generated by the following syntax: 
\vspace{-1mm}
\begin{align}
 {P}({\vth}) \enskip ::= & \quad {\cabort}[\qvec{q}]\enskip| \enskip {\cskip}[\qvec{q}] 
                \enskip | \enskip {q:=\ket{{0}}} 
                \enskip | \enskip {\qvec{q}:=U(\vth)[\qvec{q}] }
                \enskip | \nonumber \\
                & \enskip {P_1}(\vth);{P_2}(\vth) \enskip | \nonumber  \quad \qif{M[\qvec{q}]=\overline{m\to {P_{m}}({\vth})}}
                \enskip |
                 \nonumber \\
              &\enskip \qwhileT{M[\qvec{q}]=1}{{P_1}({\vth})},\nonumber \enskip 
\end{align}
\textrm{ where }
\begin{eqnarray}
  && \qwhileO{M[\qvec{q}]=1}{P_1(\vth)}\nonumber\\ &\equiv&
    \mathbf{case}\ M[\overline{q}]=
       \{ 0\to   \cskip, \enskip
        1\to   P_1(\vth); \cabort\},  \label{eqn:qwhileTStandarddefn}
 \end{eqnarray}
\begin{eqnarray}
   &&  \qwhileTgT{M[\qvec{q}]=1}{P_1(\vth)}\nonumber\\ &\equiv&
     \mathbf{case}\ M[\overline{q}]=
       \{ 0\to   \cskip, \enskip
        1\to   P_1(\vth); {\mathbf{while}^{(T-1)}}\}\nonumber
\end{eqnarray}
\emph{Unparameterized programs} can be obtained by fixing $\vth^*\in\RR^k$ in some $P(\vth)$. We denote the set of variables accessible to 
$P(\vth)$ as qVar$(P(\vth))$; 
the collection of all ``$T$-bounded 
while-programs $P(\vth)$ s.t. qVar$(P(\vth))=\qvec{v}$'' as $\augweirdCPV$, and similarly, the unparameterized one 
as $\weirdPV$.

Now let us formally define \emph{parameteriztion} of unitaries: let $\vth:=(\theta_1,\cdots,\theta_k)$ $(k\geq 1)$.
 A \emph { $k$-parameterized unitary} $U(\vth)$ is a function $\RR^k\to L(\H)$ s.t. (1) given any $\vth^*\in \RR^k$, 
    $U(\vth^*)$ is an unitary on $\calH$ 
    , and
 (2)  the parameterized-matrix representation of $U(\vth) $ is entry-wise smooth. 
    
A important family is the 
\emph{single-qubit rotations} about the Pauli axis $X,Y,Z$ with angle $\theta$ (matrix exponential here): 
\vspace{-1mm}
\begin{equation}\label{eg::SpecialRota}
\left.R_{\boldgreek{\sigma}}(\theta):= \exp \left(\frac{-i\theta}{2}\boldgreek{\sigma}\right), \boldgreek{\sigma}\in \{X,Y,Z\}\right. .
\end{equation}
One can also extend Pauli rotations to multiple qubits. For example, consider \emph{two-qubit coupling} gates $\{R_{\sigma\otimes \sigma}:=\exp(\frac{-i\theta}{2}\sigma\otimes \sigma)\}_{\sigma\in \{X,Y,Z\}}$, which generate entanglement between two qubits. Combined with single-qubit rotations, they form a \emph{universal gate set} for quantum computation. 
Another important feature of these gates is that they can already be reliably implemented in such as ion-trap quantum computers~\cite{zhu2018training}.  

As a result, we will work mostly with these gates in the rest of this paper. However, 
note that one can easily add and study other parameterized gates in our framework as well. 

The language constructed 
above 
is similar to their classical counterparts.
(0) $\cabort$ terminates the program, outputting $\mathbf{0}\in\weirdSwhat{\qvec{q}}$. (1) $\cskip$ does nothing to states in $\weirdSwhat{\qvec{q}}$.
(2) $q:=\ket{0}$ sets quantum variable $q$ to the basis state $\ket{0}$. The underlying quantum procedure is to apply super-operators $\E_{q \rightarrow 0}^{\mathrm{bool}}(\cdot)$ (or $\E_{q \rightarrow 0}^{\mathrm{int}}(\cdot)$)\footnote{$\E_{q \rightarrow 0}^{\mathrm{bool}}(\rho) = \ket{0}_q\bra{0}\rho\ket{0}_q\bra{0} + \ket{0}_q\bra{1}\rho\ket{1}_q\bra{0}$ and  $\E_{q \rightarrow 0}^{\mathrm{B-int}}(\rho)=\sum_{n=-N}^{N}\ket{0}_q\bra{n}\rho\ket{n}_q\bra{0}$ ($N\in\ZZ^+$).} to $q$ and identity operations to the rest of variables.  The correlation between $q$ and the rest of quantum variables could be potentially disturbed. 
(3) for any $\vth^*\in \RR^k$, $\overline{q}:=U(\vth^*)[\overline{q}]$ applies the unitary $U(\vth^*)$ to the qubits in $\overline{q}$.
(4) Sequencing has the same behavior as its classical counterpart.
(5) for $\vth^*\in \RR^k$, $\qif{M[\overline{q}]=\overline{m \to P_m(\vth^*)}}$ performs the measurement $M = \{M_m\}$ on the qubits in $\overline{q}$, and executes program $P_m(\vth^*)$ if the outcome of the measurement is $m$.
The bar over $\overline{m \to P_m}$ indicates that there may be one or more repetitions of this expression.
(6) 
$\qwhileT{M[\overline{q}]=1}{P_1(\vth^*)}$ performs the measurement $M = \{M_0, M_1\}$ on $\overline{q}$, and terminates if the outcome corresponds to $M_0$, or executes $P_1(\vth^*)$ then 
reiterates ($T\geq 2$) / aborts ($T=1$) otherwise. 
The program iterates at most $T$ times. 

We highlight two differences between quantum and classical while languages:
(1) Qubits may only be initialized to the 
state $\ket{0}$. There is no quantum analogue for initialization to any expression (i.e. $x:=e$) 
due to the no-cloning theorem of quantum states.
Any state $\ket{\psi} \in \H_q$, however, can be constructed by applying some unitary $U$ to $\ket{0}$.
(2) Evaluating the guard of a case statement or loop, which performs a measurement, potentially disturbs the state of the system. 

\vspace{-2mm}
\subsection{Operational and Denotational Semantics}
 We present the \emph{operational semantics} of parameterized programs in Figure \ref{fig:opsem:pardet}. Transition rules are represented as 
 $\langle P,~\rho \rangle$ $\to $ $\langle P',~\rho' \rangle$, where $\langle P,~\rho \rangle$ and $\langle P',~\rho' \rangle$ are quantum \emph{configurations}.\footnote{Recall that, fixing arbitrary $\vth^*\in\RR^k$, 
 both semantics reduce to those of unparameterized programs, so for compactness we write $P$ for $P(\vth^*)$, etc. } 
In configurations, $P$ (or $P'$) could be a quantum program or the
empty program 
$\downarrow$, and $\rho$ and $\rho'$ are partial density
operators representing the current state. 
Intuitively, in one step, we can evaluate program $P$ on input state $\rho$ to program $P'$ (or 
$\downarrow$) and output state $\rho'$.
In order to present the rules in a non-probabilistic manner, the probabilities associated with each transition are encoded in the output partial density operator.
%\footnote{If we had instead considered a probabilistic transition system, then
%one could write $\StepsToP{p_m}{\qif{M[\overline{q}]=\overline{m\to P_m}}}{\rho}{P_m}{\rho_m}$
%where $p_m = \tr(M_m\rho M_m^\dagger)$ and $\rho_m = M_m\rho M_m^\dagger / p_m$.} 
For each index $m$ of branches in a loop/control statement, the superoperator $\E_m$ is defined by $\E_m(\rho)=M_m\rho M\dagg_m$, 
yielding the post-measurement state.

%In the \emph{Initialization} rule, the superoperators  $\E_{q \rightarrow 0}^{\mathrm{bool}}(\rho)$ and $\E_{q \rightarrow 0}^{\mathrm{int}}(\rho)$, which initialize the variable $q$ in $\rho$ to $\ket{0}\bra{0}$, are defined by 
%$\E_{q \rightarrow 0}^{\mathrm{bool}}(\rho) = \ket{0}_q\bra{0}\rho\ket{0}_q\bra{0} + \ket{0}_q\bra{1}\rho\ket{1}_q\bra{0}$ and  $\E_{q \rightarrow 0}^{\mathrm{B-int}}(\rho)=\sum_{n=-N}^{N}\ket{0}_q\bra{n}\rho\ket{n}_q\bra{0}$ ($N\in\ZZ^+$). 
%Here, $\ket{\psi}_q\bra{\phi}$ denotes the outer product of states $\ket{\psi}$ and $\ket{\phi}$ associated with variable $q$; that is, $\ket{\psi}$ and $\ket{\phi}$ are in $\H_q$ and $\ket{\psi}_q\bra{\phi}$ is a matrix over $\H_q$. 
%It is 
%conventional in the quantum information literature that when operations or measurements only apply to part of the quantum system (e.g., a subset of quantum variables of the program), one should assume that an identity 
%is applied to the rest of quantum variables.  
%For example, applying $\ket{\psi}_q\bra{\phi}$ to $\rho$ means applying $\ket{\psi}_q\bra{\phi}\otimes I_{\H_{\bar{q}}}$ to $\rho$, where $\bar{q}$ denotes the set of all variables except $q$. 

We present the \emph{denotational semantics} of parameterized programs in \ref{fig:desem:pardet}, 
defining $\sem{P}$ as a superoperator 
on
$\rho \in \H_{\VE}
$~\cite{Ying16}. 
For more details 
we refer the reader to \citet{Yin11,Ying16}. 
   
We have the following connection between the denotational semantics and operational for parameterized programs: in short, the meaning of running program $P(\vth^*)$ on
input state $\rho$ and any $\vth^*\in\RR^k$ is the sum of all possible output states with multiplicity, weighted
by their probabilities.
\begin{prop}[\cite{Ying16}]\label{prop::detparadenosemwelldefn}$\forall P(\vth)\in\augweirdCPV$, and any specific $\vth^*\in \RR^k$, $\rho\in\weirdS$,
\begin{equation} \label{eqn:dsemantics_while}
    \sem{P(\vth^*)}(\rho)
    = 
    \sum \{|{\rho'}:(P(\vth^*),\rho)\to^* (\downarrow,\rho')|\}.
\end{equation}
 Here $\rightarrow^*$ is the reflexive, transitive closure of $\rightarrow$ and $\{ \vert \cdot \vert \}$ denotes a \textbf{multi-set}.
\end{prop}

\begin{figure}
  \small
  \begin{subfigure}[b]{0.5\textwidth}
    \begin{align*}
    \text{(Abort)}&\enskip \infer[]{\StepsTo{\cabort[\qvec{q}]}{\downarrow}{\mathbf{0}}}{} \\
      \text{(Skip)}&\enskip\infer[]{\StepsTo{\cskip[\qvec{q}]}{\downarrow}{\rho}}{} \\
      \text{(Init)}&\enskip\infer[]{\StepsTo{q:=\ket{0}}{\downarrow 
      }{\rho_0^q}}{}\\
     & \text{where}~\rho_0^q = \begin{cases}
                   \E_{q \rightarrow 0}^{\mathrm{bool}}(\rho)  & \text{if $type(q)$ = \textbf{Bool}}\\
                   \E_{q \rightarrow 0}^{\mathrm{B-int}}(\rho) & \text{if $type(q)$ = \textbf{Bdd Int}}
              \end{cases}  \\
      \text{(Unitary)}&\enskip\infer[]{\StepsTo{\overline{q}:=U(\vth^*)[\overline{q}]}{\downarrow 
      }{U(\vth^*)\rho U^\dagger(\vth^*)}}{}\\
      \text{(Sequence)}&\enskip \infer[]{\StepsTo{P_1(\vth^*);P_2(\vth^*)}{P_1'(\vth^*);P_2(\vth^*)}{\rho'}}{\StepsTo{P_1(\vth^*)}{P_1'(\vth^*)}{\rho'}}\\
      \text{(Case $m$)}&\enskip \infer[]{\lag \qif{M[\overline{q}]=\overline{m\to P_m(\vth^*)}},\rho \rag \to }{}\\
      & \lag {P_m(\vth^*),}{
      \E_m(\rho)}\rag \text{
      , $\forall $ outcome $m$ of 
      $M = \{ M_m \}$}\\
      \text{(While$^{(T)}$ 0)}&\enskip \infer[]{
      \lag \qwhileT{M[\overline{q}]=1}{P_1(\vth^*)},\ \rho \rag\to 
      }{} \\
      & \lag {
      \downarrow 
      ,\ }{
      \E_0(\rho)
      }\rag \\
      \text{(While$^{(T)}$ 1)}&\enskip \infer[]{
      \lag \qwhileT{M[\overline{q}]=1
      }{P_1(\vth^*)},\rho \rag \to
      }{}\\
      &\lag {P_1(\vth^*);  
      \rwhile^{(T-1)}, }{
      \E_1(\rho)}\rag 
    \end{align*}
    \caption{}
    \label{fig:opsem:pardet}
  \end{subfigure}

  \begin{subfigure}[b]{0.5\textwidth}
    \begin{align*}
      \begin{array}{rcl}
        \sem{\cabort[\qvec{q}]}\rho
         &= & \enskip \mathbf{0} \\
        \sem{\cskip[\qvec{q}]}\rho &=& \enskip \rho\\
        \sem{q:=\ket{0}}\rho  &= &\E_{q \rightarrow 0}^{\mathrm{bool}}(\rho) \text{ or } \E_{q \rightarrow 0}^{\mathrm{B-int}}(\rho)\\
        \sem{\overline{q}:= U(\vth^*)[\overline{q}]}\rho  &=&  \enskip U(\vth^*)\rho U^\dag(\vth^*) \\
        \sem{P_1(\vth^*);P_2(\vth^*)}\rho &= &\sem{P_2(\vth^*)}(\sem{P_1(\vth^*)}\rho) \\
%        & \\
        \sem{\qif{M[\overline{q}]=\overline{m\to P_m(\vth^*)}}}\rho &=& \sum_m \sem{P_m(\vth^*)}
        \E_m(\rho) \\
        \sem{\qwhileT{M[\overline{q}]=1}{P_1(\vth^*)}}\rho  & =&   \sum_{n=0}^{T-1}{\mathcal{E}_0}\circ \\ & & (\sem{P_1(\vth^*)}\circ {\mathcal{E}_1})^{n} (\rho)
     \end{array}
    \end{align*}
    \caption{}
    \label{fig:desem:pardet}
  \end{subfigure}

  \caption{
  Parameterized $T$-bounded \quwhile ~ programs: (a) operational semantics
  (b) denotational semantics. 
%  We instantiate $\vth^*\in\RR^k$ since choice of parameters affects choices made at control gates. ``\textbf{Bdd}'' stands for \textbf{Bounded}.
  }\label{fig::semPardet}
\end{figure}

We close the section with a 
notion arising from the following observation: some programs, while syntactically not ``$\cabort[\qvec{q}]$'', semantically aborts. Simple examples include $U(\vth);$ $\cabort$ or a case sentence that has $\cabort$ on each branch. 
These programs essentially don’t contribute to the finite computation output, as 
semantically aborted programs always result in \textbf{zero} output state $\mathbf{0}$. 

We formalize this concept (essential-abortion for unparameterized programs may be analogously defined) so that the compilation of our programs could be optimized: 
\begin{defn}[``Essentially Abort'']\label{defn::EssAbort} Let $P(\vth)\in \augweirdPV$. $P(\vth)$ ``\emph{essentially aborts}'' if one of the following holds:
\begin{enumerate}
    \item $P(\vth)\equiv \cabort[\qvec{q}]$;
    \item $P(\vth)\equiv P_1(\vth);P_2(\vth)$, and either $P_1(\vth)$ or $P_2(\vth)$ essentially aborts;
    \item $P\equiv\qifStandardPara$, and each $P_m(\vth)$ essentially aborts.%;
\end{enumerate}
 \end{defn}

\section{\detQCHs}\label{sec::detQCHs}
We introduce a variant of \emph{\Macroadditive}quantum programs as a succinct way to describe 
the collection of programs that are necessary to compute the derivatives. 
To that end, we introduce our design of the syntax and the semantics of 
\Macroadditive quantum programs as well as a compilation method that 
turns any \Macroadditive quantum program into a collection of normal programs for the actual computation of derivatives. 

 \subsection{Syntax}
We adopt the convention to use underlines to indicate \Macroadditive programs, such as $\udd{P(\vth)}$, to distinguish from normal program $P(\vth)$. The syntax of $\udd{P(\vth)}$ is given by 
\begin{align*}
 \udd{P(\vth)} \enskip ::= & \quad \udd{\cabort}[\qvec{q}]\enskip| \enskip\udd{\cskip}[\qvec{q}]
                \enskip | \udd{\enskip q:=\ket{{0}} }
                \enskip | \enskip \udd{\qvec{q}:=U(\vth)[\qvec{q}]} 
                \enskip |\\
                &\enskip \udd{P_1(\vth)};\udd{P_2(\vth)} \enskip | \enskip  \uqif{M[\qvec{q}]=\overline{m\to \udd{P_{m}(\vth)}}}
                \enskip |\\
                &\enskip \uqwhileT{M[\qvec{q}]=1}{\udd{P_1}(\vth)}\enskip |\enskip  \udd{P_1(\vth)}\ \msquare\ \udd{P_2(\vth)},
\end{align*}
where the only new syntax $\msquare$ is the \emph{\Macroadditive} choice. 
Intuitively, $\udd{P_1(\vth)}\ \msquare\ \udd{P_2(\vth)}$ allows the program to either execute $\udd{P_1(\vth)}$ or $\udd{P_2(\vth)}$ nondeterminisitcally. 
%This intuition will be implemented by the definition of its operational and denotational semantics. 
The denotational semantics will include all possible execution traces. We assume $\msquare$ has  lower precedence order than composition, and is left associative.\footnote{E.g., $\udd{X}\ \msquare\ \udd{Y};\udd{Z} = \udd{X}\ \msquare\ (\udd{Y};\udd{Z})$, $X\msquare Y\msquare Z:= (X\msquare Y)\msquare Z$.} If $
   \udd{P(\vth)}= \udd{P_1(\vth)}\ \msquare\ \udd{P_2(\vth)}
   $, then  {qVar}$(\udd{P(\vth)}) \equiv \textrm{qVar}(\udd{P_1(\vth)})\cup\textrm{qVar}(\udd{P_2(\vth)})$. 
 Denote the collection of all non-deterministic $\udd{P(\vth)}$ s.t. qVar$(\udd{P(\vth)})=\qvec{v}$ as $\detaugweirdCPV $.

 \subsection{Operational and Denotational Semantics}\label{subsec::OperationalDenotationalNonDet}
 We exhibit 
 operational semantics in Figure \ref{fig::semParNondet}  
 and define a similar denotational semantics for any $\udd{P(\vth^*)}\in$ 
 $ \detaugweirdCPV$.

\begin{figure}[t]
 \small
    \begin{align*}
      \text{\colorful{blue}{(Sum Components)}} &\enskip\infer[]{\StepsTo{\udd{P_1(\vth^*)}\ \msquare\ \udd{P_2(\vth^*)}}{\udd{P_{1}(\vth^*)}
      }{\rho},}{}\\
      &\enskip \StepsTo{ \udd{P_1(\vth^*)}\ \msquare\ \udd{P_2(\vth^*)}}{ \udd{P_{2}(\vth^*)}}{\rho}
    \end{align*}
  \caption{
 \detqchs: operational semantics.
   We fix $\vth^*\in\RR^k$ and inherit all the other rules from parameterized programs in Fig. \ref{fig:opsem:pardet}.
  }
  \label{fig::semParNondet}
\end{figure}

\begin{defn}[Denotational Semantics]\label{defn::DenoteSemNonDet} Fix $\vth^*$. $\forall \rho\in \weirdS$, %\in\RR^k$, 
 \begin{equation} \label{eqn:sum}
     \sem{\udd{P(\vth^*)}}(\rho) \equiv \{|\rho': \lag\udd{P(\vth^*)},\rho\rag\to^*\lag \downarrow, \rho'\rag|\}. 
 \end{equation}
\end{defn}
Note that there is no sum in (\ref{eqn:sum}) compared with (\ref{eqn:dsemantics_while}). This is because we want to capture the behavior of $\msquare$ by storing all possible execution traces in a multi-set. This resembles the idea of the sum operator in the differential lambda-calculus~\cite{ehrhard2003differential}.
% now explain the name of $\msquare$. Intuitively, the (S-C) rule in \fig{:semParNondet} provides different choices of paths which will be eventually summed up in (\ref{eqn:sum}). 
%Interestingly, some slight change of our (S-C) rule, e.g., in the study of non-deterministic quantum programs~\cite{ying2018reasoning}, can also be used to describe the interlacing behavior between different components in parallel quantum programs. 

\subsection{Compilation Rules}

We exhibit the compilation rules in Figure \ref{fig::compileReal} as a way to transform an \Macroadditive program $\udd{P(\vth)}$ into a multiset of normal programs. 
The compiled set of programs will be later used in the actual implementation of the differentiation procedure. 
%We will show, in Section~\ref{sec::resource}, that our compilation procedure is also efficient in the sense that the total number of non-essentially-aborting programs in $\compl{(\udd{P_m(\vth)})}$ is reasonably bounded, in particular for the ``fill and break'' algorithm ($\Facebook{\bullet}$) as the compilation of the case statement. 
Our compilation rule is also well-defined 
as it is compatible with the denotational semantics and operational semantics of $\udd{P(\vth)}$ in the following sense: 

 \begin{prop}\label{prop::WellDefined} Denoting with $\coprod$ the union of multisets, then for any $
 {\rho}\in\weirdS$, 
 \begin{align} \label{eqn::compilation_wd}
 \{|\rho': \rho'\neq \mathbf{0}, \rho'\in \sem{\udd{P(\vth^*)}
 }
 {\rho}|\} = &\nonumber \\ \coprod_{Q(\vth)\in\compl{(\udd{P(\vth)})}
 }\{|\rho'\neq \mathbf{0}:\lag{Q(\vth^*)},
 {\rho}\rag\to^* \lag \downarrow, \rho'\rag |\}. &
 \end{align} 
 \end{prop}
 \begin{proof}
Structural Induction. See 
\appdName\iftechrep~\ref{proof::WellDefined} \else\ \fi for details. 
 \end{proof}

Note that \eqref{eqn::compilation_wd} removes $\mathbf{0}$ from the multi-set as we are only interested in non-trivial final states.
Moreover, in $\compl{(\udd{P(\vth)})}$, some programs may essentially abort (Definition \ref{defn::EssAbort}). For implementation, we are interested in the number of $Q(\vth)\in \compl{(\udd{P(\vth)})}$ that do not essentially abort: 
\begin{defn}\label{defn::NNAP} The \emph{number of non-aborting programs} of $\udd{P(\vth)}$, denoted as $\NumNonAbort{\udd{P(\vth)}}$, is defined as 
\begin{eqnarray*}
     \NumNonAbort{\udd{P(\vth)}}=|\compl{(\udd{P(\vth)})}\setminus \{|Q(\vth)\in \compl{(\udd{P(\vth)})}:&\\ Q(\vth)\textrm{ essentially aborts.}|\}|&
 \end{eqnarray*}
where $|C|$ is  the cardinality of a multiset $C$ and 
$C_0\smallsetminus C_1$ denotes the multiset difference of $C_0$ and $C_1$.
\end{defn}
 
We remark that $\NumNonAbort{\udd{P(\vth)}}$ could be exponentially large for general $\udd{P(\vth)}$, e.g., $\udd{P(\vth)} \equiv (Q_1+R_1); ... ; (Q_n + R_n)$. 
However, as we show in Section~\ref{sec::resource}, for instances of additive programs from differentiation, this number is well bounded. (i.e., instances with exponential blow-up are \emph{irrelevant} in our context.)

 \begin{figure}
 \small
 \begin{subfigure}[t]{0.5\textwidth}
    \begin{equation*}
    \begin{array}{rl}
      \textrm{(Atomic)} & 
      \compl{(P(\vth))}\equiv \{|P(\vth)|\},
      \\
      &\textrm{if } \udd{P(\vth)}\equiv 
      {\udd{\cabort[\VE]}}\ |\ {\udd{\cskip[\VE]}}\ |\ {\udd{q:=\ket{0}}}\ 
      \\
      &|\udd{\VE:= U(\vth)[\VE]}.\\
    { \textrm{(Sequence})}&
      \compl{(\udd{P_1(\vth);P_2(\vth)})}\equiv\\
      &
     \left\{\begin{array}{l}
       \{|\cabort|\},\textrm{ if }\compl{(\udd{P_1
       (\vth)})}=\{|\cabort|\};
       \\
       \{|\cabort|\},\textrm{ if }\compl{(\udd{P_2
       (\vth)})}=\{|\cabort|\};\\
      \{|Q_1(\vth);Q_2(\vth): 
       Q_{b}(\vth)\in\compl{(\udd{P_b(\vth)})}
      |\}, \\
      \qquad\qquad\qquad\qquad\qquad\qquad\quad  \textrm{otherwise.} 
     \end{array}\right.
      \\
      \textrm{(Case }m) & 
      \compl{(\udd{\ccase})}\equiv \Facebook{\udd{\ccase}},\textrm{ described in Fig.}\ref{RTAT}.
      \\
      \textrm{(While}^{(T)}) & 
      \compl{(\udd{\rwhileT})}: \textrm{ 
      use (Case }m)\textrm{ and (Sequence)}
      .\\
     {\textrm{(Sum 
     )}} & 
      \compl{(\udd{P_1(\vth)\ \msquare\ P_2(\vth)})}\equiv\\
      &
     \left\{\begin{array}{l}
       \compl{(\udd{P_1(\vth)})}
       {\coprod}\compl{(\udd{P_2(\vth)})},     \textrm{ if }\forall b\in \{1, 
       \\
       \qquad\qquad\qquad  2\}, \compl{(\udd{P_b(\vth)})}\neq \{|\cabort|\};\\
       \compl{(\udd{P_1(\vth)})},\textrm{ if }\compl{(\udd{P_2(\vth)})}=\{|\cabort|\},\\
      \qquad\qquad\qquad\quad \compl{(\udd{P_1(\vth)})}\neq \{|\cabort|\} ;\\
       \compl{(\udd{P_2(\vth)})},\textrm{ if }\compl{(\udd{P_1(\vth)})}=\{|\cabort|\},\\
        \qquad\qquad\qquad\quad\compl{(\udd{P_2(\vth)})}\neq \{|\cabort|\};\\
       \{|\cabort|\},   \textrm{otherwise} 
     \end{array} \right.
    \end{array}
    \end{equation*}
    \caption{}\label{CompilationNewCase}
    \end{subfigure}
    
    \begin{subfigure}[t]{0.5\textwidth}
    \begin{enumerate}
        \item $\forall m\in [0,w]$, let $C_m$ denote the sub-multiset of $\compl{(\udd{P_m(\vth)})}$ composed of programs that do \emph{not} essentially abort; 
    without loss of generality, assume $|C_0|\geq |C_1|\geq\cdots\geq |C_w|$. 
        \item 
        If all $C_m$'s are empty, return $\Facebook{\ccase}\equiv \{|\cabort[\VE]|\}$; 
        else, 
        pad each $C_m$ to size $|C_0|$ by adding  
        ``$\cabort[\VE]$''. 
        \item $\forall m\in [0,w]$, index programs in $C_m$ as $\{|Q_{m,0}(\vth),\cdots$, $Q_{m,|C_0|-1}(\vth)|\}$. Return $\Facebook{\ccase}$ $\equiv$ $\{|$ $\qif{M[\qvec{q}]=\overline{m\to {Q_{m,j^*}}}}$ $|\}_{j^*}$ with ${0\leq j^*\leq |C_0|-1}$.
    \end{enumerate}
     \caption{}\label{RTAT}
    \end{subfigure}
  \caption{
  nondeterministic ~ programs: 
  (a) compilation rules. (b) ``Fill and Break'' (``$\Facebook{\bullet}$'') 
  procedure for computing $\compl{(\udd{\ccase})}$.  
  $\udd{\ccase}$ stands for $\uqif{M[\overline{q}]=\overline{m\to \udd{P_m(\vth)}}}$; $\udd{\rwhileT}$ stands for $\uqwhileTParaStandard$. 
  Here $\coprod$ denotes union of multisets. One may observe from a routine structural induction and the definition of ``essentially abort'' that:  
  for all $\udd{P(\vth)}$, either $\compl{(\udd{P(\vth)})} = \{|\cabort|\}$, or  $
  \compl{(\udd{P(\vth)})}$ does not contain essentially abort programs.
  }\label{fig::compileReal}
\end{figure}
 
\begin{example}[Generic-Case]\label{eg::CaseP1BoxP2P3}
 Consider the following simple program with the case statement
  \begin{equation*}
        {\left. \begin{array}{rl} \udd{P(\vth)}\equiv \udd{\mathbf{case}}\ M[\overline{q}]=
        0\to  & \enskip \udd{P_1(\vth)}\ \msquare \udd{P_2(\vth)}, \\
        1\to  &\enskip \udd{P_3(\vth)} 
     \end{array}  \right. }
    \end{equation*}
    where $P_1(\vth),P_2(\vth),P_3(\vth)\in\augweirdCPV$, none of them essentially aborts, and each of $P_1(\vth),P_2(\vth),P_3(\vth)$ contains no control gates. 
    Then for any $\rho\in \weirdS$, fixing $\vth^*$ we have 
    \begin{equation*}
        \begin{array}{rl}
        \lag P(\vth^*),\rho\rag \stackrel{\textrm{(Case }m)}{\to}     &\lag \udd{P_1(\vth^*)}\ \msquare \udd{P_2(\vth^*)},M_0\rho M\dagg_0\rag  \\ \stackrel{\textrm{(Sum})}{\to}
        &\lag \udd{P_1(\vth^*)},M_0\rho M\dagg_0\rag \\ {\to^*} & \lag\downarrow, \sem{P_1(\vth^*)}(M_0\rho M\dagg_0) \rag;\\
         \lag P(\vth^*),\rho\rag \stackrel{\textrm{(Case }m)}{\to}&\lag \udd{P_1(\vth^*)}\ \msquare \udd{P_2(\vth^*)},M_0\rho M\dagg_0\rag\\ \stackrel{\textrm{(Sum})}{\to}  &\lag \udd{P_
         {2}(\vth^*)},M_0\rho M\dagg_0\rag {\to^*} \lag\downarrow, \sem{P_{
         {2}}(\vth^*)}(M_0\rho M\dagg_0) \rag;\\
         \lag P(\vth^*),\rho\rag \stackrel{\textrm{(Case }m)}{\to}&\lag 
         {\udd{P_3(\vth^*)}},M_{1}
         \rho M\dagg_{1}
         \rag \\
         {\rightarrow^*}& \lag\downarrow, 
         {\sem{P_{{3}}(\vth^*)}}(M_{1}
         \rho M\dagg_{1}
         ) \rag
        \end{array}
    \end{equation*}
    Hence by Definition \ref{defn::DenoteSemNonDet}.
    \begin{align*}
    \sem{\udd{P(\vth^*)}}\rho = \{| \sem{P_1(\vth^*)}(M_0\rho M\dagg_0) , \sem{P_2(\vth^*)}(M_0\rho M\dagg_0), &\\ \sem{P_3(\vth^*)}(M_1\rho M\dagg_1) |\}&
    \end{align*}
We verify 
computation results from the compilation rules are consistent with this. Writing ``compilation rule'' as ``CP'' for short, one observes
$
\compl{(\udd{P_1(\vth)}\msquare\udd{P_2(\vth)})}\stackrel{\textrm{CP,Sum}}{=}\{|{P_1(\vth)},{P_2(\vth)}|\},
$ while $\compl{(\udd{P_3(\vth))}}=\{|P_3(\vth)|\}$ since we assumed non-
essentially-abortness. 
Apply our ``fill and break'' procedure to obtain $C_0=\{|P_1(\vth),P_2(\vth)|\},$ $C_1 = \{|P_3(\vth),\cabort[\VE]|\}$.
 
 $\begin{array}{rl}
  \compl{(\udd{P(\vth)})}  =&\Big{\{}|{\left. \begin{array}{rl}  \udd{\mathbf{case}}\ M[\overline{q}]=
        0\to  & \enskip {P_1(\vth)}, \\
        1\to  &\enskip {P_3(\vth)}\boldgreek{,} 
     \end{array}  \right. } \\
  &{\left. \begin{array}{rl}  \udd{\mathbf{case}}\ M[\overline{q}]=
        0\to  & \enskip {P_2(\vth)}, \\
        1\to  &\enskip {\cabort[\VE]} 
     \end{array}  \right. }|\Big{\}}\label{eqn::U1U2U3EgLaterUse}
\end{array}$

\noindent It's easy to check that evolving pursuant to the normal programs operational semantics (Fig \ref{fig::semPardet}) agrees with $\sem{\udd{P(\vth^*)}}\rho$.
%\[
%\coprod_{Q(\vth)\in\compl{(\udd{P(\vth)})}
% }\{|\rho':\lag{Q(\vth^*)},
% {\rho}\rag\to^* \lag \downarrow, \rho'\rag |\}\]
%is  $\{|\sem{P_1(\vth^*)}(M_0\rho M\dagg_0) , \sem{P_2(\vth^*)}(M_0\rho M\dagg_0), \sem{P_3(\vth^*)}(M_1\rho M\dagg_1),\mathbf{0}|\}$.
%\begin{eqnarray*}
% &&\coprod_{Q(\vth)\in\compl{(\udd{P(\vth)})}
% }\{|\rho':\lag{Q(\vth^*)},
% {\rho}\rag\to^* \lag \downarrow, \rho'\rag |\}\nonumber\\
% &=&\{|\sem{P_1(\vth^*)}(M_0\rho M\dagg_0) , \sem{P_2(\vth^*)}(M_0\rho M\dagg_0), \nonumber\\&& \sem{P_3(\vth^*)}(M_1\rho M\dagg_1),\mathbf{0}|\}.
% \end{eqnarray*}
\end{example}

\section{Observable and Differential Semantics}
\label{sec::Obssemm} 
To capture physically observable quantities from quantum systems, physicists propose the notation of $\emph{observable}$ which is a Hermitian matrix over the same space. 
Any observable $O$ is a combination of information about quantum measurements and classical values for each measurement outcome. 
To see why, let us take its spectral decomposition of $O=\sum_m \lambda_m \ket{\psi_m}\bra{\psi_m}$. 
Then $\{\ket{\psi_m}\bra{\psi_m}\}_m$ form a projective measurement. 
We can design an experiment to perform this projective measurement and output $\lambda_m$ when the outcome is $m$. 
The \emph{expectation} of the output is exactly given by 
\begin{equation}
    \tr(O\rho)=\sum_m \lambda_m\tr(M_m\rho M\dagg_m).
\end{equation}
%Comparing to the random outcome of any quantum measurement, 
The expectation $\tr(O\rho)$ represents meaningful classical information of quantum systems, which is also used in the loss functions in quantum machine learning applications.  
%In applications of quantum machine learning, it is also the quantity used in the loss function. 
Thus, given any observable $O$, we will define the \emph{observable} semantics of quantum programs as both the mathematical object to take derivatives 
from the original programs and the read-out of the programs that compute these derivatives. 

%Moreover, one should note the extra cost to approximate $\tr(O\rho)$ to high precision. 
%For example, 
One can repeat the $\{\ket{\psi_m}\bra{\psi_m}\}_m$ measurement and use the statistical information to recover $\tr(O\rho)$. 
The number of iterations depends on the additive precision $\delta$ and the norm of $O$. 
To simplify our presentation, also to make a precise resource count as detailed in Section~\ref{sec::resource}, we assume that\footnote{$\sqsubseteq$ is defined by $A\sqsubseteq B\iff B-A$ positive semidefinite. } %See \appdName\iftechrep~\ref{proof::WellDefined}\else\ \fi.} 
\begin{equation}
    -I_{\H}\sqsubseteq O\sqsubseteq I_{\H}.
\end{equation}
Note that the observable $O$ is different from quantum predicate $P$ ($ 0 \sqsubseteq P \sqsubseteq I$), which is defined~\cite{DP2006} as the quantum analogue of continuous logic with true values in $[0,1]$. 
By statistically concentration bounds (e.g. the Chernoff bound), to approximate $\tr(O\rho)$ with additive error $\delta$, one needs to repeat $O(1/\delta^2)$ times with $O(1/\delta^2)$ copies of initial states. 

\subsection{Observable Semantics}\label{subsec::obssemInput}
We define the \emph{observable semantics} of both normal (denoted by $P(\vth),P'(\vth)$) and \Macroadditive (denoted by $\udd{S(\vth)},\udd{S'(\vth)}$) parameterized programs as follows. 

\begin{defn}[Observable Semantics]\label{defn::ObsSemInput}
$\forall P(\vth)$ $\in$ 
$\augweirdCPV$, any observable $O\in\weMV$ and input state $\rho\in\weirdS$, the \emph{observable semantics} of $P$, denoted $\ObsSemMacro{O}{P(\vth)}{\rho}$, is  
\begin{equation}\label{eqn::ObsSemInput}
\ObsSemMacro{O}{P(\vth)}{\rho} (\vth^*) \equiv \tr(O \sem{P(\vth^*)}\rho), \forall \vth^* \in \RR^k. 
\end{equation}
Namely, $\ObsSemMacro{O}{P(\vth)}{\rho}$ is a function from $\RR^k$ to $\RR$ whose value per point is given by \eq{ObsSemInput}. 

Similarly, for any 
$\udd{S(\vth)}\in $
$\detaugweirdCPV$ with $\texttt{Compile}$ ${( 
\udd{S(\vth)})}$ $= $ $\{|{P_i}(\vth)|\}_{i=1}^t$ where 
$P_i(\vth)\in$ $\augweirdCPV$,  
its \emph{observable semantics} is given by, $\forall \vth^* \in \RR^k$,  
    \begin{equation}\label{eqn::InputSpaceSemNonDet}
    \ObsSemMacro{O}{\udd{S(\vth)}}{\rho} (\vth^*) \equiv \sum_{i\in [1,t]} \ObsSemMacro{O}{P_i(\vth)}{\rho}(\vth^*).  
    \end{equation}
\end{defn}

To compute gradients of quantum observables for each parameter, one needs an ancilla variable as hinted by results in quantum information theory about gradient calculations for simple unitaries (e.g., ~\citet{SBVCK18,bergholm2018pennylane}).  
To that end, we can easily extend quantum programs with ancilla variables. 
For each $j\in [1, k]$, the \emph{$j$-th ancilla of $\augweirdPV$}  is a quantum variable denoted by $ A_{j,\qvec{v}}$ 
disjoint from  $\qvec{v}$. We write $A$ instead of $A_{j,\VE}$ when $j,\VE$ are clear from context. Ancilla $A$ could consist of any number of qubits while we will mostly use \emph{one-qubit} $A$ in this paper.  
 
 We will only consider programs augmented with \emph{one ancilla variable} $A_j$ at any time. (So let us fix $j$ for the following discussion). 
We will then consider programs that operate on the larger space $\weirdSwhat{\VA}$ and an additional observable $A$ to define the \emph{observable semantics with ancilla}. 

\begin{definition}[Observable Semantics with Ancilla]\label{defn::ObsSemAncilla}
Given any $P'(\vth)\in\augweirdCPVA$, any observable $O\in\weMV$, input state $\rho\in\weirdS$, and moreover the observable $O_A$ on ancilla $A$, the \emph{observable semantics} with ancilla of $P$, overloading the notation $\ObsSemMacro{O}{P'(\vth)}{\rho}$, is  
\begin{eqnarray}\label{eqn::ObsSemInput2}
\ObsSemMacro{(O, O_A)}{P'(\vth)}{\rho} (\vth^*) \equiv  \nonumber \\
\tr\Big{(}\big{(} \modify{O_A}\otimes O\big{)}\sem{P'(\vth^*)}((
                    \modify{\ket{\qvec{0}}_{\Anc}\bra{\qvec{0}}})\otimes \rho)\Big{)}
, \forall \vth^* \in \RR^k. &\label{eqn::obs33}
\end{eqnarray}
Again, $\ObsSemMacro{(O, O_A)}{P(\vth)}{\rho}$ is a function from $\RR^k$ to $\RR$ whose value per point is given by \eq{obs33}. 

Similarly, for  
$\udd{S'(\vth)}\in $  $ \detaugweirdCPVA$ 
s.t.  $
\texttt{Compile}$ ${(\udd{S'(\vth)})}=\{|{P'_i}(\vth)|\}_{i=1}^t$ where $P'_i(\vth)\in \augweirdCPVA$,  
its \emph{observable semantics} is: 
$\forall \vth^* \in \RR^k$ ,
%    \begin{eqnarray}
%    \ObsSemMacro{(O, O_A)}{\udd{S'(\vth)}}{\rho} (\vth^*) & \equiv \nonumber \\\sum_{i\in [1,t]}\ObsSemMacro{(O, O_A)}{P'_i(\vth)}{\rho}(\vth^*) &\label{eqn::InputSpaceSemNonDet2}
%    \end{eqnarray}
   \begin{equation*}
    \ObsSemMacro{(O, O_A)}{\udd{S'(\vth)}}{\rho} (\vth^*) \equiv \sum_{i\in [1,t]}\ObsSemMacro{(O, O_A)}{P'_i(\vth)}{\rho}(\vth^*). 
    \end{equation*}
\end{definition}

The only difference from the normal observable semantics lies in \eq{obs33}, where we initialize the ancilla with $\ket{0}$, which is a 
natural choice and evaluate the observable $O_A \otimes O$. 
As we will see in the technique, the independence between $O_A$ and $O$ in the form of $O_A \otimes O$ will help us obtain the strongest guarantee of our differentiation procedure. 

\subsection{Differential Semantics}
Given the definition of observable semantics, its differential semantics can be naturally defined by  

\begin{defn}[Differential Semantics]
\label{defn::ComputeDeriSem} 
Given \Macroadditive program $\udd{S(\vth)}\in\detaugweirdCPV$, its $j$-th \emph{differential semantics} is defined by 
\begin{equation} \label{eqn::diff_semantics}
    \parthetajgenNotRelation{\ObsSemMacro{O}
    {{\udd{S(\vth)}}}{\rho}},
\end{equation}
which is again a function from $\RR^k$ to $\RR$. Moreover, for any $\udd{S'(\vth)}\in\detaugweirdCPVa$ with ancilla $A$, we say that ``\emph{$\udd{S'(\vth)}$ computes the $j$-th differential semantics of $\udd{S(\vth)}$}'' if and only if there exists an observable $O_A$ on ancilla $A$ for $\udd{S'(\vth)}$ such that $ \forall O\in\weMV,\rho\in\weirdS,$
\begin{equation} \label{eqn::s_diff_s}
   \ObsSemMacro{(O, O_A)}
    {{\udd{S'(\vth)}}}{\rho} =
    \parthetajgenNotRelation{\ObsSemMacro{O}
    {{\udd{S(\vth)}}}{\rho}}.
\end{equation}
\end{defn}

We remark that \eq{diff_semantics} is well defined because $\ObsSemMacro{O}{{\udd{S(\vth)}}}{\rho}$ is a function from $\RR^k$ to $\RR$. 
It is also a smooth function because we assume that parameterized unitaries are entry-wise smooth, and the observable semantics is obtained by multiplication and addition of such entries. 
Note also that there is one specific choice of $O_A$ in our current design. We leave it as a parameter to allow flexibility for future designs. 

We also remark that the order of quantifiers in \eq{s_diff_s} 
is the strongest that one can hope for. 
This is because the observable semantics of $\udd{S(\vth)}$ will depend on $O$ and $\rho$ in general. Thus, the program to compute its differential semantics could also depend on $O$ and $\rho$ in general. 
However, in our definition, $\udd{S'(\vth)}$ is a single fixed program that works for any $O$ and $\rho$ regardless of the seemingly complicated relationship. This definition is consistent with the classical case where a single program can compute the derivatives for any input. We can achieve the same definition in the quantum setting and it is critical in the proof of Theorem~\ref{thm::sound} (item (5)). 

\section{Code Transformations and the Differentiation 
Logic}\label{sec::CTRules}
We describe the code transformation rules of the differentiation operator $\parthetagenNotRelation{\cdot}$ in Section \ref{sec::CTT}. We also define a logic and prove its soundness for reasoning about the correctness of these code transformations, with the following judgement 
\begin{equation}
    \udd{S'(\vth)}\vert\udd{S(\vth)}, 
\end{equation}
which states that $\udd{S'(\vth)}\textrm{ computes the differential semantics of }$ $\udd{S(\vth)}$ in the sense of Definition~\ref{defn::ComputeDeriSem}. 
We fix $\theta=\theta_j$ and hence $A$ stands for $A_{j,\VE}$ and $\parthetagenNotRelationNopar$ for $\parthetajgenNotRelationNopar$ through this section.\footnote{If $A$ already exists, i.e., $\udd{S(\vth)} \in \detaugweirdCPVa$, we treat $\overline{v}_{\mathrm{new}}$ as $\overline{v}_{\mathrm{old}} \cup A_{\mathrm{old}}$ and add $A_{\mathrm{new}}$.
Any observable $O$ on $\overline{v}_{\mathrm{old}}$ becomes $
O_{A_{\mathrm{old}}}\otimes O$ on $\overline{v}_{\mathrm{new}}$.
Both $A_{\mathrm{old}}$ and $A_{\mathrm{new}}$ are initialized to $\ket{0}$ in observable semantics.}

\subsection{Code Transformations}\label{sec::CTT} 
We first define some gates associated with the single-qubit rotation and the two-qubit coupling gates, which will appear in the code transformation rules. Let $A$ be a single qubit. 

\begin{defn}\label{defn::specialUnitaryPrimes}
 \begin{enumerate}
     \item %Consider $ {q_1}:=R_\sigma(\theta)\ket{q_1}$,  
     Consider unitary $R_\sigma(\theta)$ where $\sigma\in \{X,Y,Z\}$. We define unitary $C\_R_{\sigma}(\theta)$ as 
     \begin{equation}
      C\_R_{\sigma}(\theta) \equiv \ket{0}_A\bra{0}\otimes R_\sigma(\theta) + \ket{1}_A\bra{1}\otimes R_\sigma(\theta\modify{+\pi}). 
        \label{eqn::controlRot}
     \end{equation}
     We also define a new gadget program $R{\boldgreek{'}}_\sigma(\theta)$ as 
     \begin{eqnarray}
     R{\boldgreek{'}}_\sigma(\theta)[A,q_1] &\equiv& A:= H[A]; A,q_1:=C\_R_\sigma(\theta)[A,q_1]; \nonumber \\ 
          & &   A:= H[A].\label{eqn::ConjugateControlRot}
     \end{eqnarray}
     \item Substituting $\sigma\otimes \sigma$ for $\sigma$ and $q_1,q_2$ for $q_1$ in Eqns (\ref{eqn::controlRot},\ref{eqn::ConjugateControlRot}), one defines $C\_R_{\sigma\otimes \sigma}(\theta), R\boldgreek{'}_{\sigma\otimes \sigma}(\theta)$.
 \end{enumerate}
\end{defn}

For $1$-qubit rotation $R_\sigma(\theta)$, the ``controlled-rotation'' gate $C\_R_\sigma(\theta)$ 
maps $\ket{0,q_1}\mapsto \ket{0}\otimes R_\sigma(\theta)$ $\ket{q_1}$, and  $\ket{1,q_1}\mapsto \ket{1}\otimes R_\sigma(\theta
{+\pi})$  $\ket{q_1}$; $R\boldgreek{'}_\sigma(\theta)$ conjugates $C\_R_\sigma(\theta)$ with Hadamard. Similarly for corresponding two-qubit coupling gates. 

We exhibit our code transformation rules in Figure \ref{fig:CT}. For Unitary rules we only include $1$-qubit rotations and two-qubit coupling gates, since they form a universal gate set and are easy to implement on quantum machines. 
It is also possible to include more unitary rules (e.g., by following the calculations in \cite{SBVCK18}), which we will leave as future directions. 

\begin{figure}[t]
    \begin{equation*}
    \begin{array}{lll}
      \textrm{(Trivial)} & 
      \parthetagenNotRelation{\udd{\cabort[\VE]}}, \parthetagenNotRelation{\udd{\cskip[\VE]}}, \parthetagenNotRelation{\udd{q:=\ket{0}}} \enskip \equiv &
      \\
      &\qquad\qquad\qquad\qquad\qquad\quad   \udd{\cabort[\VA]}.& \\
      \textrm{(Trivial-U)}&\parthetagenNotRelation{\udd{\VE:=U(\vth)[\VE]}} \equiv  \udd{\cabort[\VA]}, \textrm{if }\theta_j\notin \vth.\\
%      &\qquad\qquad\qquad\qquad\qquad\qquad\qquad \textrm{if }\theta_j\notin \vth.&\\
      \textrm{(1-qb)} & 
      \parthetagenNotRelation{\udd{q_1:=R_\sigma(\theta)[q_1]}} \equiv  \udd{A,q_1:=R\boldgreek{'}_\sigma(\theta)[A,q_1]}.
      \\
      \textrm{(2-qb)} & 
      \parthetagenNotRelation{\udd{q_1,q_2:=R_{\sigma\otimes \sigma}(\theta)[q_1,q_2]}} \enskip \equiv & \\
      & \enskip \udd{A,q_1,q_2:=R\boldgreek{'}_{\sigma\otimes \sigma}(\theta)[A,q_1,q_2]}.&
      \\
     \textrm{(Sequence)} & 
     \parthetagenNotRelation{\udd{S_{{1}}(\vth);S_{{2}}(\vth)}} \enskip \equiv \enskip 
      ({\udd{S_1(\vth)};\parthetagenNotRelation{\udd{S_2(\vth)}}})\ \msquare&\\
      &\qquad\qquad\qquad\qquad\quad ({\parthetagenNotRelation{\udd{S_1(\vth)}};\udd{S_2(\vth)}}).&\\
      \textrm{(Case)}&  \parthetagenNotRelation{\uqifParaStandardS} \enskip\equiv&\\
      &\qquad 
      \udd{\ccase}\ M[
         \qvec{q}] =  \overline{\ m \to \parthetagenNotRelation{\udd{S_m(\vth)}}}~\mathbf{end}.
      &\\
      \textrm{(while}^{(T)})& \textrm{Use (Case) and (Sequence).}& \\
      \textrm{(S-C)} & \parthetagenNotRelation{\udd{S_1(\vth)}\ \msquare\ \udd{S_2(\vth)}} \equiv \parthetagenNotRelation{\udd{S_1(\vth)}}\ \msquare  \parthetagenNotRelation{\udd{S_2(\vth)}}.
    \end{array}
    \end{equation*}
  \caption{Code Transformation Rules. For (1-qb Rotation) and (2-qb Coupling), $(\sigma\in\{X,Y,Z\})$; $R'_{\sigma}(\theta),R'_{\sigma\otimes\sigma}(\theta)$ are as in Definition \ref{defn::specialUnitaryPrimes}. $\theta_j\notin\vth$ means ``the unitary $U(\vth)$ trivially uses $\theta_j$'': for example in $P(\vth)\equiv R_{X}(\theta_1);R_{Z}(\theta_2)$, $\vth=(\theta_1,\theta_2)$ and $R_{X}(\theta_1)$ trivially uses $\theta_2$. 
  }
  \label{fig:CT}
\end{figure}

\subsection{The differentiation logic and its soundness}\label{subsec::AutoDiffLogic}

\begin{figure}
  \begin{gather*}
      \begin{tabular}{rl}
    \text{(Abort)} \enskip \infer[]{
    \parthetagenNotRelation{\udd{\cabort[\VE]}}|\udd{\cabort[\VE]}}{} & \text{(Skip)}   \enskip \infer[]{
    \parthetagenNotRelation{\udd{\cskip[\VE]}}|\udd{\cskip[\VE]}}{}
  \end{tabular}
  \end{gather*}
  \begin{gather*}
  \begin{tabular}{ll}
   \text{(Initialization)}
    & \infer[]{\parthetagenNotRelation{\udd{q:=\ket{0}}}|\udd{(q:=\ket{0})}}{} \\
      \text{(Trivial-Unitary)}
      & \infer[]
      {\parthetagenNotRelation{\udd{\overline{q}=U(\vth)[\VE]}}|\udd{\overline{q}=U(\vth)[\VE]} }{\theta_j\notin\vth}\\
   \text{(Rot-Couple)}& \infer[]
      {\parthetagenNotRelation{\udd{\overline{q}=R_{\boldgreek{\sigma}}(\vth)[\VE]}}|\udd{(\overline{q}=
      R_{\boldgreek{\sigma}}(\vth)[\VE])} }{}\\
  \text{(Sequence)}&\infer[]
      {\parthetagenNotRelation{\udd{S_0(\vth);S_1(\vth)}}|(\udd{S_0(\vth);S_1(\vth)})}
      {\parthetagenNotRelation{\udd{S_0(\vth)}}|\udd{S_0(\vth)} 
       \qquad \parthetagenNotRelation{\udd{S_1(\vth)}}|\udd{S_1(\vth)}}\\
    \text{(Case)} &\infer[]
      {\parthetagenNotRelation{\qif{M[\overline{q}]=\overline{m\to {\udd{S_m(\vth)}}}} }|
                  }
      {\forall m, \parthetagenNotRelation{\udd{S_m(\vth)}}|\udd{S_m(\vth)}} \\
      { } & $\qif{M[\overline{q}]=\overline{m\to \udd{S_m(\vth)}}}$
      \\
    \text{(While{$^{(T)}$})} &\infer[]
     {
       \parthetagenNotRelation{{\uqwhileTParaStandardS}}|}
       {\parthetagenNotRelation{\udd{S_1(\vth)}}\Big{|}\udd{S_1(\vth)}}\\
       &${\uqwhileTParaStandardS}$\\
       \text{(Sum Component)}&\infer[]
     {\parthetagenNotRelation{\udd{S_0(\vth)}\  \msquare\  \udd{S_1(\vth)}}|(\udd{S_0(\vth)}\ \msquare\ \udd{S_1(\vth))}
                 }
     {\parthetagenNotRelation{\udd{S_0(\vth)}}|\udd{S_0(\vth)} 
           \qquad \parthetagenNotRelation{\udd{S_1(\vth)}}|\udd{S_1(\vth)}}
  \end{tabular}
 \end{gather*}
  \caption{The differentiation logic. Wherever applicable, $q\in \VE,\qvec{q}\subseteq\qvec{v}, \udd{S_i(\vth)}\in \detaugweirdCPV$. In (Rot-Couple), $
  {\boldgreek{\sigma}\in\{X,Y,Z,X\otimes X,Y\otimes Y, Z\otimes Z\}}$. 
  }
  \label{fig:reliability-sem}
\end{figure}

We develop the differentiation logic given in \fig{reliability-sem} to reason about the correctness of code transformations. It suffices to show that our logic is sound. 
For ease of notation, in future analysis we write $\parthetagenNotRelation{P(\vth)}$ in place of $\parthetagenNotRelation{\udd{P(\vth)}}$ when $P(\vth)\in\augweirdCPV$.

\begin{theorem}[Soundness]\label{thm::sound} 
Let $\udd{S(\vth)}\in\detaugweirdCPV$, $\udd{S'(\vth)}\in\detaugweirdCPVa$. Then, $\udd{S'(\vth)}\vert\udd{S(\vth)}$ implies that $\udd{S'(\vth)}$ computes the differential semantics of $\udd{S(\vth)}$.
\end{theorem}

Let us highlight the ideas behind the proof of the soundness and all detailed proofs are deferred to 
\appdName\iftechrep~\ref{app::CTRules}\else\fi. First remember that $\theta=\theta_j$ and for all the proofs we can choose $Z_A=\proj 0 -\proj 1$ as the observable on the one-qubit ancilla $A$. Thus, we will omit $Z_A$ and \emph{overload} the notation,  $\forall P'(\vth)\in\augweirdCPVa$: 
\begin{equation}\label{eqn::OverLoadingRefByAppendix}
\ObsSemMacro{O}{P'(\vth)}{\rho} \text{ means } \ObsSemMacro{(O, Z_A)}{P'(\vth)}{\rho},
\end{equation}
to simplify the presentation. We make similar overloading convention for $\udd{S'(\vth)}\in \detaugweirdCPVa$. Let us go through these logic rules one by one.

\begin{enumerate}[leftmargin=2mm]
    \item \emph{Abort, Skip, Initialization, Trivial-Unitary} rules work because these statements do not depend on $\theta$. 
    \item Since While{$^{(T)}$} can be deemed as a macro of other statements, the correctness of \emph{While{$^{(T)}$}} rule follows by unfolding while$^{(T)}$ and applying other rules. 
    \item The \emph{Sum Component} rule is due to the property of observable semantics ($\sem{\cdot}$) and \Macroadditive operator ($\msquare$):
    \begin{equation}
        \parthetagenNotRelation{\sem{P_1 \msquare P_2}} = \sem{\parthetagenNotRelation{P_1}} + \sem{\parthetagenNotRelation{P_2}}, 
    \end{equation}
    which follows from our definition design. 
    \item Our \emph{Rot-Couple} rule is different from the phase-shift rule in~\cite{SBVCK18} by using only one circuit in derivative computing. However, the proof of the \emph{Rot-Couple} rule is largely inspired by the one of the phase-shift rule.
    \item The proof of the \emph{Sequence} rule relies very non-trivially on our design of the observable semantics with ancilla (Definition~\ref{defn::ObsSemAncilla}) and the strong requirement of computing differential semantics in Definition~\ref{defn::ComputeDeriSem}. Firstly, note that
    \begin{eqnarray*}
         \ObsSemMacro{O}{\parthetagenNotRelation{\udd{S_0(\vth);S_1(\vth)}}}{\rho} & =  \ObsSemMacro{O}{\parthetagenNotRelation{\udd{S_0(\vth)}};\udd{S_1(\vth)}}{\rho}\\
       & +   \ObsSemMacro{O}{\udd{S_0(\vth)};\parthetagenNotRelation{\udd{S_1(\vth)}}}{\rho}.
    \end{eqnarray*}
    We use the induction hypothesis to reason about each term above. Consider the case  $\udd{S_0(\vth)}=S_0(\vth)$ and $\udd{S_1(\vth)}=S_1(\vth)$. Note that $S_0(\vth),$ $ S_1(\vth)$ $\in$ 
    $\augweirdCPV$ and  $\parthetagenNotRelation{S_0(\vth)}, $ $ \parthetagenNotRelation{S_1(\vth)}$  $ \in$ $\detaugweirdCPVa$. First, we show
    \begin{equation}
      \ObsSemMacro{O}{S_0(\vth);\parthetagenNotRelation{S_1(\vth)}}{\rho}  =
        \ObsSemMacro{O}{\parthetagenNotRelation{S_1(\vth)}}{\sem{S_0(\vth)}(\rho)} \label{eqn::seq_1}.
    \end{equation}
    This is because $\parthetagenNotRelation{S_1(\vth)}$ computes the derivative for \emph{any} input state and observable. We simply choose the input state $\sem{S_0(\vth)}(\rho)$ and observable $O$.  Secondly, we show 
    \begin{equation}
     \ObsSemMacro{O}{\parthetagenNotRelation{S_0(\vth)};S_1(\vth)}{\rho} =\ObsSemMacro{\sem{S_1(\vth)}^*(O)}{\parthetagenNotRelation{S_0(\vth)}}{\rho}. \label{eqn::seq_2}
    \end{equation}
For \eqref{eqn::seq_2}, we don’t change the state $\rho$ but change the observable $O$ by applying the \emph{dual} super-operator $\sem{S_1(\vth)}^*$.
Since $ \parthetagenNotRelation{S_0(\vth)}$ computes the derivative for any input state and any observable, we choose the input state $\rho$ and observable $\sem{S_1(\vth)}^*(O)$.
The dual super-operator $\sem{S_1(\vth)}^*$ has the property that $\tr(O \sem{S_1(\vth)}(\rho)) = \tr (\sem{S_1(\vth)}^*(O) \rho)$, which corresponds to the Schrodinger picture (evolving states) and Heisenberg picture (evolving observables) respectively in quantum mechanics. 
    \item The proof of the \emph{Case} rule basically follows from the linearity of the observable semantics and the smooth semantics of Case. It is interesting to compare with the classical case~\cite{BF94} where the non-smoothness of the guard causes an issue for auto differentiation. 
   
\end{enumerate}

\begin{example}[Simple-Case]\label{ex:ToBeInSection6} 
Consider the following simple instantiating of Example~\ref{eg::CaseP1BoxP2P3}
  \begin{equation*}
        {\left. \begin{array}{rl} {P(\theta)} \equiv {\mathbf{case}}\ M[{q_1}]=
        0\to  & \enskip {R_X(\theta)}[q_1]; {R_Y(\theta)}[q_1], \\
        1\to  &\enskip R_Z(\theta)[q_1]
     \end{array}  \right. }
    \end{equation*}
Let us apply code transformation and compilation. Let CT, CP to denote ``code transformation'' and ``compilation'', and  ``Seq'' and ``Rot'' denote Sequence and Rotation rules resp. 
\begin{eqnarray*}
   \parthetagenNotRelation{\udd{P(\theta)}}\stackrel{\textrm{CT,}\ccase}{=}& {\left. \begin{array}{rl}  {\mathbf{case}}\ M[{q_1}]=
        0\to  & \enskip \parthetagenNotRelation{{R_X(\theta)}[q_1]; \\
        & \qquad {R_Y(\theta)}[q_1]}, \\
        1\to  &\enskip \parthetagenNotRelation{R_Z(\theta)[q_1]}
     \end{array}  \right. }
\end{eqnarray*}     
\begin{eqnarray*} 
     \stackrel{\textrm{CT,Seq+Rot}}{=}& {\left. \begin{array}{rl}  {\mathbf{case}}\ M[{q_1}]=
        0\to  & \enskip ({{R\boldgreek{'}_X(\theta)}[A,q_1]};\\
        &\qquad R_Y(\theta)[q_1])\msquare\\
        & \qquad (R_X(\theta)[q_1];\\
        &\qquad {{R\boldgreek{'}_Y(\theta)}[A,q_1]}), \nonumber\\
        1\to  &\enskip {R\boldgreek{'}_Z(\theta)[A,q_1]}
     \end{array}  \right. }
     \\
        \stackrel{\compl{(\bullet)}}{\longmapsto}& \Big{\{}|{\left. \begin{array}{rl}  {\mathbf{case}}\ M[{q_1}]=
        0\to  & \enskip {{R\boldgreek{'}_X(\theta)}[A,q_1]};\\
        & R_Y(\theta)[q_1], \nonumber\\
        1\to  &\enskip {R\boldgreek{'}_Z(\theta)[A,q_1]}\boldgreek{,}
     \end{array}  \right. }\\
        &{\left. \begin{array}{rl}  {\mathbf{case}}\ M[{q_1}]=
        0\to  & \enskip  R_X(\theta)[q_1];\\
        & {{R\boldgreek{'}_Y(\theta)}[A,q_1]}, \nonumber\\
        1\to  &\cabort.
     \end{array}  \right. }|\Big{\}}
  \end{eqnarray*}
\end{example}

\section{Execution and Resource Analysis}  \label{sec::resource}

In this section we illustrate the execution of the entire differentiation procedure and analyze its resource cost.  Consider any program $P(\vth)\in\augweirdCPV$ and the parameter $\theta$.

\vspace{1mm} \noindent \textbf{Execution.} The first step in differentiation is to apply the code transformation rules (in Section~\ref{sec::CTRules}) to $P(\vth)$ and obtain an  \Macroadditive program $\parthetagenNotRelation{P(\vth)}$. Then one needs to compile $\parthetagenNotRelation{P(\vth)}$ into a multiset $\{|P'_i(\vth)|\}_{i=1}^m$ of normal non-aborting programs $P'_i(\vth)$. The total count of these programs is given by $m=\NumNonAbort{\parthetagenNotRelation{P(\vth)}}$. Note that the above procedure could be done at the compilation time. 

Given any pair of $O$ and $\rho$, the real execution to compute the derivative of $\ObsSemMacro{O}{P(\vth)}{\rho}$ is to approximate the observable semantics $\ObsSemMacro{O}{\parthetagenNotRelation{P(\vth)}}{\rho}$.  By Definition~\ref{defn::ObsSemAncilla}, we need to approximate 
\begin{equation} \label{eqn::cost}
    \sum_{i=1}^m \tr\Big{(}\big{(}
                    Z_A\otimes O\big{)}\sem{P'_i(\vth)}((\ket{\qvec{0}}_{\Anc}\bra{\qvec{0}})\otimes \rho)\Big{)},
\end{equation}
where each term is the observable $Z_A \otimes O$ on the output state of $P_i'(\theta)$ given input state $\rho$ and the ancilla qubit $\ket{0}$. 

To approximate the sum in \eqref{eqn::cost} to precision $\delta$ , one could first treat the sum divided by $m$ as the observable applied on the program that starts with a uniformly random choice of $i$ from $1, \cdots, m$ and then execute $P'_i(\vth)$. By Chernoff bound, one only needs to repeat this procedure $O(m^2/\delta^2)$ times.

\vspace{1mm} \noindent \textbf{Resource count.} 
We are only interested in non-trivial (extra) resource that is something that you wouldn't need if you only run the original program. 
Ancilla qubits count as the non-trivial resource. However, for our scheme, the number of required ancillae is 1 qubit per each parameter. %when applying  $\parthetagenNotRelation{\cdot}$ once. 

The more non-trivial resource is the number of the copies of input state (each copy of the input state is to be prepared from scratch), which is directly related to the number of repetitions in the procedure, which again connects to $m=\NumNonAbort{\parthetagenNotRelation{P(\vth)}}$.
We argue that our code transformation is efficient so that $m$ is reasonably bounded. To that end, we show the relation between $m$ and
a natural quantity defined on the original program $P(\vth)$ (i.e., before applying any $\parthetagenNotRelation{\cdot}$ operator) called the \emph{occurrence count} of the parameter $\theta$. 

\begin{defn}
The ``Occurrence Count for $\theta_j$'' in ${P(\vth)}$, 
denoted $\ResCount{j}{P(\vth)}$, is defined as follows: 
\begin{enumerate}
    \item If $P(\vth)\equiv\cabort[\VE]|\cskip[\VE] | q:=\ket{0}$ ($q\in\VE$), then $\ResCount{j}{P(\vth)}=0$;
    \item  
    $P(\vth)\equiv U(\vth)$: if $U(\vth)$
    trivially uses $\theta_j$, then $\ResCount{j}{P(\vth)}=0$;  
    otherwise $\ResCount{j}{P(\vth)}=1$.
    \item If $P(\vth)\equiv U(\vth)={P_1(\vth);{P_2(\vth)}})$ then $\ResCount{j}{P(\vth)}=$ $\ResCount{j}{P_1(\vth)}$ $+\ResCount{j}{P_2(\vth)}$. 
    \item If ${P(\vth)}\equiv\qifStandardPara$ then $\ResCount{j}{{P(\vth)}} = 
    \max_m \ResCount{j}{{P_m(\vth)}}$.
    \item If ${P(\vth)}\equiv\qwhileTParaStandard$ then $\ResCount{j}{P(\vth)} $ $=T\cdot\ResCount{j}{P_1(\vth)} $. 
\end{enumerate}
\end{defn}

Intuition of the ``Occurrence Count'' definition is clear: it basically counts the number of non-trivial occurrences of $\theta_j$ in the program, treating $\ccase$ as if it is deterministic. 
To see why this is a reasonable quantity, consider the auto-differentiation in the classical case. For any non-trivial variable $v$ (i.e., $v$ has some dependence on the parameter $\theta$), we will compute both $v$ and $\parthetagenNotRelation{v}$ and store them both as variables in the new program. 
Thus, the classical auto-differentiation essentially needs the number of non-trivial occurrences more space and related resources. 
As we argued in the introduction, we cannot directly mimic the classical case due to the no-cloning theorem. 
The extra space requirement in the classical setting turns into the requirement on the extra copies of the input state in the quantum setting. 
Indeed, we can bound $m$ by the occurrence count. 

\begin{prop}\label{prop::NOjisSize}
$\NumNonAbort{\parthetajgenNotRelation{P(\vth)}} \leq \ResCount{j}{P(\vth)}$. 
\end{prop}
\begin{proof} Structural induction. For details, see 
\appdName\iftechrep~\ref{appd::NOjisSize}\else\fi.
\end{proof}

\section{Implementation and Case Study} \label{sec::case}
We have built a compiler (written in OCaml) that implements our code transformation and compilation rules\footnote{Codes are availabe at \url{https://github.com/LibertasSpZ/adcompile}.}.
We use it to train one VQC instance with controls and empirically verify its resource-efficiency on representative VQC instances. 
%All codes and test programs are available in the supplementary materials. 
Complete details can be found in \appdName\iftechrep~\ref{appd::evalPen}\else\fi.
Experiments are performed on a MacBook Pro with a Dual-Core Intel Core i5 Processor clocked at 2.7 GHz, and 8GB of RAM. 

\subsection{Training VQC instances with controls}\label{subsec::P1P2}
Consider a simple classification problem over 4-bit inputs $z = z_1z_2z_3z_4 \in \{0,1\}^4$ with true label given by $f(z) = \neg(z_1 \oplus z_4)$. 
We construct two $4$-qubit VQCs $P_1$ (\textbf{no control}) and $P_2$ (\textbf{with control}) that consists of a single-qubit Pauli X,Y and Z rotation gate on each qubit and compare their performance. 

For parameters $\Gamma = \{\gamma_1,\dots,\gamma_{12}\}$ define the program
\begin{equation*}
  {\left. \begin{array}{rl} {Q(\Gamma)}\equiv 
    &R_X(\gamma_1)[q_1]; R_X(\gamma_2)[q_2]; R_X(\gamma_3)[q_3]; R_X(\gamma_4)[q_4]; \\
  & R_Y(\gamma_5)[q_1]; R_Y(\gamma_6)[q_2]; R_Y(\gamma_7)[q_3]; R_Y(\gamma_8)[q_4]; \\
  & R_Z(\gamma_9)[q_1]; R_Z(\gamma_{10})[q_2]; R_Z(\gamma_{11})[q_3]; R_Z(\gamma_{12})[q_4],
\end{array}  \right. }
\end{equation*}
where $q_1, q_2, q_3, q_4$ refer to 4 qubit registers. Given parameters $\Theta = \{\theta_1,\dots,\theta_{12}\},\Phi = \{\phi_1,\dots,\phi_{12}\}$, define
\begin{align}
 P_1(\Theta,\Phi)  \equiv  &Q(\Theta); Q(\Phi).
\end{align}
Similarly, for parameters $\Theta = \{\theta_1,\dots,\theta_{12}\}$,$\Phi = \{\phi_1,\dots,\phi_{12}\}$, $\Psi = \{\psi_1,\dots,\psi_{12}\}$, define 
\begin{equation}
  {\left. \begin{array}{rl} P_2(\Theta,\Phi,\Psi) \equiv  Q(\Theta);
    {\mathbf{case}}\ M[{q_1}]=
  0\to  & \enskip Q(\Phi) \\
  1\to  &\enskip Q(\Psi).
\end{array}  \right. }
\end{equation}
Note that $P_1$ and $P_2$ execute the same number of gates for each run.
To use $P_i$ to perform the classification or in the training, we first initialize $q_1,q_2,q_3,q_4$ to the classical feature vector $z = z_1z_2z_3z_4$ and then execute $P_i$. The predicted label $y$ is given by measuring the $4^{\mathrm{th}}$ qubit $q_4$ in the 0/1 basis. 

We conduct a supervised learning by minimizing a \emph{loss} function. A natural choice is the average negative log-likelihood which is commonly used in machine learning to evaluate classifiers 
that assign a certain probability to each label since quantum outcomes are probabilistic. However, this loss function is not currently supported by Pennylane. 
Denote the output of the classifier with input $z$ and parameters $\vth$ by $l_{\vth}(z)$. 
To enable a direct comparison, we will treat $l_{\vth}(z)$ as the average value of the labels from probabilistic quantum outcomes, and use the squared loss function as follows: 
\begin{align}
  \mathrm{loss} &= \sum_{z \in \{0,1\}^4} 0.5*(l_{\vth}(z) - f(z))^2.
\end{align}
Note that $\mathrm{loss}$ is a function of $\vth=(\Theta,\Phi)$ (or $\Theta,\Phi.\Psi$). More importantly, for each $z$, $l_{\vth}(z)$ can be represented by the observable semantics of $P_1$(or $P_2$) with observable $\proj 1$. 
Thus, the gradient of $\mathrm{loss}$ can be obtained by using the collection of $\frac{\partial}{\partial\alpha}(P_1)$ for $\alpha \in \Theta,\Phi$ (or $\frac{\partial}{\partial\alpha}(P_2)$ for $\alpha \in \Theta,\Phi, \Psi$). We classically simulate the training procedure with gradient descent. For the training of $P_1$, we use Pennylane for a direct comparison (see \fig{control-separation}). 
After 1000 epochs with some hyperparameters, the loss for $P_1$ (\textbf{no control}) attains a minimum of $0.5$ in less than $100$ epochs and subsequently plateaus. The loss for $P_2$ (\textbf{with control}) continues to decrease and attains a minimum of $0.016$. It demonstrates the advantage of both controls in quantum machine learning and our scheme to handle controls, whereas previous schemes (such as Pennylane due to its quantum-node design~\cite{bergholm2018pennylane}) fail to do so.  

\vspace{-2mm}
\begin{figure}[htbp]
    \centering
    \includegraphics[scale=0.4]{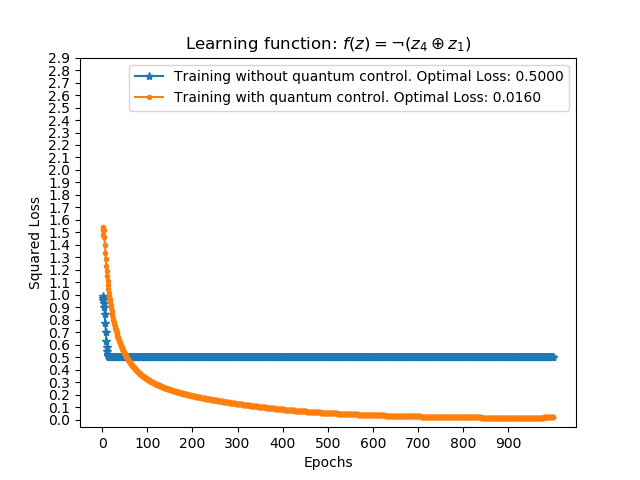}
    \caption{Training $P_1$ and $P_2$ to classify inputs according to the labelling function $f(z) = \neg(z_1 \oplus z_4)$.}
    \label{fig:control-separation}
\end{figure}

\subsection{Benchmark testing on representative VQCs}
We also test our compiler on important VQC candidates such as quantum neural-networks (QNN) for solving machine learning tasks~\cite{EH18}, quantum approximate optimization algorithms (QAOA) for solving combinatorial optimization~\cite{QAOA}, and variational quantum eigensolver (VQE) for approximating ground state energies in quantum chemsitry~\cite{NC-VQE}, all of which are promising candidates for actual implementation on near-term quantum machines. 
These VQCs typically consists of alternating \emph{layers} of single-qubit gates and two-qubit coupling gates, such as the 1-qubit, 2-qubit Pauli rotation gates considered in our paper, to represent the alternation between local interaction and neighboring interaction in real quantum physics systems. 

We enrich these examples, by adding simple controls (the \emph{if/condition} statement) or 2-bounded loops (the bounded-\emph{while} statement) and increasing the number of qubits to 18$\sim$40, to make them sufficiently sophisticated but yet realistic for near-term quantum applications. For example, we use $\QNN_{\mathbf{M,i}}$ to denote an enriched QNN VQC instance of \textbf{m}edium size and with \textbf{i}f controls.  The size of $\QNN_{M,i}$ can also be directly illustrated by the number of qubits ($\numqb$), the gate count ($\numgate$), the number of alternating layers ($\numlayer$), and the number of lines to code such instances ($\numline$). 
Similarly for $\QNN_{\mathbf{L,w}}$ except that it is an instance of \textbf{l}arge size and with \textbf{w}hile controls. 
\begin{table}[b]
    \centering
\begin{tabular}{ |p{1.2cm}|
p{0.81cm}|p{0.91cm}|p{0.84cm}|p{0.87cm}|p{0.87cm}|p{0.87cm}|  }
 \hline
 $P(\vth)$     &
 $\ResCount{}{\cdot}$ &
 $\NumNonAbort{\parthetagenNotRelation{\cdot}}$
 &   $\numgate$ 
 & 
 $\numline$ &
 $\numlayer$
 &$\numqb$ \\
\hline

$\QNN_{M,i}$& 24&	24	&165	&189	&	3	&18\\
\hline

$\QNN_{M,w}$ &56&	24&	231	&121	&	5&	18\\
\hline

$\QNN_{L,i}$&48 &	48&	363	&414	&	6&	36\\
\hline

$\QNN_{L,w}$& 504& 	48	& 2079& 	244	& 	33	&36\\
\hline

$\VQE_{M,i}$ & 15 & 	15	& 224	& 241 &		3	& 12\\
\hline 

$\VQE_{M,w}$ & 35 & 	15 &	224 &	112	&	5	& 12\\
\hline 

$\VQE_{L,i}$ & 40	 & 40	& 576&	628	&	5	&40\\
\hline 

$\VQE_{L,w}$ & 248	& 40 & 1984	& 368	&	17	& 40\\
\hline

$\QAOA_{M,i}$ & 18& 	18& 	120&	142	&	3&	18\\
\hline 

$\QAOA_{M,w}$ & 42	& 18& 	168	& 94	&	5	& 18\\
\hline 

$\QAOA_{L,i}$ & 36	& 36 &	264	 & 315	& 	6 &	36\\
\hline 

$\QAOA_{L,w}$&  378	 & 36	& 1512	& 190	&	33 &	36\\
\hline 

\end{tabular}
 \caption{Output on selective examples. $\{M,L\}$ stands for ``medium, large''; $\{i,w\}$ stands for including ``if, while''.} 

    \label{tab:my_little_tabel}
\end{table}

A selective output performance of our compiler is in Table~\ref{tab:my_little_tabel}, with details in \appdName\iftechrep~\ref{appd::evalPen}\else\fi.
It is easy to see that our scheme is also empirically resource-efficient as $\NumNonAbort{\parthetagenNotRelation{\cdot}}$ is always reasonably bounded.

\bibliography{references.bib}

\ifsubmit
{}
\else 
\clearpage 

\appendix
\section{Detailed Quantum Preliminary}\label{sec::qprel}

This is a more detailed treatment of \sec{prelim}. For a further extended background, we recommend the notes
by~\citet{Wat06} and the textbook by~\citet{MI2002}.    

\subsection{Preliminaries}

For any 
non-negative integer $n$, an $n$-dimensional Hilbert space $\H$
is essentially the space $\mathbb{C}^n$ of complex vectors.
We use Dirac's notation, $\ket{\psi}$, to denote a complex vector in $\mathbb{C}^n$. The inner product of two vectors $\ket{\psi}$ and $\ket{\phi}$ is denoted by $\langle\psi|\phi\rangle$,
which is the product of the Hermitian conjugate of $\ket{\psi}$, denoted by $\bra{\psi}$, and vector $\ket{\phi}$.
The norm of a vector $\ket{\psi}$ is denoted by $\nm{\ket{\psi}}=\sqrt{\langle\psi|\psi\rangle}$.

We define (linear) \emph{operators} as linear mappings between Hilbert spaces.
Operators between $n$-dimensional Hilbert spaces are represented by $n\times n$ matrices.
For example, the identity operator $I_\H$ can be identified by the identity matrix on $\H$.
The Hermitian conjugate of operator $A$ is denoted by $A^\dag$. Operator $A$ is \emph{Hermitian} if $A=A^\dag$.
The trace of an operator $A$ 
is the sum of the entries on the main diagonal, i.e., $\tr(A)=\sum_i A_{ii}$. 
We write $\bra{\psi}A\ket{\psi}$ to mean the inner product between
$\ket{\psi}$ and $A\ket{\psi}$.
A Hermitian operator $A$ is \emph{positive semidefinite} (resp.,
\emph{positive definite}) if for all vectors $\ket{\psi}\in\H$,
$\bra{\psi}A\ket{\psi}\geq 0$ (resp., $>0$).
This gives rise to the \emph{L\"owner order} $\sqsubseteq$ among operators:

$\begin{array}{rl}
  A\sqsubseteq B   & \text{ if } B-A \text{ is positive semidefinite, } \\
  A\sqsubset B   & \text{ if } B-A \text{ is positive definite. }
\end{array}$

\subsection{Quantum States}

The state space of a quantum system is a Hilbert space. 
The state space of a \emph{qubit}, or quantum bit, is a 2-dimensional Hilbert space.
One important orthonormal basis of a qubit system is the \emph{computational} basis with $\ket{0}=(1,0)^\dag$ and $\ket{1}=(0,1)^\dag$, which encode the classical bits 0 and 1 respectively.  
Another important basis, called the $\pm$ basis, consists of $\ket{+}=\frac{1}{\sqrt{2}}(\ket{0}+\ket{1})$ and $\ket{-}=\frac{1}{\sqrt{2}}(\ket{0}-\ket{1})$.
The state space of multiple qubits is the \emph{tensor product} of single qubit state spaces.
For example, classical 00 can be encoded by $\ket{0}\otimes\ket{0}$
(written $\ket{0}\ket{0}$ or even $\ket{00}$ for short) in the Hilbert space $\mathbb{C}^2\otimes\mathbb{C}^2$. An important 2-qubit state is the EPR state $\ket{\beta_{00}}=\frac{1}{\sqrt{2}}(\ket{00}+\ket{11})$.
The Hilbert space for an $m$-qubit system is $(\mathbb{C}^2)^{\otimes m} \cong \mathbb{C}^{2^m}$.

A \emph{pure} quantum state is represented by a unit vector, i.e., a vector $\ket{\psi}$ with $\nm{\ket{\psi}}=1$.
A \emph{mixed} state can be represented by a classical distribution over an ensemble of pure states $\{(p_i,\ket{\psi_i})\}_i$,
i.e., the system is in state $\ket{\psi_i}$ with probability $p_i$.
One can also use \emph{density operators} to represent both pure and mixed quantum states.
A density operator $\rho$ for a mixed state representing the ensemble $\{(p_i,\ket{\psi_i})\}_i$ is a positive semidefinite operator $\rho=\sum_i p_i\ket{\psi_i}\bra{\psi_i}$, where $\ket{\psi_i}\bra{\psi_i}$ is the outer-product of $\ket{\psi_i}$; in particular, a pure state $\ket{\psi}$ can be identified with the density operator $\rho=\ket{\psi}\bra{\psi}$.
Note that $\tr(\rho)=1$ holds for all density operators. A positive semidefinite operator $\rho$ on $\H$ is said to be a \emph{partial} density operator if $\tr(\rho)\leq 1$.
The set of partial density operators is denoted by $\D(\H)$.

\subsection{Quantum Operations} 

Operations on quantum systems can be characterized by unitary operators. Denoting the set of linear operators on $\H$ as $L(\H)$, an operator $U\in L(\H)$ is \emph{unitary} if its Hermitian conjugate is its own inverse, i.e., $U^\dag U=UU^\dag=I_\H$. For a pure state $\ket{\psi}$, a unitary operator describes an \emph{evolution} from $\ket{\psi}$ to $U\ket{\psi}$. For a density operator $\rho$, the corresponding evolution is $\rho \mapsto U\rho U^\dag$. Common 
unitary operators include
\begin{align*}
  H=
  \left[ {\begin{array}{cc}
   \frac{1}{\sqrt{2}}& \frac{1}{\sqrt{2}} \\
   \frac{1}{\sqrt{2}}& \frac{-1}{\sqrt{2}}\\
  \end{array} } \right],
  \
  X=
  \left[ {\begin{array}{cc}
   0 & 1 \\
   1 & 0 \\
  \end{array} } \right],
  \
  Z=
  \left[ {\begin{array}{cc}
   1 & 0 \\
   0 & -1 \\
  \end{array} } \right],\
  Y=-iXZ
\end{align*}

The \emph{Hadamard} operator $H$ transforms between the computational and the $\pm$ basis. For example, $H\ket{0}=\ket{+}$ and $H\ket{1}=\ket{-}$.
The \emph{Pauli $X$} operator is a bit flip, i.e., $X\ket{0}=\ket{1}$ and  $X\ket{1}=\ket{0}$. The \emph{Pauli $Z$} operator is a phase flip, i.e., $Z\ket{0}=\ket{0}$ and $Z\ket{1}=- \ket{1}$. \emph{Pauli $Y$} 
maps $\ket{0}$ to $i\ket{1}$ and $\ket{1}$ to $-i\ket{0}$. The \emph{CNOT} gate $C$ maps $\ket{00}\mapsto \ket{00},\ket{01}\mapsto \ket{01},\ket{10}\mapsto \ket{11},\ket{11}\mapsto \ket{10}$. 
One may obtain the EPR state $\ket{\beta_{00}}$ via $\ket{00}\stackrel{H_1}{\mapsto}\frac{1}{\sqrt{2}}(\ket{0}+\ket{1})\ket{0}\stackrel{C_{1,2}}{\mapsto}\frac{1}{\sqrt{2}}(\ket{00}+\ket{11})$.

More generally, the evolution of a quantum system can be characterized by an \emph{admissible superoperator} $\E$, which is a \emph{completely-positive} and \emph{trace-non-increasing} linear map from $\D(\H)$ to $\D(\H')$ for Hilbert spaces $\H, \H'$.
A superoperator is positive if it maps from $\D(\H)$ to $\D(\H')$ for Hilbert spaces $\H, \H'$.
A superoperator $\E$ is $k$-positive if for any $k$-dimensional Hilbert space $\A$, the superoperator $\E\otimes I_\A$ is a positive map on $\D(\H\otimes\A)$. 
A superoperator is said to be completely positive if it is $k$-positive for any positive integer $k$.
A superoperator $\E$ is trace-non-increasing if for any initial state $\rho\in\D(\H)$, the final state $\E(\rho)\in \D(\H')$ after applying $\E$ satisfies $\tr(\E(\rho))\leq\tr(\rho)$.

For every superoperator $\E : \D(\H)\to\D(\H')$, there exists a set of Kraus operators $\{E_k\}_k$ such that $\E(\rho)=\sum_k E_k\rho E_k^\dag$ for any input $\rho\in\D(\H)$.
Note that the set of Kraus operators is finite if the Hilbert space is finite-dimensional.
The \emph{Kraus form} of $\E$ is written as $\E=\sum_k E_k\circ E_k^\dag$.
A unitary evolution can be represented by the superoperator $\E=U\circ U^\dag$. An identity operation refers to the superoperator $\mathcal{I}_{\H} = I_{\H} \circ I_{\H}$.
The Schr\"odinger-Heisenberg \emph{dual} of a superoperator $\E=\sum_k E_k\circ E_k^\dag$, denoted by $\E^*$, is defined as follows: for every state $\rho\in\D(\H)$ and any operator $A$, $\tr(A\E(\rho))=\tr(\E^*(A)\rho)$. The Kraus form of $\E^*$ is $\sum_k E_k^\dag\circ E_k$. 

\subsection{Quantum Measurements and Observables} 

The way to extract information about a quantum system is called
a quantum \emph{measurement}. 
A quantum measurement on a system over Hilbert space $\H$ can be
described by a set of linear operators 
$\{M_m\}_m$ with $\sum_m M_m^\dag M_m=I_\H$. 
If we perform a measurement $\{M_m\}_m$ on a state $\rho$, the outcome $m$ is observed with probability $p_m=\tr(M_m\rho M_m^\dag)$ for each $m$.
A major difference between classical and quantum computation is that a
quantum measurement changes the state. In particular, after a
measurement yielding outcome $m$, the state collapses to $M_m\rho M_m^\dag/p_m$.
For example, a measurement in the computational basis is described by $M=\{M_0=\ket{0}\bra{0}, M_1=\ket{1}\bra{1}\}$.
If we perform the computational basis measurement $M$ on state $\rho=\ket{+}\bra{+}$, then with probability $\frac{1}{2}$ the outcome is $0$ and $\rho$ becomes $\ket{0}\bra{0}$.
With probability $\frac{1}{2}$ the outcome is $1$ and $\rho$ becomes $\ket{1}\bra{1}$.

\section{More on the definition of parameterized quantum programs}\label{sec::Appendix}

\subsection{Definition of qVar}\label{appd::qVarA}

Given $P(\vth)$ a $T$-bounded $k$-parameterized quantum \textbf{while}-program, 
let $\mathit{qVar}(P(\vth))$, read \emph{the set of quantum variables accesible to $P$, be recursively defined as follows}~\cite{Ying16}:
\begin{enumerate}
    \item If $P(\vth)\equiv \cskip[\qvec{q}],\cabort[\qvec{q}]$ or $U[\qvec{q}]$, then $\mathit{qVar}(P(\vth)) = \overline{q}$. 
    \item If $P(\vth)\equiv q:=\ket{{0}}$, then $\mathit{qVar}(P(\vth)) = {q}$.
    \item If $P(\vth)\equiv P_1(\vth);P_2(\vth)$, then $\mathit{qVar}(P(\vth)) =\mathit{qVar}(P_1(\vth))\cup \mathit{qVar}(P_2(\vth))$. When analyzing $P_1(\vth);P_2(\vth)$, we identify $P_1(\vth)$ with $\calI\otimes P_1(\vth)$, where $\calI\equiv I_{\mathit{qVar}(P_2(\vth))\setminus\mathit{qVar}(P_1(\vth))}\circ I_{\mathit{qVar}(P_2(\vth))\setminus\mathit{qVar}(P_1(\vth))}$, the identity operation on the variables where $P_1$ originally has no access to. 
    
    \item If $P(\vth)\equiv \qif{M[\qvec{q}]=\overline{m\to P_{m}(\vth)}}
              $, then $\mathit{qVar}(P(\vth)) = \qvec{q}\cup\bigcup_m \mathit{qVar}(P_i(\vth))$. We make the same 
              identification for $P_i(\vth)$'s as the above.
    \item If $P(\vth)\equiv \qwhileTParaStandard 
    $, then $\mathit{qVar}(P(\vth)) = \qvec{q}\cup\ \mathit{qVar}(P_1(\vth))$. We make the same identification for $P_1(\vth)$ 
    as the above.
\end{enumerate}

One defines $\mathit{qVar}(P)$ for unparameterized $P$ analogously. 
\section{Detailed Proof from Section \ref{sec::detQCHs}}

\subsection{Proof of Prop \ref{prop::WellDefined}: Non-deterministic Compilation Rules Well-Defined 
}\label{proof::WellDefined}
\begin{proof}Structural induction.
\begin{enumerate}
    \item (Atomic)
    : as operational semantics of atomic operations are inherited from parameterized quantum while programs, they do not induce non-determinism. For these operations, $\sem{\udd{P(\vth^*)}}\rho=\{|\sem{P(\vth^*)}\rho|\}$ as is the right hand side. 
    \item (Sequence) If one of $\compl{(\udd{P_b(\vth^*)})}$ is $\{\cabort\}$ the statement is immediately true, so we assume otherwise for below.
    \begin{enumerate}
        \item $\subseteq$: Let 
        $\rho'\in 
        \{|\rho'\neq\mathbf{0}:\lag \udd{P_1(\vth^*);P_2(\vth^*)},\rho\rag\to^*\lag\downarrow,\rho'\rag |\}$. By definition there exists $\rho_2\neq \mathbf{0}$ s.t. $\lag \udd{P_1(\vth^*)};$ $\udd{P_2(\vth^*)},\rho\rag\to^*\lag P_2(\vth^*),\rho_2\rag$ and $\lag \udd{P_2(\vth^*)},\rho_2\rag\to^*\lag\downarrow,\rho'\rag$. By inductive hypothesis, $\rho'\in 
        \coprod_{Q_j(\vth)\in\compl{(\udd{P_2(\vth)})}
        }$ $\{|\rho'\neq \mathbf{0}:\lag{Q_j(\vth^*)},\rho_2\rag\to^* \lag \downarrow, \rho'\rag |\}$.

        Since $\lag \udd{P_1(\vth^*);P_2(\vth^*)},\rho\rag\to^*\lag P_2(\vth^*),\rho_2\rag$, there must be some $\lag P_3(\vth^*), \rho_3\rag\to \lag \downarrow,\rho_2\rag$ triggering the last transition according to our operational semantics. Inductively apply such an argument, knowing that all computation paths are finitely long,\footnote{We don't count ``transitions'' like $\lag P(\vth^*),\rho\rag\to \lag P(\vth^*),\rho\rag$ when analyzing computation paths, as nothing evolves in these ``trivial transitions''.} 
        we conclude that $\rho_2\in \{|\rho'_2:\lag \udd{P_1(\vth)},\rho\rag\to^* \lag \downarrow,\rho'_2\rag|\}$. 
        By the inductive hypothesis  
        we have $\rho_2\in \coprod_{R_i(\vth)\in\compl{(\udd{P_1(\vth)})}
        }$  $\{|\rho'_2\neq 0:\lag{R_i(\vth^*)},\rho\rag\to^* \lag \downarrow, \rho'_2\rag |\}
        $. Thus
        \begin{eqnarray}
        \rho'&\in&  \coprod_{R_i(\vth)\in \compl{(\udd{P_1(\vth)})}, Q_j(\vth)\in \compl{(\udd{P_2(\vth)})}
        }\\
        &&\{|\rho'\neq 0:
        \lag{R_i(\vth^*)},\rho\rag\to^* \lag \downarrow, \rho_2\rag\\&&\qquad\qquad\qquad\qquad\qquad\qquad\qquad \\
        &&\bigwedge \lag{Q_j(\vth^*)},\rho_2\rag\to^* \lag \downarrow, \rho'\rag|\}
        \nonumber
        \\ &=& \coprod_{R_i(\vth)\in \compl{(\udd{P_1(\vth)})}, Q_j(\vth)\in \compl{(\udd{P_2(\vth)})}
        }\\
        &&\{|\rho'\neq 0:\lag{R_i(\vth^*);Q_j(\vth^*)},\rho\rag\to^* \lag\downarrow, \rho'\rag |\}\nonumber\\&& \label{eqn::A15}
        \end{eqnarray}
        as desired.  
        \item $\supseteq$:  
        Let $\rho'$ be a member of the multiset in right hand side (``RHS'') of $\ref{eqn::A15}$ 
        for some $ R_i(\vth)\in \compl{(\udd{P_1(\vth)})},$ $ Q_j(\vth)\in \compl{(\udd{P_2(\vth)})}$. Then by inductive hypothesis, $\lag \udd{P_2(\vth)},$ $\sem{R_i(\vth^*)}\rho\rag$ $\to^* \lag\downarrow, \rho'\rag$ ($\rho'\neq\mathbf{0}$), and $\lag \udd{P_1(\vth)},\rho\rag$ $\to^* \lag\downarrow, \sem{R_i(\vth^*)}\rho\rag$ ($\sem{R_i(\vth^*)}\rho\neq \mathbf{0}$). One may start with $\lag \udd{P_2(\vth^*)},\sem{R_i(\vth^*)}\rho\rag$ $\to^* \lag\downarrow, \rho'\rag$ then repeatedly apply the Sequential rule in the operational semantics, tracing back the path $\lag \udd{P_1(\vth)},\rho\rag$ $\to^* \lag\downarrow, \sem{R_i(\vth^*)}\rho\rag$, 
        and conclude that $\lag \udd{P_1(\vth^*)};\udd{P_2(\vth^*)},\rho\rag\to^*\lag \downarrow, \rho'\rag$ ($\rho'\neq 0$). 
    \end{enumerate}
    The above said the two multisets have the same set of distinct elements, and furthermore each copy of $\rho'\in \sem{\udd{P(\vth^*)}}\rho$ induces at least (by $\subseteq$, or by the deterministic nature of $R_i(\vth^*);Q_j(\vth^*)$) as well as at most (by $\supseteq$
    ) one copy of the same 
    element in RHS$(\ref{eqn::A15})$. This says the two multisets are equal.

    \item (Case) Applying the operational transition rule for $1$ step one may observe that (writing ``IH'' the inductive hypothesis) 
    \begin{eqnarray}
       && \{| \mathbf{0}\neq \rho'\in \sem{\uqifParaStandard}\rho|\}\nonumber\\ &=& \coprod_{m^*}\{|\rho'\neq 0: \lag P_m(\vth^*),M_m\rho M\dagg_m\rag\to \lag\downarrow, \rho'\rag|\}\nonumber\\
        &\stackrel{IH}{=}&\coprod_{m^*}\coprod_{Q_{m^*,i_{m^*}}(\vth)
        \in\compl{(\udd{P_{m^*}(\vth)})}}\\&&\{|\rho'\neq 0:\lag\sem{Q_{m^*,i_{m^*}}(\vth^*)},M_{m^*}\rho M\dagg_{m^*}\rag\\
        && \qquad\quad\ \qquad\qquad\qquad\qquad\qquad\to^*\lag \downarrow, \rho'\rag|\}\nonumber
    \end{eqnarray}
    
    On the other hand we have 
\begin{eqnarray}
&& 
\coprod_{Q(\vth)\in\compl{(\udd{\uqifParaStandard})}
 }\nonumber\\&& \{|\rho'\neq \mathbf{0}:\lag{Q(\vth^*)},\rho\rag\to^* \lag \downarrow, \rho'\rag |\}\nonumber\\ 
         &\stackrel{\textrm{CP,Case}}{=}&
         \coprod_{Q(\vth)\in\Facebook{\udd{\uqifParaStandard}}
 }\nonumber\\&& \{|\rho'\neq \mathbf{0}:\lag{Q(\vth^*)},\rho\rag\to^* \lag \downarrow, \rho'\rag |\}\nonumber\\
         &\stackrel{(**)}{=}&\coprod_{m^*}\coprod_{Q_{m^*,i_{m^*}}(\vth)
        \in\compl{(\udd{P_{m^*}(\vth)})}}\{|\rho'\neq \mathbf{0}:\lag\sem{Q_{m^*,i_{m^*}}(\vth^*)},\nonumber\\&&M_{m^*}\rho M\dagg_{m^*}\rag
        \to^*\lag \downarrow, \rho'\rag|\}\label{eqn::A20}
\end{eqnarray}
as 
desired.  Note that the last step (**) is due to the behavior of the 
superoperators in $\Facebook{\uqifParaStandard}$: when evolved by one transition, they either go to some \emph{non-}\emph{(essentially}\emph{-aborting)} \newline $ \lag\sem{Q_{m^*,i_{m^*}}(\vth^*)},$  $M_{m^*}\rho M\dagg_{m^*}\rag$ for some $Q_{m^*,i_{m^*}}(\vth)\in\newline  \compl{(P_{m^*}(\vth))}$ for some $m^*$, or go to $\lag\cabort,M_{m^*}\rho M\dagg_{m^*} \rag$. Besides, each copy of final state of non-(essentially-aborting) 

$\lag\sem{Q_{m^*,i_{m^*}}(\vth^*)},$ $M_{m^*}\rho M\dagg_{m^*}\rag$ appears exactly once in the multiset, Due to our construction of $\Facebook{\uqifParaStandard}$. On the other hand, computational paths starting with \emph{essentially aborting} $\lag Q_{m^*,i_{m^*}}(\vth^*),$  $M_{m^*}\rho M\dagg_{m^*}\rag$ always terminate in $\lag\downarrow, \mathbf{0}\rag$, 
therefore not counted on either side of $(**)$.
    \item (While$^{(T)}$) By definition, While$^{(T)}$ reduces to (Case) and (Sequence).
    \item (Sum Components) If one of $\compl{(\udd{P_b(\vth^*)})}$ is $\{\cabort\}$ the statement is immediately true, so we assume otherwise for below. Like in the $\ccase$ case, 
    applying the operational transition rule for $1$ step one may observe that (writing ``IH'' the inductive hypothesis)
    \begin{eqnarray}
     &&   \mathbf{0}\neq\rho'\in \sem{\udd{P_1(\vth^*)}\msquare\udd{P_2(\vth^*)}}\rho \\&=& \coprod_{m^*=1,2}\{|\rho'\neq \mathbf{0}: \lag P_m(\vth^*),\rho \rag\to \lag\downarrow, \rho'\rag|\}\nonumber\\
        &\stackrel{IH}{=}&\coprod_{m^*=1,2}\{|\rho'\neq \mathbf{0}:\lag\sem{Q_{m^*,i_{m^*}}(\vth^*)},\rho\rag \to^*\nonumber\\
        &&\lag\downarrow,\rho'\rag, Q_{m^*,i_{m^*}}(\vth)\in\compl{(\udd{P_{m^*}(\vth)})}|\}\nonumber\\
       &\stackrel{\textrm{CP,Sum Component}}{=}&\coprod_{Q(\vth)\in\compl{(\udd{P_1(\vth)}\msquare \udd{P_2(\vth)})}}
      \nonumber \\&&\{|\rho'\neq 0: \lag Q(\vth^*),\rho\rag\to^*\lag \downarrow,\rho'\rag|\} ,
    \end{eqnarray}
   as desired!

\end{enumerate}

\end{proof}

\section{Detailed Proofs from Section \ref{sec::CTRules}} \label{app::CTRules}

Throughout, we use LHS, RHS to denote ``left hand side'' and ``right hand side'' resp., and ``IH'' to denote ``Inductive Hypothesis''. Wherever applicable, we adopt the overloading convention explained in the main text (See Eqn \ref{eqn::OverLoadingRefByAppendix}). 

Before giving proofs, we record the full details for code transformation rule for $\textrm{while}^{(T)}$ for your interest. Let us denote $\parthetagenNotRelation{\uqwhileTParaStandardS} \equiv $ $\mathbf{Seq}^{T}$, with: 

\begin{equation*}
        {\left. \begin{array}{rl}\quad \mathbf{Seq}^{(1)}\ \equiv\ \udd{\mathbf{case}}\ M[\overline{q}]=
        0\to  & \enskip \udd{\cabort}, \\
        1\to  &\enskip \Big{(}{\parthetajgenNotRelation{\udd{P_1(\vth)}}};
        \udd{\cabort}\Big{)}\msquare\\
        &\Big{(}{{\udd{P_1(\vth)}}};
        \udd{\cabort}\Big{),}
     \end{array}  \right. }
    \end{equation*}

and $\mathbf{Seq}^{(T\geq 2)}$ recursively defined as: 

   \begin{equation*}
        {\left. \begin{array}{rl}\mathbf{Seq}^{(T\geq 2)}\ \equiv\ \udd{\mathbf{case}}\ M[\overline{q}]=
        0\to  & \enskip \udd{\cabort}, \\
        1\to  &\enskip \Big{(}{\parthetajgenNotRelation{\udd{P_1}(\vth)}};\rwhileTminus\Big{)}\msquare\\
        &\Big{(}{{\udd{P_1}(\vth)}}; \mathbf{Seq}^{(T-1)}
        \Big{)}\\
     \end{array}  \right. }
    \end{equation*}
Note that $\mathbf{Seq}^{(T-1)}$ essentially aborts. Let us 
now introduce some helper lemmas for the soundness proof:
\subsection{Technical Lemmas}\label{sec:TechLem}
\begin{lemma}\label{lem:UnitaryDerive}
$U(\theta)\in \{R_\sigma(\theta),R_{\sigma\otimes\sigma}(\theta)\}_{\sigma\in\{X,Y,Z\}}$. 
Let $\entrideri U(\theta)$ denote the entry-wise derivative of $U(\theta)$. Then,
    \begin{eqnarray}
    \entrideri U(\theta)&=&\frac{1}{2}U(\theta+\pi);\label{lemmaEquation:entriwise}\\
   && \parthetagenNotRelation{\ObsSemMacro{O}{{\udd{U(\theta)}}}{\rho} }\nonumber\\ &=& \tr(O\cdot U(\theta)\cdot \rho\cdot (\entrideri U(\theta))\dagg) +\nonumber\\&&\tr(O\cdot \entrideri U(\theta)\cdot \rho\cdot U\dagg(\theta)) \label{LemmaEquation:uudag}\\
    &\stackrel{\ref{lemmaEquation:entriwise}}{=}&\frac{1}{2}\tr\bigg{(}O\Big{(}U(\theta)\rho U\dagg(\theta+\pi)+\nonumber\\&&U(\theta+\pi)\rho U\dagg(\theta)\Big{)}\bigg{)}.\label{LemmaEqn:UUDagContinued}
    \end{eqnarray}
\end{lemma}
\begin{proof}
Let $U(\theta)\in \{R_\sigma(\theta),R_{\sigma\otimes\sigma}(\theta)\}_{\sigma\in \{X,Y,Z\}}$. Then $U(\theta)=\cos{\frac{\theta}{2}}I_{\bullet}-i\sin(\frac{\theta}{2})\boldgreek{\sigma}$, where $\boldgreek{\sigma}\in\{X,Y,Z,X\otimes X,Y\otimes Y, Z\otimes Z\}$. In the cases where $\boldgreek{\sigma}$ is $1$-qubit gate, $I_\bullet$ denotes identity on that one qubit; likewise for 2-qubit cases. Correctness of Eqns \ref{lemmaEquation:entriwise}, \ref{LemmaEquation:uudag} basically follows from straightforward computation.

\begin{enumerate}
    \item  Equation \ref{lemmaEquation:entriwise}:
    \begin{eqnarray}
 && \entrideri(U(\theta))  \nonumber\\&=&\frac{1}{2}\BPp{(-\sin(\frac{\theta}{2}))I_\bullet - i\cos(\frac{\theta}{2})\boldgreek{\sigma}}\\
  &=&\frac{1}{2}\BPp{\cos(\frac{\theta}{2}+\frac{\pi}{2})I_\bullet-i\sin(\frac{\theta}{2}+\frac{\pi}{2})\boldgreek{\sigma}}\\
  &=&\frac{1}{2}\BPp{\cos(\frac{(\theta+\pi)}{2})I_\bullet -i\sin(\frac{(\theta+\pi)}{2})\boldgreek{\sigma}}\\
  &=&\frac{1}{2}U(\theta+\pi)
    \end{eqnarray}
    \item Equation \ref{LemmaEquation:uudag}:
    \begin{eqnarray}
   && \parthetagenNotRelation{\ObsSemMacro{O}{{\udd{U(\theta)}}}{\rho}}\nonumber\\ &=& \parthetagenNotRelation{\tr(OU(\theta) \rho U\dagg(\theta))}\\
    &=&\parthetagenNotRelation{\tr\Big{(}O\cdot (\cos{\frac{\theta}{2}}I_{\bullet}-i\sin(\frac{\theta}{2})\boldgreek{\sigma})\cdot \rho \cdot\nonumber\\&& (\cos{\frac{\theta}{2}}I_{\bullet}+i\sin(\frac{\theta}{2})\boldgreek{\sigma}\dagg)\Big{)}}\\
    &=&\frac{\partial}{\partial\theta}\Big{(}\tr(\cos^2(\frac{\theta}{2})O\rho + i\cos(\frac{\theta}{2})\sin(\frac{\theta}{2})O\rho\boldgreek{\sigma}\dagg\\
    &&-i\sin(\frac{\theta}{2})\cos(\frac{\theta}{2})O\boldgreek{\sigma}\rho +\sin^2(\frac{\theta}{2})\boldgreek{\sigma}\rho\boldgreek{\sigma}\dagg)\Big{)}\\
    &=&\tr(2\cos(\frac{\theta}{2})\cdot (\cos(\frac{\theta}{2}))'O\rho) \\
    &&+ \tr(i[\cos(\frac{\theta}{2})(\sin(\frac{\theta}{2}))'+(\cos(\frac{\theta}{2}))'\sin(\frac{\theta}{2})]O(\rho\boldgreek{\sigma}\dagg\nonumber-\boldgreek{\sigma}\rho))\\ &&+\tr(2\sin(\frac{\theta}{2})\cdot(\sin(\frac{\theta}{2}))'O\boldgreek{\sigma}\rho\boldgreek{\sigma}\dagg),
    \end{eqnarray}
On the other hand,
\begin{eqnarray}
&&\tr(O\cdot U(\theta)\cdot \rho\cdot (\entrideri U(\theta))\dagg) +\nonumber\\&& \tr(O\cdot \entrideri U(\theta)\cdot \rho\cdot U\dagg(\theta))\\
&=&\tr\Big{(}O\cdot (\cos{\frac{\theta}{2}}I_{\bullet}-i\sin(\frac{\theta}{2})\boldgreek{\sigma})\cdot \rho\cdot \nonumber\\
&&((\cos{\frac{\theta}{2}})'I_{\bullet}+i(\sin(\frac{\theta}{2}))'\boldgreek{\sigma}\dagg)\Big{)}+\nonumber\\
&&\tr\Big{(}O\cdot ((\cos{\frac{\theta}{2}})'I_{\bullet}-i(\sin(\frac{\theta}{2}))'\boldgreek{\sigma})\cdot \rho\cdot\nonumber\\&& (\cos{\frac{\theta}{2}}I_{\bullet}+i\sin(\frac{\theta}{2})\boldgreek{\sigma}\dagg)\Big{)}\\
&=&\tr(2\cos(\frac{\theta}{2})\cdot (\cos(\frac{\theta}{2}))'O\rho) \\
    &&+ \tr(i[\cos(\frac{\theta}{2})(\sin(\frac{\theta}{2}))'+(\cos(\frac{\theta}{2}))'\sin(\frac{\theta}{2})]O(\rho\boldgreek{\sigma}\dagg\nonumber-\boldgreek{\sigma}\rho))\\ &&+\tr(2\sin(\frac{\theta}{2})\cdot(\sin(\frac{\theta}{2}))'O\boldgreek{\sigma}\rho\boldgreek{\sigma}\dagg),
\end{eqnarray}
as desired. 
\end{enumerate}

\end{proof}

\begin{lemma}\label{lem::propertyAncillaObservableSemantics} Let $P(\vth)\in\augweirdCPV, P
{'}(\vth)\in$\newline  $\augweirdCPVa$. Then for arbitrary 
$\vth^*\in\RR^k$, $\rho\in\weirdS$, $O\in\weMV$,
\begin{enumerate}
    \item $\ObsSemMacro{O}{P(\vth^*);P
    {'}(\vth^*)}{\rho}=\ObsSemMacro{O}{P
    {'}(\vth^*)}{\sem{P(\vth^*)}\rho}$.
    \item $\ObsSemMacro{O}{P
    {'}(\vth^*);P(\vth^*)}{\rho}=\ObsSemMacro{\sem{P(\vth^*)}^*(O)}{P
    {'}(\vth^*)}{\rho}$, with $\sem{P(\vth^*)}^*$ 
    the dual 
    of $\sem{P(\vth^*)}.$
\end{enumerate}
\end{lemma}
\begin{proof}
Observe that both $P(\vth^*);P
{'}(\vth^*)$ and $P
{'}(\vth^*);P(\vth^*)$ lives in $\augweirdCPVa$, so we identify the ``smaller'' program $P(\vth^*)$ with $\calI_A\otimes P(\vth^*)$, 
where $\calI\equiv I_A\circ I_A$ denotes the identity operation on Ancilla.
\begin{enumerate}
    \item Unfolding definition \ref{defn::ObsSemAncilla},
    \begin{eqnarray*}
      & \ObsSemMacro{O}{P(\vth^*);P
      {'}(\vth^*)}{\rho}\\=&\tr((Z_A\otimes O)\sem{P(\vth^*);P
      {'}(\vth^*)}((\ket{0}_A\bra{0})\otimes\rho))\\
       \stackrel{\textrm{Fig}\ref{fig:desem:pardet},\textrm{ Seq.}}{=}& \tr((Z_A\otimes O)\sem{P
       {'}(\vth^*)}\sem {\calI_A\otimes P(\vth^*)}((\ket{0}_A\bra{0})\otimes\rho))\nonumber\\
       {=}&\tr((Z_A\otimes O)\sem{P
       {'}(\vth^*)}((\ket{0}_A\bra{0})\otimes(\sem {P(\vth^*)}\rho)))\\
       =&\ObsSemMacro{O}{P
       {'}(\vth^*)}{\sem{P(\vth^*)}\rho},\textrm{ as desired. }
    \end{eqnarray*}
    \item Observe that:
    \begin{eqnarray}
       &&\ObsSemMacro{O}{P
       {'}(\vth^*);P(\vth^*)}{\rho} \nonumber\\&=&\tr((Z_A\otimes O)\sem{P
       {'}(\vth^*);P(\vth^*)}\cdot\nonumber\\
       &&\Big{(}(\ket{0}_A\bra{0})\otimes\rho)\Big{)}\nonumber\\
       &\stackrel{\textrm{Fig}\ref{fig:desem:pardet},\textrm{ Sequence}}{=}& \tr((Z_A\otimes O)\sem{\calI_A\otimes P(\vth^*)}\sem{P
       {'}(\vth^*)}\cdot\nonumber\\&&\Big{(}(\ket{0}_A\bra{0})\otimes\rho)\Big{)}\nonumber\\
       &\stackrel{\zeta\equiv \sem{P
       {'}(\vth^*)}((\ket{0}_A\bra{0})\otimes\rho)}{=}&\tr((Z_A\otimes O)\sem{\calI_A\otimes P(\vth^*)}\zeta)\label{eqn::RHSdual1},
       \end{eqnarray}
       {while}
       \begin{eqnarray}
       &&\ObsSemMacro{\sem{P(\vth^*)}^*(O)}{P
       {'}(\vth^*)}{\rho}\\&=& \tr((Z_A\otimes \sem{P(\vth^*)}^*(O))\sem{P
       {'}(\vth^*)}((\ket{0}_A\bra{0})\otimes\rho))\nonumber\\
      &=&\tr((Z_A\otimes \sem{P(\vth^*)}^*(O))\zeta).\label{eqn::RHSdual2}
    \end{eqnarray}
    Now RHS (Eqn$\ref{eqn::RHSdual1})$ equals RHS$(\textrm{Eqn(\ref{eqn::RHSdual2})})$ via duality, as $\calI_A\otimes P(\vth^*)$ does nothing on the ancilla, while $\sem{P(\vth^*)}^*(O)\in\weMV$ and $Z_A
    $ are compatible observables. 
\end{enumerate}
\end{proof}
\begin{lemma}\label{lem::constructingCompl}Let $\udd{S_m(\vth)}\in\detaugweirdCPV$ ($\forall m\in [0,w],$ where $ w\geq 1$). Denote:\footnote{As a reminder: we shall frequently refer to Notations in Eqns \ref{eqn::C26} and \ref{eqn::realC27} for the rest of the section.} 
\begin{eqnarray}
   &&\compl{(\udd{S_m(\vth)})}\nonumber\\ &\equiv & \{|Q_{m,0}(\vth), \cdots, Q_{m,
   {s_m}}(\vth)|\},\label{eqn::C26}\\
   &&\compl{(\parthetagenNotRelation{\udd{S_m(\vth)}})}\nonumber\\ &\equiv& \{|P_{m,0}(\vth), \cdots, P_{m,
   {t_m}}(\vth)|\}.\label{eqn::realC27}
\end{eqnarray} 
Then for arbitrary $\rho\in\weirdS$, $O\in\weMV$, we have the following computational properties regarding the observable semantics:
 \begin{eqnarray}
    &&\ObsSemMacro{O}{\parthetagenNotRelation{\udd{S_0(\vth);S_1(\vth)}}}{\rho}\nonumber\\&=&\ObsSemMacro{O}{\parthetagenNotRelation{\udd{S_0(\vth)}};\udd{S_1(\vth)}}{\rho}\nonumber\\
    &&+\ObsSemMacro{O}{\udd{S_0(\vth)};\parthetagenNotRelation{\udd{S_1(\vth)}}}{\rho},\label{eqn::C27}\\
    && \ObsSemMacro{O} {\parthetagenNotRelation{\uqifParaStandardS}}{\rho}\nonumber\\ &=& \sum_{m}\sum_{i_m }\ObsSemMacro{O}{
     P_{m,i_m}(\vth)}{M_m\rho M\dagg_m},
    \nonumber \\&&\textrm{ where }i_m\textrm{ runs through }[0,t_m]; \label{eqn::C28}
    \end{eqnarray}
    
     \begin{eqnarray}
      && \ObsSemMacro{O}{ \BPp{ \parthetagenNotRelation{\udd{S_0(\vth)\msquare S_1(\vth)}}}}{\rho}\nonumber \\ &=&\ObsSemMacro{O}{\parthetagenNotRelation{\udd{S_0(\vth)}}}{\rho} +\\
      && \ObsSemMacro{O}{\parthetagenNotRelation{\udd{S_1(\vth)}}}{\rho}\label{eqn::C29}
     \end{eqnarray}
     \begin{eqnarray}
        && {\ObsSemMacro{O}{{ \udd{S_0(\vth)\msquare S_1(\vth)}}}{\rho}}\nonumber\\&=&{\ObsSemMacro{O}{{ \udd{S_0(\vth)}}}{\rho}} +{\ObsSemMacro{O}{ { \udd{ S_1(\vth)}}}{\rho}},\label{eqn::C30}
     \end{eqnarray}
  and 
    \begin{eqnarray}
      &&\ObsSemMacro{O}{ {\qif{M[\overline{q}]=\overline{m\to \udd{S_m(\vth)}}}}}{\rho}\nonumber  \\&=&\sum_m\sum_{j_m\in[0,s_m]}\ObsSemMacro{O}{{Q_{m,j_m}(\vth)}}{\E_m\rho} \label{eqn::C31}
    \end{eqnarray}
    
\end{lemma}
\begin{proof}

We first 
make some direct application of the code transformation (hereinafter ``CT'', Fig \ref{fig:CT}) and compilation (hereinafter ``CP'', Fig \ref{fig::compileReal}) rules; some results will immediately follow from these application.
\begin{eqnarray}
&&\compl{\Big{(}\parthetagenNotRelation{\udd{S_0(\vth);S_1(\vth)}}\Big{)}}\nonumber\\ &\stackrel{\textrm{CT,Sequence}}{=}&\texttt{Compile}{(}({\udd{S_0(\vth)};\parthetagenNotRelation{\udd{S_1(\vth)}}})\ \msquare \nonumber\\&& (\parthetagenNotRelation{\udd{S_0(\vth)}};\udd{S_1(\vth)}){)}\nonumber\\
&\stackrel{\textrm{CP,Sum Components}}{=}&\compl{\Big{(}{\udd{S_0(\vth)};\parthetagenNotRelation{\udd{S_1(\vth)}}}\Big{)}}\\
&&\coprod \compl{\Big{(} {\parthetagenNotRelation{\udd{S_0(\vth)}};\udd{S_1(\vth)}}\Big{)}}\nonumber\\
&\stackrel{\textrm{CP,Sequence}}{=}&\{|Q_{0,g}(\vth);P_{1,h}(\vth)|\}_{g\in [s_0],h\in  [t_1]}\label{eqn::CPSeq1}
     \\
       &&\coprod\nonumber\\&& \{|P_{0,i}(\vth);Q_{1,j}(\vth)|\}_{i\in [t_0],j\in  [s_1]};\label{eqn::CPSeq2}
\end{eqnarray}
, where $[s_0],[s_1],[t_0],[t_1]$ stand for e.g. $[0,s_0],[0,s_1],[0,t_0],[0,t_1]$ respectively; and 
\begin{eqnarray}
&& \compl{\Big{(}\parthetagenNotRelation{\uqifParaStandardS}\Big{)}}\nonumber \\ &\stackrel{\textrm{CT,Case}}{=}& \compl{\Big{(}\udd{\ccase}\ M[
         \qvec{q}] =  \overline{\ m \to \parthetagenNotRelation{\udd{S_m(\vth)}}}~\mathbf{end}\Big{)}} \nonumber\\
          &\stackrel{\textrm{CP,Case}}{=}& \FacebookBPp{\udd{\ccase}\ M[
         \qvec{q}] =  \overline{\ m \to \parthetagenNotRelation{\udd{S_m(\vth)}}}~\mathbf{end}};
   \nonumber
\end{eqnarray}
also,
\begin{eqnarray}
&&\compl{\BPp{ \parthetagenNotRelation{S_0(\vth)\msquare S_1(\vth)}}}\\&\stackrel{\textrm{CT,Sum Component}}{=}&\compl{\BPp{\parthetagenNotRelation{S_0(\vth)}\msquare \parthetagenNotRelation{S_1(\vth)} }}\nonumber\\&&\label{eqn::C4000} \\
&\stackrel{\textrm{CP,Sum Component}}{=}&\{|P_{0,0}(\vth),\cdots, P_{0,t_0}|\}
 \\
      &&
      \coprod  \{|P_{1,0}(\vth),\cdots, P_{1,t_1}|\}\emph{. }\label{eqn::C4010}
\end{eqnarray}

As promised, Eqns \ref{eqn::C27} and \ref{eqn::C29} immediately follows from the results above together with the corresponding definitions. So is Eqn \ref{eqn::C30}, following the same lines of unfolding by CP-(ND component) laid out in Eqns \ref{eqn::C4000}$\sim$ \ref{eqn::C4010} then a direct pattern matching by definition. To obtain \ref{eqn::C28}, observe that 
\begin{eqnarray}
&&\ObsSemMacro{O} {\parthetagenNotRelation{\uqifParaStandardS}}{\rho}\nonumber \\
&=&\sum_{Q(\vth)\in\Facebook{\udd{\ccase}\ M[
         \qvec{q}] =  \overline{\ m \to \parthetagenNotRelation{\udd{S_m(\vth)}}}~\mathbf{end} } }\\&&
         \ObsSemMacro{O} {Q(\vth)}{\rho}\\
         &=&\sum_{Q(\vth)\in\Facebook{\udd{\ccase}\ M[
         \qvec{q}] =  \overline{\ m \to \parthetagenNotRelation{\udd{S_m(\vth)}}}~\mathbf{end} } }\\&&\tr((Z_A\otimes O)\sem{Q(\vth)}(\ket{0}_A\bra{0}\otimes\rho) )\nonumber\\
\\&\stackrel{A\notin \qvec{q}\subseteq\VE
}{=}&\sum_{m^*}\sum_{i_{m^*}\in[0,t_{m^*}]}\tr\Big{(}(Z_A\otimes O)\cdot \nonumber\\
&&( \sem{P_{m^*,i_{m^*}}(\vth)}((\ket{0}_A\bra{0})\otimes M_{m^*}\rho M\dagg_{m^*}))\Big{)}\label{eqn::C40}\\
&=&\sum_{m^*}\sum_{i_{m^*}}\ObsSemMacro{O}{
     P_{m^*,i_{m^*}}(\vth)}{M_{m^*}\rho M\dagg_{m^*}}\textrm{ as desired. }
     \nonumber
\end{eqnarray}

A few more words on Eqn \ref{eqn::C40}. Since $A$ is disjoint from $\VE\supseteq \qvec{q}$, for each $Q(\vth)\in \Facebook{\udd{\ccase}\ M[
         \qvec{q}] =  \overline{\ m \to \parthetagenNotRelation{\udd{S_m(\vth)}}}~\mathbf{end}}$, evolving one step per the operational semantics will lead to the configuration $\lag \sem{P_{m^*,i_{m^*}}(\vth)},((\ket{0}_A\bra{0})\otimes M_{m^*}\rho M\dagg_{m^*})\rag$ for some non-(essentially-aborting) $P_{m^*,i_{m^*}}(\vth)$. Besides, each such configuration appears exactly once 
         affording one summand, due to construction of $\Facebook{\bullet}$. On the other hand, if $P_{m^*,i_{m^*}}(\vth)$ essentially aborts, the trace computed as a summand vanishes, making no contribution to the sum, thereby maintaining the equation.   
         
Eqn \ref{eqn::C31} is proved using the same lines of logic:
\begin{eqnarray}
&&\ObsSemMacro{O}{ {\qif{M[\overline{q}]=\overline{m\to \udd{S_m(\vth)}}}}}{\rho}
\\ &=&\sum_{Q(\vth)\in\Facebook{\udd{\ccase}\ M[
         \qvec{q}] =  \overline{\ m \to {\udd{S_m(\vth)}}}~\mathbf{end} } }
         \ObsSemMacro{O}{Q(\vth)}{\rho}\nonumber\\
         &=&\sum_{Q(\vth)\in\Facebook{\udd{\ccase}\ M[
         \qvec{q}] =  \overline{\ m \to {\udd{S_m(\vth)}}}~\mathbf{end} } }\tr( O\sem{Q(\vth)}\rho )\nonumber\\
\\&\stackrel{A\notin\qvec{q}\subseteq\VE
}{=}&\sum_{m^*}\sum_{j_{m^*}\in[0,s_{m^*}]}\tr(O\sem{Q_{m^*,j_{m^*}}(\vth)} M_{m^*}\rho M\dagg_{m^*})\nonumber\\
&=&\sum_{m^*}\sum_{j_{m^*}\in[0,s_{m^*}]}\ObsSemMacro{O}{Q_{m^*,j_{m^*}}(\vth)}{\E_{m^*}\rho}
     \nonumber
\end{eqnarray}

\textrm{ as desired. }
\end{proof}

\begin{lemma}\label{lem::MoreHelperSequential} Keeping the notations in Lemma \ref{lem::constructingCompl}, we have $\forall\vth^*\in\RR^k, O\in\weMV, \rho\in\weirdS$, 
\begin{eqnarray}
    &&\left.\Big{(}\parthetagenNotRelation{\ObsSemMacro{O}{\udd{S_0(\vth);S_1(\vth)}}{\rho}}\Big{)}\right\vert_{\vth^*}\nonumber\\&=& \sum_{g\in[0,s_0]} \left.\Big{(}\parthetagenNotRelation{\ObsSemMacro {O}{\udd{S_1(\vth)}}{\sem{Q_{0,g}(\vth^*)}\rho}}\Big{)}\right\vert_{\vth^*}\label{eqn::620}\\
&&+\sum_{j\in[0,s_1]} \left.\Big{(}\parthetagenNotRelation{\ObsSemMacro{\sem{Q_{1,j}(\vth^*)}^*(O)}{\udd{S_0(\vth)}}{\rho}}\Big{)}\right\vert_{\vth^*}\nonumber\\&&\label{eqn::621}
\end{eqnarray}

\end{lemma}
\begin{proof} 
Observe:
\begin{eqnarray}
&&\textrm{LHS}(\ref{eqn::620})\nonumber\\&\stackrel{\textrm{CP},\textrm{Sequence}}{=}& \sum_{g\in[0,s_0],j\in [0,s_1]}\nonumber\\
&&\left.\Big{(}\parthetagenNotRelation{\ObsSemMacro{O}{Q_{0,g}(\vth);Q_{1,j}(\vth)}{\rho}}\Big{)}\right\vert_{\vth^*}\nonumber\label{eqn::622}\\
&\stackrel{(***), \textrm{See Below 
}}{=}&\sum_{g\in[0,s_0],j\in [0,s_1]}\nonumber\\
&&\left.\Big{(}\parthetagenNotRelation{\ObsSemMacro{O}{Q_{1,j}(\vth)}{\sem{Q_{0,g}(\vth^*)}\rho}}\Big{)}\right\vert_{\vth^*}\nonumber\\&&\label{eqn::623}\\
&&+\sum_{g\in[0,s_0],j\in [0,s_1]}\nonumber\\
&&\left.\Big{(}\parthetagenNotRelation{\ObsSemMacro{\sem{Q_{1,j}(\vth)}^*(O)}{Q_{0,g}(\vth)}{\rho}}\Big{)}\right\vert_{\vth^*}\nonumber\\&&\label{eqn::624}\\
&{=}&\sum_{g\in[0,s_0]}\nonumber\\
&&\left.\Big{(}\parthetagenNotRelation{\ObsSemMacro {O}{\udd{S_1(\vth)}}{\sem{Q_{0,g}(\vth^*)}\rho}}\Big{)}\right\vert_{\vth^*}\label{eqn::6200}\\
&&+\sum_{j\in[0,s_1]}\nonumber\\&& \left.\Big{(}\parthetagenNotRelation{\ObsSemMacro{\sem{Q_{1,j}(\vth^*)}^*(O)}{\udd{S_0(\vth)}}{\rho}}\Big{)}\right\vert_{\vth^*}\nonumber\\&&\label{eqn::62005}
\end{eqnarray}

Where the last step is direct application of definition of input-space observable semantics. For the step (***) 
claiming RHS(\ref{eqn::622})$=$ RHS(\ref{eqn::623})$+$ RHS(\ref{eqn::624}), 
observe that it boils down to the following: for arbitrary $(g^*,j^*)\in [0,s_0]\times [0,s_1]$, 
\begin{eqnarray}
\left.\Big{(}\parthetagenNotRelation{\ObsSemMacro{O}{Q_{0,g^*}(\vth);Q_{1,j^*}(\vth)}{\rho}}\Big{)}\right\vert_{\vth^*}\label{appd::622}
\end{eqnarray}
\begin{eqnarray}
&{=}&\left.\Big{(}\parthetagenNotRelation{\ObsSemMacro{O}{Q_{1,j^*}(\vth)}{\sem{Q_{0,g^*}(\vth^*)}\rho}}\Big{)}\right\vert_{\vth^*}\label{appd::623}
\end{eqnarray}
\begin{eqnarray}
&&+ \left.\Big{(}\parthetagenNotRelation{\ObsSemMacro{\sem{Q_{1,j^*}(\vth)}^*(O)}{Q_{0,g^*}(\vth)}{\rho}}\Big{)}\right\vert_{\vth^*}\label{appd::624}
\end{eqnarray}
We argue the correctness of \ref{appd::622}$\sim$\ref{appd::624} below.  
Simplifying notations, write $Q_{0,g^*}(\vth),Q_{1,j^*}(\vth)$ as $Q_0(\vth),Q_1(\vth)$ respectively, then let
 
\begin{eqnarray}
Q_0(\vth)&\equiv&\sum_{x} K_x(\vth)\circ K_x\dagg(\vth),\\
Q_1(\vth)&\equiv &\sum_{y} J_y(\vth)\circ J_y\dagg(\vth)
\end{eqnarray}
denote the kraus operator decomposition of $Q_0(\vth),Q_1(\vth)$, each with finitely many summands.\footnote{This is doable because for each $\vth^*\in\RR^k$ we have 
\begin{eqnarray}
Q_0(\vth^*)&=&\sum_{k=1} K_x(\vth^*)\circ K_x\dagg(\vth^*),\\
Q_1(\vth^*)&=&\sum_{q=1} J_y(\vth^*)\circ J_y\dagg(\vth^*).
\end{eqnarray}
Since parameterized unitaries have a parameterized-matrix representation, 
the superoperator $\calU(\vth)=U(\vth))\circ U\dagg(\vth)$ for any parameterized unitaries also have a parameterized operator representation. In our parameterized quantum program syntax only unitaries 
are parameterized, thus the kraus operators $K_x(\vth),J_y(\vth)$'s can be chosen uniformly. If one 
insists one may perform a standard structural induction to show the existence of such a parameterized Kraus operator decomposition for any $Q(\vth)\in\augweirdCPV$. For example, assuming each $S_m(\vth)$ decomposes to $\sum_{t_{m,j}=t_{m,1}}^{t_{m,j_m}}K_{t_{m,j}}(\vth)\circ K\dagg_{t_{m,j}}(\vth)$, then $\ccase\ M[\qvec{q}] =  \ m \to S_{m}
(\vth)$ decomposes into $$\sum_m  (\sum_{t_{m,j}}K_{t_{m,j}}(\vth)\circ K\dagg_{t_{m,j}}(\vth))\circ \mathcal{E}_m.$$ 
\\As in the unitary case, 
the linear Kraus operators are 
entry-wisesmooth, ensured by the entry-wise smoothness of 
matrix representations of parameterized unitaries.}

Then,
\begin{eqnarray}
&&\left.\Big{(}\parthetagenNotRelation{\ObsSemMacro{O}{Q_{g}(\vth);Q_{j}(\vth)}{\rho}}\Big{)}\right\vert_{\vth^*}\\&=&\left.\parthetaj{\tr(O\sem{Q_0(\vth);Q_1(\vth)}\rho)}\right\vert_{\vth^*}
\end{eqnarray}
\begin{eqnarray}
&=&\left.\parthetaj{\tr(O \sum_y J_y(\vth)(\sum_x K_x(\vth)\rho K\dagg_x(\vth))J\dagg_y(\vth))}\right\vert_{\vth^*}
\end{eqnarray}
\begin{eqnarray}
&=&\sum_{x,y}\left.\parthetaj{\tr(O J_y(\vth)K_x(\vth)\rho K\dagg_x(\vth)J\dagg_y(\vth))}\right\vert_{\vth^*}
\end{eqnarray}
\begin{eqnarray}
&\stackrel{(a)}{=}&\sum_y\left.\Big{[}(\tr(O \fparthetaj{J_y(\vth)}(\sum_x K_x(\vth^*)\rho K\dagg_x(\vth^*))J\dagg_y(\vth^*))\right.\nonumber
\end{eqnarray}
\begin{eqnarray}
&&\left.+ \tr(O {J_y(\vth^*)}(\sum_xK_x(\vth^*)\rho K\dagg_x(\vth^*))\fparthetaj{J\dagg_y(\vth)}))\Big{]}\right\vert_{\vth^*}\nonumber
\end{eqnarray}
\begin{eqnarray}
&&+\sum_y\Big{[}(\tr(O {J_y(\vth^*)}(\sum_x\fparthetaj{K_x(\vth)}\rho K\dagg_x(\vth^*))J\dagg_y(\vth^*))\nonumber
\end{eqnarray}
\begin{eqnarray}
&&+ \left.\tr(O {J_y(\vth^*)}(\sum_xK_x(\vth^*)\rho \fparthetaj{K\dagg_x(\vth)}){J\dagg_y(\vth^*)}))\Big{]}\right\vert_{\vth^*}\nonumber
\end{eqnarray}
\begin{eqnarray}
&\stackrel{(b)}{=}&\left.\parthetaj{\tr(O \sem{Q_1(\vth)}(
{\sem{Q_0(\vth^*)}\rho}))}\right\vert_{\vth^*}
\end{eqnarray}
\begin{eqnarray}
&& +\left. \parthetaj{\tr( {O\sem{Q_1(\vth^*)}}\sem{Q_0(\vth)}\rho)}\right\vert_{\vth^*}
\end{eqnarray}

\begin{eqnarray}
&\stackrel{\textrm{S-H Dual}}{=}&\left.\Big{(}\parthetajgenNotRelation{\ObsSemMacro{O}{Q_{1}(\vth)}{\sem{Q_{0}(\vth^*)}\rho}}\Big{)}\right\vert_{\vth^*}\label{eqn::62?3}
\end{eqnarray}
\begin{eqnarray}
&&+ \left.\Big{(}\parthetajgenNotRelation{\ObsSemMacro{ 
{\sem{Q_{1}(\vth^*)}^*(O)}}{Q_{0}(\vth)}{\rho}}\Big{)}\right\vert_{\vth^*}
\end{eqnarray}

where $\fparthetaj{\bullet}$ again denotes the entry-wise derivative, and ${(a),(b)}$ is obtained by unfolding the definition of trace, noticing that all entries of parameterized unitaries are smooth. S-H dual here is short for Schr\"{o}dinger-Heisenberg dual, as is everywhere else in this section.

\end{proof}
\subsection{Proof Details of Theorem \ref{thm::sound}: Soundness of Differentiation Logic}
\begin{proof} As stated in the main text, throughout we fix $O_A\equiv Z_A$. Unfolding by definition, it suffices to show 
\begin{eqnarray}
  && \forall O\in\weMV,\rho\in\weirdS, \ObsSemMacro{(O, Z_A)}{\parthetagenNotRelation{{\udd{S(\vth)}}}}{\rho}\nonumber\\ &= &
    \parthetagenNotRelation{\ObsSemMacro{O}
    {{\udd{S(\vth)}}}{\rho}}.\label{eqn::appd::s_diff_s}
\end{eqnarray}

Let $O\in\weMV$, $\rho\in\weirdS$ be arbitrary. We  consider each rule in \fig{reliability-sem} possibly used as the final step of the derivation.
\begin{enumerate}
    \item (Abort)$\sim$(Trivial Unitary): for $\udd{S(\vth)}\equiv \udd{\cabort[\VE]}$, $\udd{\cskip[\VE]}$ or $\udd{q:=\ket{0}}$, $\ObsSemMacro{O}{\udd{S(\vth)}}{\rho}$ is a constant, so $\parthetagenNotRelation{\ObsSemMacro{O}{\udd{S(\vth)}}{\rho}}=0$. On the other hand, $\parthetagenNotRelation{\udd{S(\vth)}}$ $\stackrel{\textrm{CT,trivialize}}{\equiv}$  $ \udd{\cabort[\VA]}$, meaning $\ObsSemMacro{O}{\parthetagenNotRelation{\udd{S(\vth)}}}{\rho}=\tr((Z_A\otimes O)\cdot\mathbf{0})=0$.
    
    In Trivial-Unitary, $\ObsSemMacro{O}{\udd{U(\vth)}}{\rho}$ doesn't depend on $\theta_j$ since $\theta_j\notin \vth$; hence $\parthetajgenNotRelation{\ObsSemMacro{O}{\udd{U(\vth)}}{\rho}}=0$, as faithfully 
    represented by the code transformation rule.
    \item (Unitary): Assume $\udd{S(\vth)}\equiv\udd{U(\theta)}$ with $U(\theta)\in \{R_\sigma(\theta),$  $R_{\sigma\otimes\sigma}(\theta)\}_{\sigma\in\{X,Y,Z\}}$. Let $\boldgreek{\sigma}$ range over $\{\sigma,\sigma\otimes\sigma\}_{\sigma\in X,Y,Z}$ (proof below works for all $6$ cases where $\boldgreek{\sigma}^2=I_{\bullet}$). Let $\entrideri U(\theta)$ denote the entry-wise derivative of $U(\theta)$, then recall from Lemma \ref{lem:UnitaryDerive} that 
    \begin{eqnarray}
    \entrideri U(\theta)&=&\frac{1}{2}U(\theta+\pi);\textrm{ and }\label{eqn::entriwise}\\
 && \parthetagenNotRelation{\ObsSemMacro{O}{{\udd{U(\theta)}}}{\rho} }\nonumber \\&=& \tr(O\cdot U(\theta)\cdot \rho\cdot (\entrideri U(\theta))\dagg)\nonumber\\&& +\tr(O\cdot \entrideri U(\theta)\cdot \rho\cdot U\dagg(\theta)). \label{eqn::uudag}\\
    &\stackrel{\ref{eqn::entriwise}}{=}&\frac{1}{2}\tr\bigg{(}O\Big{(}U(\theta)\rho U\dagg(\theta+\pi)\nonumber\\&&+U(\theta+\pi)\rho U\dagg(\theta)\Big{)}\bigg{)}\nonumber
    \end{eqnarray}
    Meanwhile,
    \begin{eqnarray}
        &&\ObsSemMacro{O}{\parthetagenNotRelation{\udd{U(\theta)}}}{\rho}\nonumber\\&\stackrel{\textrm{CT,1-qb,2-qb}}{=}&\tr((Z_A\otimes O)\sem{R\boldgreek{'}_{\boldgreek{\sigma}}(\theta)}\cdot  \nonumber\\&&((\ket{0}_A\bra{0})\otimes\rho))\label{eqn::C71}
    \end{eqnarray}
    Let us unfold the definition of $R\boldgreek{'}_{\boldgreek{\sigma}}(\theta)$:
    \begin{eqnarray}
  && \textrm{RHS}(\ref{eqn::C71}) \nonumber\\ &{=}&\frac{1}{2}\sum_{x,y\in\{0,1\}}\tr((Z_A\otimes O)\sem{C\_R_{\boldgreek\sigma}(\theta);H_A}\cdot\nonumber\\&&((\ket{x}_A\bra{y})\otimes\rho))\nonumber\\
    &{=}&\frac{1}{2}\tr\bigg{(}(Z_A\otimes O)\sem{H_A}\Big{(}\ket{0}_A\bra{0}\otimes U(\theta)\rho U\dagg(\theta) \nonumber\\
    &&+\ket{0}_A\bra{1}\otimes U(\theta)\rho U\dagg(\theta+\pi)\\
    &&+\ket{1}_A\bra{0}\otimes U(\theta+\pi)\rho U\dagg(\theta)\\
    &&+\ket{1}_A\bra{1}\otimes U(\theta+\pi)\rho U\dagg(\theta+\pi)\Big{)}\bigg{)}\\
    &=&\frac{1}{2}\tr\bigg{(}(Z_A\otimes O)\Big{(}\frac{\sum_{x,y}\ket{x}_A\bra{y}}{2}\otimes U(\theta)\rho U\dagg(\theta)\label{eqn::C75}\\
    && + \frac{\sum_x\ket{x}_A\bra{0}-\sum_x\ket{x}_A\bra{1}}{2}U(\theta)\rho U\dagg(\theta+\pi)\\
    && + \frac{\sum_x\ket{0}_A\bra{x}-\sum_x\ket{1}_A\bra{x}}{2}U(\theta+\pi)\rho U\dagg(\theta)\\
    &&+\frac{\sum_x\ket{x}_A\bra{x}-(\ket{1}_A\bra{0}+\ket{0}_A\bra{1})}{2}\cdot\nonumber\\&&U(\theta+\pi)\rho U\dagg(\theta+\pi)\Big{)}\bigg{)}\label{eqn::C78}
    \end{eqnarray}
    Let's regroup the terms of \ref{eqn::C75} $\sim$ \ref{eqn::C78}, bearing in mind that when 
    observing $Z_A\otimes O$ on the final state, the ``off-diagonal'' part is killed by $Z_A$.\footnote{
    Namely, the trace computation can be decomposed into a finite linear combination of the form $\tr((Z_A\otimes O)(\ket{b}_A\bra{b'}\otimes \bullet))$ ($b,b'\in\{0,1\}$), and this term vanishes whenever $b\neq b'$. } Hence, 
    \begin{eqnarray}
   &&\textrm{RHS}(\ref{eqn::C75} \sim \ref{eqn::C78})\\ &{=}&\frac{1}{4}\tr\Bigg{(}(Z_A\otimes O)\bigg{(}\Big{(}\ket{1}_A\bra{1}+\ket{0}_A\bra{0}\Big{)}\label{eqn::C79}\\
    &&
    \otimes \Big{(}U(\theta)\rho U\dagg(\theta)+U(\theta+\pi)\rho U\dagg(\theta+\pi)\Big{)}\\
    &&
    +\Big{(}\ket{0}_A\bra{0}-\ket{1}_A\bra{1}\Big{)}\\
    &&
    \otimes \Big{(}U(\theta)\rho U\dagg(\theta+\pi)+U(\theta+\pi)\rho U\dagg(\theta)\Big{)}\bigg{)}\Bigg{)}
    \label{eqn::C80}
    \end{eqnarray}
    Lastly, the behavior of $Z_A\otimes O$ ensures 
    \begin{equation}
       \textrm{RHS}(\ref{eqn::C79}\sim\ref{eqn::C80})  {=}\frac{1}{2}\tr\bigg{(}O\Big{(}U(\theta)\rho U\dagg(\theta+\pi)+U(\theta+\pi)\rho U\dagg(\pi)\Big{)}\bigg{)},\nonumber
    \end{equation}

   \textrm{as desired.} From this point, we again adopt the notations 
   from Lemma \ref{lem::constructingCompl} Eqns \ref{eqn::C26}, \ref{eqn::realC27}.
    \item (Sum Component)
\begin{eqnarray}
&& \ObsSemMacro{O}{ \BPp{ \parthetagenNotRelation{\udd{S_0(\vth)\msquare S_1(\vth)}}}}{\rho}\nonumber\\&\stackrel{\textrm{Lem }\ref{lem::constructingCompl}}{=}&\ObsSemMacro{O}{\parthetagenNotRelation{\udd{S_0(\vth)}}}{\rho}\\&& +\ObsSemMacro{O}{\parthetagenNotRelation{\udd{S_1(\vth)}}}{\rho}\\
 &\stackrel{\textrm{IH}}{=}& \parthetagenNotRelation{\ObsSemMacro{O}{\udd{S_0(\vth)}}{\rho} \\&&+ \ObsSemMacro{O}{\udd{S_1(\vth)}}{\rho}}\\
 &\stackrel{\textrm{Lem }\ref{lem::constructingCompl}}{=}&\parthetagenNotRelation{\ObsSemMacro{O}{\udd{S_0(\vth)}\msquare\udd{S_1(\vth)} }{\rho}},
\end{eqnarray}

\textrm{ as desired.}
\item (Sequence) Assume $\udd{S(\vth)}\equiv\udd{S_0(\vth)};\udd{S_1(\vth)} $; It suffices to show that $\forall\vth^*$,
\begin{equation}
    \ObsSemMacro{O}{\parthetagenNotRelation{\udd{S_0(\vth^*)};\udd{S_1(\vth^*)}}}{\rho} =\left.\Big{(}\parthetagenNotRelation{ \ObsSemMacro{O}{\udd{S_0(\vth);S_1(\vth)}}{\rho}}\Big{)}\right\vert_{\vth^*}\label{eqn::CompoToShow}
\end{equation}

Let us first manipulate the equation using some computational properties developped from the helper lemmas:
\begin{eqnarray} &&\textrm{LHS}(\ref{eqn::CompoToShow})\\&\stackrel{\textrm{Lemma }\ref{lem::constructingCompl}, \textrm{Eqn }\ref{eqn::C27}}{=}&\ObsSemMacro{O}{\parthetagenNotRelation{\udd{S_0(\vth^*)}};S_1(\vth^*)}{\rho}\nonumber\\
&&+\ObsSemMacro{O}{\udd{S_0(\vth^*)};\parthetagenNotRelation{S_1(\vth^*)}}{\rho}\nonumber
\end{eqnarray}
This equals, by \textrm{Lem }\ref{lem::constructingCompl}\textrm{ Eqns }\ref{eqn::CPSeq1},\ref{eqn::CPSeq2}, 
\begin{eqnarray}
&
&\sum_{g\in [0,s_0],h\in  [0,t_1]}\ObsSemMacro{O}{Q_{0,g}(\vth^*);P_{1,h}(\vth^*)}{\rho} \\&&+\sum_{{i\in [0,t_0],j\in  [0,s_1]}} \ObsSemMacro{O}{P_{0,i}(\vth^*);Q_{1,j}(\vth^*)}{\rho}
\end{eqnarray}

Via \textrm{Lem }\ref{lem::propertyAncillaObservableSemantics}, this translates to 
\begin{eqnarray}
&&\sum_{g\in [0,s_0],h\in  [0,t_1]}\ObsSemMacro{O}{P_{1,h}(\vth^*)}{\sem{Q_{0,g}(\vth^*)}\rho}\\
&&+ \sum_{{i\in [0,t_0],j\in  [0,s_1]}} \ObsSemMacro{\sem{Q_{1,j}(\vth^*)}^*(O)}{P_{0,i}(\vth^*)}{\rho}\nonumber\\
&=&\sum_{g\in s_0}\ObsSemMacro{O}{\parthetagenNotRelation{\udd{S_1(\vth^*)}}}{\sem{Q_{0,g}(\vth^*)}\rho}\label{eqn::C91}\\
&&+\sum_{j\in s_1}\ObsSemMacro{\sem{Q_{1,j}(\vth^*)}^*(O)}{\parthetagenNotRelation{\udd{S_0(\vth^*)}}}{\rho},\label{eqn::C92}
\end{eqnarray}
where the last step comes from definition of observable semantics with ancilla. Now apply the inductive hypothesis on $\udd{S_1(\vth)}$, a universal statement for all $(\rho',O')\in \weirdS\times \weMV$, to the instances $(\sem{Q_{0,g}(\vth^*)}\rho, O))$ $\ (\forall g\in [0,s_0])$, we have the first summand
\begin{eqnarray}
&&\textrm{RHS}(\ref{eqn::C91})\nonumber\\&{=}&\sum_{g} \left.\Big{(}\parthetagenNotRelation{\ObsSemMacro{O}{\udd{S_1(\vth)}}{\sem{Q_{0,g}(\vth^*)}\rho}}\Big{)}\right\vert_{\vth^*}\label{eqn::C93}
\end{eqnarray}
Likewise, apply the inductive hypothesis on $\udd{S_0(\vth)}$ 
to the instances $(\rho, \sem{Q_{1,j}(\vth^*)}^*O))\ (\forall j\in [0,s_1])$, we have the second summand
\begin{eqnarray}
&&\textrm{RHS}(\ref{eqn::C92})\nonumber\\
&{=}&\sum_{j} \left.\Big{(}\parthetagenNotRelation{\ObsSemMacro{ \sem{Q_{1,j}(\vth^*)}^*(O)}{\udd{S_0(\vth)}}{\rho}}\Big{)}\right\vert_{\vth^*}\label{eqn::C94}
\end{eqnarray}
Collecting everything, using again a techical lemma from above,
\begin{eqnarray}
&&\textrm{RHS}(\ref{eqn::C91})+\textrm{RHS}(\ref{eqn::C92})\nonumber\\&=&\textrm{RHS}(\ref{eqn::C93})+\textrm{RHS}(\ref{eqn::C94})\\
&\stackrel{\textrm{Lem }\ref{lem::MoreHelperSequential}}{=}&\textrm{RHS}(\ref{eqn::CompoToShow}),\textrm{ as desired.}
\end{eqnarray}
\item (Case) This follows again from the computational lemmas above, and applications of all inductive hypothesis. For conciseness, let us denote by $\textrm{IH}(m)\mapsto (M_m\rho M\dagg_m,O)$ the application of the $m$-th ($m\in [0,w]$) inductive hypothesis to the instance $(M_m\rho M\dagg_m,O)\in \weirdS\times \weMV$. Then, 
\begin{eqnarray}
&&\ObsSemMacro{O} {\parthetagenNotRelation{\uqifParaStandardS}}{\rho} \nonumber\\&\stackrel{\textrm{Lem }\ref{lem::constructingCompl}}{=}&\sum_{m}\sum_{i_m\in [0,t_m]}\ObsSemMacro{O}{
     P_{m,i_m}(\vth)}{M_m\rho M\dagg_m} 
     \\
     &=&\sum_m \ObsSemMacro{O}{\udd{\parthetagenNotRelation{S_m(\vth)}}}{M_m\rho M\dagg_m}
     \end{eqnarray}
     Applying the inductive hypothesis $\textrm{IH}(m)\mapsto (M_m\rho M\dagg_m,O),$ $\forall m$, we get that the above equals
     \begin{eqnarray}
     &&\sum_m\parthetagenNotRelation{\ObsSemMacro{O}{\udd{S_m(\vth)}}{(M_m\rho M\dagg_m)\rag}}\nonumber\\
     &=&\parthetagenNotRelation{\sum_m\sum_{j_m\in[0,s_m]}\ObsSemMacro{O}{Q_{m,j_m}(\vth)}{\E_m\rho\rag}}\nonumber\\
     &\stackrel{\textrm{Lem }\ref{lem::constructingCompl}}{=}&\parthetagenNotRelation{\ObsSemMacro{O}{\qif{M[\overline{q}]=\overline{m\to \udd{S_m(\vth)}}}}{\rho}}\nonumber
\end{eqnarray}

\textrm{, as desired.}
\item (While$^{(T)}$) This is verified by successive application of (Case) rule and (Sequence) rule $T$ times.

\end{enumerate}
\end{proof}
\section{Detailed Proofs from Section \ref{sec::resource}}
\subsection{Proof for Proposition \ref{prop::NOjisSize}: Upper Bound of $\NumNonAbort{\parthetajgenNotRelation{P(\vth)}}$ }\label{appd::NOjisSize}
\begin{proof}
\begin{enumerate}
    \item If $P(\vth)\equiv\cabort[\VE]|\cskip[\VE] | q:=\ket{0}$ ($q\in\VE$), then $\ResCount{j}{P(\vth)}=0=\NumNonAbort{\parthetajgenNotRelation{P(\vth)}}$ by definition. Similarly, if $P(\vth)\equiv U(\vth)$ trivially used $\theta_j$, no non-aborting programs exists in $\compl{(\parthetagenNotRelation{P(\vth)})}$, yielding $\NumNonAbort{\parthetajgenNotRelation{P(\vth)}}=0$; otherwise, $\NumNonAbort{\parthetajgenNotRelation{P(\vth)}}$ consists of a single $R'_{\boldgreek{\sigma}}(\vth)$, so both numbers are $1$.
    \item Assume $\NumNonAbort{\parthetajgenNotRelation{P_b(\vth)}} \leq \ResCount{j}{P_b(\vth)}$ for each $b\in\{1,2\}$ and $P(\vth)\equiv P_1(\vth);P_2(\vth)$. Then (denoting as above Inductive Hypothesis as ``IH''):
    \begin{eqnarray}
&&    \NumNonAbort{\parthetajgenNotRelation{P(\vth)}}\nonumber\\&\stackrel{(*)}{=}&\NumNonAbort{({P_1(\vth)};\parthetajgenNotRelation{P_2(\vth)})} \nonumber\\
&&+ \NumNonAbort{(\parthetajgenNotRelation{P_1(\vth)};P_2(\vth))}\\
    &\leq&\NumNonAbort{(\parthetajgenNotRelation{P_2(\vth)})}\nonumber\\&& + \NumNonAbort{(\parthetajgenNotRelation{P_1(\vth)})}\label{eqn::72}\\
    &\stackrel{\textrm{IH}}{\leq}& \ResCount{j}{P_1(\vth)} + \ResCount{j}{P_2(\vth)}\\
    &=&\ResCount{j}{P(\vth)}.
    \end{eqnarray}
    
    where step $(*)$ is via rules \textrm{CT,CP-(ND-Component)}.
   \item  Assume $\NumNonAbort{\parthetajgenNotRelation{P_m(\vth)}} \leq \ResCount{j}{P_m(\vth)}$ for each $m\in [0,w]$ and $P(\vth)\equiv \qifStandardPara$. Then by the CT and CP rules of $\ccase$, $\parthetajgenNotRelation{P(\vth)}$ compiles to 
   $\max_m 
   \NumNonAbort{\parthetajgenNotRelation{P_{m}(\vth)}}$ 
   programs. Assume $m^*\in [0,w]$ is s.t. $\NumNonAbort{\parthetajgenNotRelation{P_{m^*}(\vth)}}$ attains that maximum. Then, 
    \begin{eqnarray}
    \NumNonAbort{\parthetajgenNotRelation{P(\vth)}}&=&\NumNonAbort{\parthetajgenNotRelation{P_{m^*}(\vth)}}
    \\
    &\stackrel{\textrm{IH}}{\leq}& \ResCount{m^*}{P_{m^*}(\vth)}\\
    &\leq &\max_{m}\ResCount{m}{P_{m}(\vth)}\\
    &=& \ResCount{j}{P(\vth)}
    \end{eqnarray}
    
    ,\textrm{ as desired. }
  \item If $P(\vth)\equiv \qwhileTParaStandard$, it suffices to show $\NumNonAbort{\parthetagenNotRelation{P(\vth)}}\leq T\cdot \NumNonAbort{\parthetagenNotRelation{P_1(\vth)}}$. This is true for $T=1$ as $\parthetagenNotRelation{\rwhile^{(1)}}$ essentially aborts. 
  By induction if this is true for $T$, then for $T+1$ we have 
  \begin{eqnarray}
     &&\NumNonAbort{ \parthetagenNotRelation{\mathbf{while}^{(T+1)}~M[\overline{q}]=1~\mathbf{do}~P_1(\vth)~\mathbf{done}}}\\&\leq&\NumNonAbort{\parthetagenNotRelation{P_1(\vth);\mathbf{while}^{(T)}~M[\overline{q}]=1~\mathbf{do}~P_1(\vth)~\mathbf{done}}}\nonumber\\
     &\stackrel{\textrm{Eqn}\ \ref{eqn::72}}{\leq}&\NumNonAbort{\parthetagenNotRelation{P_1(\vth)}}\nonumber\\&&+\NumNonAbort{\parthetagenNotRelation{\mathbf{while}^{(T)}~M[\overline{q}]=1~\mathbf{do}~P_1(\vth)~\mathbf{done}}}\nonumber\\
     &\leq& (T+1)\NumNonAbort{\parthetagenNotRelation{P_1(\vth)}},\textrm{ as desired. }
  \end{eqnarray} 
\end{enumerate}
\end{proof}

\section{Evaluation Details}\label{appd::evalPen}
\subsection{Training VQC instances with controls}

For any parameter $\alpha$, we exhibit the final results of $\texttt{Compile} $  ${\newline ( \dfrac{\partial P_1(\Theta,\Phi)}{\partial\alpha })}$ and $\compl{( \dfrac{\partial P_2(\Theta,\Phi, \Psi)}{\partial\alpha })}$, with $P_1(\Theta,\Phi)$ and $P_2(\Theta,\Phi,\Psi)$ defined in the main text: 

\begin{enumerate}
  \item If $\alpha \in \Theta$: without loss of generality assume $\alpha = \theta_1$ (other situations are completely analogous). We denote:
  \begin{equation*}
      Q'(\Theta)\equiv R^{\boldgreek{'}}_X(\theta_1)\ket{A,q_1};\cdots;\\ R_Z(\theta_{12})\ket{q_{12}}
  \end{equation*}
  
  then 
  \begin{equation*}
    {\left. \begin{array}{rl} \compl{\newline ( \dfrac{\partial P_1(\Theta,\Phi)}{\partial\alpha })}
    \equiv &\{|\ Q'(\Theta); Q(\Phi)\ |\},\\ 
   \compl{ ( \dfrac{\partial P_2(\Theta,\Phi,\Psi)}{\partial\alpha })} &\equiv \{|\  Q'(\Theta);\\
     {\mathbf{case}}\ M[{q_1}]=
    0\to  & \enskip Q(\Phi)  
    \\
    1\to  &\enskip Q(\Psi) 
    \ |\}.
  \end{array}  \right. }
  \end{equation*}
  \item If $\alpha \in \Phi$: without loss of generality assume $\alpha = \phi_1$ (other situations are completely analogous). We denote:
  
  \begin{equation*}
      Q'(\Phi)\equiv R^{\boldgreek{'}}_X(\phi_1)\ket{A,q_1};\cdots;\\ R_Z(\phi_{12})\ket{q_{12}}
  \end{equation*}
  \begin{equation*}
    {\left. \begin{array}{rl} \compl{\newline ( \dfrac{\partial P_1(\Theta,\Phi)}{\partial\alpha })}
    \equiv &\{|\ Q(\Theta);Q'(\Phi)\ |\},\\ 
   \compl{ ( \dfrac{\partial P_2(\Theta,\Phi,\Psi)}{\partial\alpha })} &\equiv \{|\  
   Q(\Theta);\\
     {\mathbf{case}}\ M[{q_1}]=
    0\to  & \enskip Q'(\Phi) 
    \\
    1\to  &\enskip Q(\Psi) 
    \ |\}.
  \end{array}  \right. }
  \end{equation*}
  
 \item If $\alpha \in \Psi$: without loss of generality assume $\alpha = \psi_1$ (other situations are completely analogous). We denote:
  
  \begin{equation*}
      Q'(\Psi)\equiv R^{\boldgreek{'}}_X(\psi_1)\ket{A,q_1};\cdots;\\ R_Z(\psi_{12})\ket{q_{12}}
  \end{equation*}
  \begin{equation*}
    {\left. \begin{array}{rl} \compl{\newline ( \dfrac{\partial P_1(\Theta,\Phi)}{\partial\alpha })}
    \equiv &\{|\ 
    \cabort[A,q_1,\cdots, q_{12}]
    \ |\},\\ 
   \compl{ ( \dfrac{\partial P_2(\Theta,\Phi,\Psi)}{\partial\alpha })} &\equiv \{|\  
   Q(\Theta);\\
     {\mathbf{case}}\ M[{q_1}]=
    0\to  & \enskip Q(\Phi)
    \\
    1\to  &\enskip Q'(\Psi) 
    \ |\}.
  \end{array}  \right. }
  \end{equation*}
\end{enumerate}

\subsection{Benchmark testing on representative VQCs}

\begin{table}[t]
    \centering
\begin{tabular}{ |p{1.2cm}|%p{1.5cm}|
p{0.81cm}|p{0.91cm}|p{0.84cm}|p{0.87cm}|p{0.87cm}|p{0.87cm}|  }
 \hline
 $P(\vth)$     &
 $\ResCount{}{\cdot}$ &
 $\NumNonAbort{\parthetagenNotRelation{}}$
 &   $\numgate$ 
 & 
 $\numline$ &
 $\numlayer$
 &$\numqb$ \\
\hline
$\QNN_{S,b}$ &1	&        1	&20	&24	&	1&	4\\
\hline
$\QNN_{S,s}$ &5	 &       5	&20	&24	&	1&	4\\
\hline
$\QNN_{S,i}$ &10	&	10&	60&	67	&	2&	4\\
\hline
$\QNN_{S,w}$& 15	&	10&	60&	66	&	3&	4\\
\hline
$\QNN_{M,i}$& 24&	24	&165	&189	&	3	&18\\
\hline
$\QNN_{M,w}$ &56&	24&	231	&121	&	5&	18\\
\hline
$\QNN_{L,i}$&48 &	48&	363	&414	&	6&	36\\
\hline
$\QNN_{L,w}$& 504& 	48	& 2079& 	244	& 	33	&36\\
\hline 
$\VQE_{S,b}$ & 1&	        1	& 14 &	16 &		1&	2\\
\hline 
$\VQE_{S,s}$ & 2 &		2& 	14	& 16 & 		1& 	2\\
\hline 
$\VQE_{S,i}$& 4&		4&	28 &	38	&	2&	2\\
\hline 
$\VQE_{S,w}$& 6	&	4& 	42& 	32	&	3&	2\\
\hline 
$\VQE_{M,i}$ & 15 & 	15	& 224	& 241 &		3	& 12\\
\hline 
$\VQE_{M,w}$ & 35 & 	15 &	224 &	112	&	5	& 12\\
\hline 
$\VQE_{L,i}$ & 40	 & 40	& 576&	628	&	5	&40\\
\hline 
$\VQE_{L,w}$ & 248	& 40 & 1984	& 368	&	17	& 40\\
\hline 
$\QAOA_{S,b}$ & 1	&	1	& 12	& 15	&	1	& 3\\
\hline 
$\QAOA_{S,s}$ & 3	&	3&	12&	15	&	1	& 3	\\
\hline 
$\QAOA_{S,i}$ & 6	&	6	& 36	& 41	&	2	& 3\\
\hline 
$\QAOA_{S,w}$ & 9	&	6	& 36 &	29	&	3 &      3\\
\hline 
$\QAOA_{M,i}$ & 18& 	18& 	120&	142	&	3&	18\\
\hline 
$\QAOA_{M,w}$ & 42	& 18& 	168	& 94	&	5	& 18\\
\hline 
$\QAOA_{L,i}$ & 36	& 36 &	264	 & 315	& 	6 &	36\\
\hline 
$\QAOA_{L,w}$&  378	 & 36	& 1512	& 190	&	33 &	36\\
\hline 
\end{tabular}
\caption{Output on Test Examples. Note that: (1) $\{S,M,L\}$ stands for ``small, medium, large''; $\{b,s,i,w\}$ stands for ``basic, shared, if, while''; $\numline$ is the number of lines for input programs. (2) gate count and layer count for $T$-bounded while is done by multiplying the corresponding count for loop body with $T$. (3) for $*_{*, w}$ we have $\NumNonAbort{\parthetagenNotRelation{P(\vth)}}<\ResCount{}{P(\vth)}$ because differentiating the unrolled bounded while generates essentially aborting programs which are optimized out of the final multiset.}
    \label{tab:my_tabel}
\end{table}

In this section we introduce three examples from real-world quantum machine learning 
and quantum approximation algorithms, all of which are very promising candidates for 
implementation on near-term quantum devices. Without loss of generality, throughout this section $\theta$ denotes $\theta_1$.

Our first case, $\texttt{QNN}_{*}$, starts from a slightly simplified case from an actual quantum neural-network that has been implemented on ion-trap quantum machines~\cite{zhu2018training}. $\QAOA_{*}$ is a quantum algorithm that produces approximate solutions for combinatorial optimization problems, which is regarded as one of the most promising candidates for demonstrating quantum supremacy~\cite{QAOA,wang2018introductionQAOA}. Quantum eigensolvers, crucial to quantum phase estimation algorithms, usually requires fully coherent evolution. In \cite{NC-VQE} Peruzzo et al. introduced an alternative approach that relaxes the fully coherent evolution requirement combined with a
new approach to state preparation based on ans\"{atze} and classical optimization~\cite{NC-VQE,IBM-QE,PG18}, namely our third example $\VQE_{*}$.  The control flow creating multi-layer structures for the three families of examples are very similar, and these algorithms mainly differ from each other in their basic ``rotation-entanglement'' blocks. Therefore I will introduce the basic blocks for all three families, then use $\QNN_{*}$ as an example to illustrate how the control flow (if, bounded while) works.

The basic ``rotate-entangle'' building block for 
 $\QNN$ consists of a rotation stage and an entanglement stage -- we consider these two stages together a single layer:
in the rotation stage, one performs parameterized $Z$ rotations followed by parameterized $X$ rotations and then again parameterized $Z$ rotations on the first $4$ (small scale) or $6$ qubits (medium or large scale); in the entanglement stage, one performs the parameterized $X\otimes X$ rotation on all pairs from the first $4$ or $6$ qubits. See Figure \ref{fig::chennianlaodaimaaa4}. 

Basic block for $\VQE$ consists of three stages: the first stage is parameterized $X$ followed by parameterized $Z$; the second stage uses $H$ and CNOT to entangle; the third stage performs parameterized $Z,X,Z$ in that order. Basic block for $\QAOA$, on the other hand, entangles using H and CNOT in the first stage, and then performs parameterized $X$ rotations on the second stage. 

The more interesting part lies in building multiple layers 
using control flow, which we will 
explain using $\QNN_{*}$. The small scale if-controlled $\QNN$ (denoted as $\QNN_{S,i}$) is two-layered. Let $B$ denote the basic rotate-entangle block for $\QNN$ explained as above; $B', B''$ be two similarly-structured rotate-entangle blocks using different parameter-qubit combinations. 
$\QNN_{S,i}$ performs $B$ as the first layer, then measures on the first qubit, and performs $B'$ or $B''$ as the second layer dependent on the measurement output. For medium and large scale if-controlled $\QNN$ (i.e. $\QNN_{M,i}$ and $\QNN_{L,i}$) we have larger rotate-entangle blocks,  
various parameter-qubit combinations (hence more of the $B',B''$'s involving different parameter-qubit combinations), and more layers of measurement-based control.

(Bounded) while-controled $\QNN$ works similarly: take $\QNN_{S,w}$ as an example, it performs the rotate-entangle block $B$ as the first layer, then measure the first qubit; if it outputs $0$ we halt; otherwise we perform some $B'$, measure qubit $q1$ again, halts if outputs $0$, performs $B'$ the third time and aborts otherwise
. Note that this is just a verbose way of saying ``we wrapped $B'$ with a $2$-bounded while-loop''! And similarly, we build more layers on larger systems ($\QNN_{M,w}$, $\QNN_{L,w}$) using larger rotate-entangle blocks and more layers of bounded-while loops. One should note how bounded while is a succinct way to represent the circuits: as shown in Table \ref{tab:my_tabel}, in $\QNN_{L,w}$ we managed to represent $2079$ unitary gates in just $244$ lines of code. 
\begin{figure}[t]
\includegraphics[trim=0 0 0 0,clip,height=0.12
\textheight,width=
0.45\textwidth]{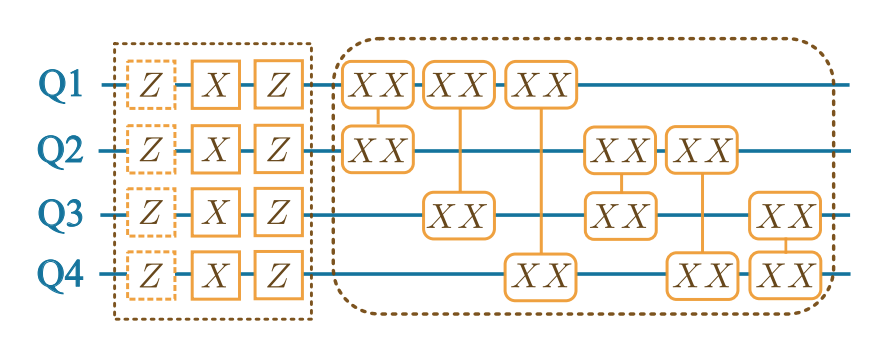}\caption{Circuit 
representation of the basic building block of QNN. Figure credit: ~\cite{zhu2018training}}\label{fig::chennianlaodaimaaa4}
\end{figure}

We auto-differentiated the three families of quantum programs using our code transformation (hereinafter ``CT'') and code compilation (hereinafter "CP") rules. 
As shown in Table \ref{tab:my_tabel}, the computation outputs the desired multi-set of derivative programs, and 
the number of non-aborting programs ($\NumNonAbort{\parthetagenNotRelation{P(\vth)}}$) agrees with the upper bound described in Proposition $\ref{prop::NOjisSize}$. It should be noted that Table \ref{tab:my_tabel} indicates our auto-differentiation scheme works well for variously sized input programs, be the size measured by code length, gate count, layer count or qubit count, to name a few.

\fi

\end{document}